\documentclass[journal,onecolumn,draftclsnofoot]{IEEEtran}
\usepackage{multirow,setspace,verbatim,amsfonts,graphicx,amsmath,amsthm,amsbsy,amssymb,epsfig,url}
\usepackage{amsthm}
\usepackage{gensymb}
\usepackage{graphicx}

\usepackage{cite}
\usepackage{makecell}

\usepackage{amsfonts}
\usepackage{amsmath,bm}
\usepackage{amssymb}
\usepackage{amsthm}
\usepackage{tikz,graphicx}
\usepackage{mathrsfs}
\usepackage{algorithmic}
\usepackage{algorithm}
\usepackage{array,multirow}
\usepackage{color}

\usepackage{mathdots}
\usepackage{MnSymbol}%
\usepackage{mathtools}
\usepackage{cleveref}
\crefname{equation}{}{}

\newtheorem{theorem}{Theorem}

\newtheorem{lemma}[theorem]{Lemma}

\newtheorem{remark}[theorem]{Remark}

{}

\newfont{\smallmathfont}{cmmib10 at 8pt}

\newcommand{\ab}{{\mathbf a}}

\newcommand{\eb}{{\mathbf e}}

\newcommand{\gb}{{\mathbf g}}
\newcommand{\hb}{{\mathbf h}}

\newcommand{\lb}{{\mathbf l}}

\newcommand{\ub}{{\mathbf u}}

\newcommand{\xb}{{\mathbf x}}
\newcommand{\yb}{{\mathbf y}}
\newcommand{\zb}{{\mathbf z}}

\newcommand{\Hc}{\mathcal{H}}
\newcommand{\Ac}{\mathcal{A}}

\newcommand{\Lc}{\mathcal{L}}
\newcommand{\Tc}{\mathcal{T}}

\newcommand{\Yc}{\mathcal{Y}}
\newcommand{\Vc}{\mathcal{V}}

\newcommand{\Cd}{{\mathbb C}}

\newcommand{\Pc}{{{\mathcal P}}}

\newcommand{\Zd}{\mathbb{Z}}
\newcommand{\hank}{\mathscr{H}}

\newcommand{\conv}{\mathscr{C}}

\newcommand{\zerob}{\mathbf{0}}

\newcommand{\Null}{\textsc{Nul}}
\newcommand{\Ran}{\textsc{Ran}}
\newcommand{\rank}{\textsc{rank}}
\newcommand{\id}{\mathrm{id}}

\newcommand{\beq}{\begin{equation}}
\newcommand{\eeq}{\end{equation}}
\newcommand{\beqa}{\begin{eqnarray}}
\newcommand{\eeqa}{\end{eqnarray}}

\newcommand{\Fc}{{\mathcal F}}
\newcommand{\norm}[1]{\left\lVert #1 \right\rVert}

\usepackage{chngcntr}

\begin{document}

\title{Compressive Sampling using  Annihilating Filter-based Low-Rank Interpolation}
\author{ Jong Chul Ye {\em Senior Member, IEEE}, Jong Min Kim,   Kyong Hwan Jin, and Kiryung Lee }

 \maketitle

\begin{abstract}
\setstretch{1}
While the recent theory of compressed sensing provides an opportunity to overcome the Nyquist limit in recovering sparse signals, a solution approach usually  takes a form of inverse problem of the unknown signal,
 which  is crucially dependent on specific signal representation.
In this paper, we propose a drastically different two-step Fourier compressive sampling framework in continuous domain  that can be implemented as a measurement domain interpolation, after which a signal reconstruction can be done using classical analytic reconstruction methods.
The main idea is  originated from the fundamental duality between the sparsity in the primary space and the low-rankness of a structured matrix in the spectral domain, which shows that a low-rank interpolator  in the spectral domain  can enjoy all the benefit of sparse recovery     with performance guarantees.
Most notably,  
 the proposed low-rank interpolation approach
can be regarded as a  generalization of  recent  spectral compressed sensing to recover
  large class of finite rate of innovations (FRI) signals at near optimal sampling rate.
Moreover, for the case of cardinal representation,  we can show that  the proposed low-rank interpolation will benefit from inherent regularization and the optimal
incoherence parameter.
Using  the powerful dual certificates and golfing scheme, we show that the new framework still achieves the near-optimal sampling rate for general class of  FRI signal recovery,
and    the sampling rate can be further reduced   for the class of
cardinal splines.
Numerical results using various type of FRI signals confirmed that the proposed low-rank interpolation approach has significant better phase transition than the conventional CS approaches.
\end{abstract}


\noindent Correspondence to:\\
Jong Chul Ye, Ph.D\\
Professor\\
Department of Bio and Brain Engineering\\
Korea Adv. Inst. of Science and Technology (KAIST)\\
373-1 Guseong-Dong, Yuseong-Gu, Daejon 305-701, Korea\\
Tel: +82-42-350-4320\\
Email: jong.ye@kaist.ac.kr



\clearpage

\section{Introduction}

{C}{ompressed} sensing or compressive sampling (CS) theory \cite{Do06,CaRoTa06,CaTa05}
addresses the accurate recovery of unknown sparse signals from
underdetermined linear measurements. 
%
%
%
In particular, a Fourier CS problem, which  recovers unknown signals from sub-sampled Fourier measurements,  has many
important applications in imaging applications such as magnetic resonance imaging (MRI),  X-ray computed tomography (CT), optics, and so on.
Moreover,  this problem  is closely related to the classical harmonic retrieval
problem that
computes the amplitudes and  frequencies  at {off the grid} locations of 
a superposition of complex sinusoids from their {consecutive} or { bunched} Fourier samples.
Harmonic retrieval can be solved by various methods including Prony's method\cite{prony1795essai}, and matrix pencil algorithm \cite{hua1990matrix}.
These methods were proven to succeed at the minimal sample rate in the noiseless case, 
because it satisfies an algebraic condition called the full spark (or full Kruskal rank) condition \cite{kruskal1977three} that
 guarantees the unique identification of the unknown signal. 
Typically, when operating at the critical sample rate, these method are not robust to perturbations in the measurements
due to the large condition number.

Accordingly,  to facilitate robust reconstruction of off the grid spectral components,
CS algorithms from non-consecutively sub-sampled  Fourier measurements are required.
The scheme is called {\em spectral compressed sensing}, which is also known as {\em compressed sensing off the grid}, when the underlying signal is composed of Diracs.
Indeed, 
 this has been developed with a close link to the recent  { super-resolution} researches  \cite{duarte2013spectral,candes2014towards,tang2013compressed,chen2014robust}. 
For example, Candes and/or Fernandez-Granda \cite{candes2014towards,candes2013super} 
showed that if the minimum distance of the Diracs is bigger than $2/f_c$ where $f_c$ denotes the cut-off frequency of the measured spectrum,
then a simple convex optimization can solve the locations of the Diracs.
Under the same minimum distance condition, Tang \cite{tang2013compressed,tang2015near}
proposed an atomic norm minimization approach for the recovery of Diracs from random spatial  and Fourier samples, respectively.
Unlike these direct signal recovery methods,
Chen and Chi  \cite{chen2014robust} proposed a two-step approach consisting of interpolation followed by a matrix pencil algorithms.
In addition, they provided performance guarantees at near optimal sample complexity (up to a logarithmic factor). 
One of the main limitations of these spectral compressed sensing approaches is, however,  that the unknown signal  is restricted to a stream of  Diracs.
The approach by Chen and Chi \cite{chen2014robust} is indeed a special case  of the proposed approach, but
they did not realize its potential  for recovering much wider  class of signals.

Note that the stream of Diracs is a special  instance of a signal model called {\em the  signals with the finite rate of innovation (FRI)}  \cite{vetterli2002sampling,dragotti2007sampling,maravic2005sampling}.
Originally  proposed by Vetterli et al \cite{vetterli2002sampling},
the class of FRI signals  includes a stream of Diracs, a stream of differentiated Diracs, non-uniform splines, 
piecewise smooth polynomials, and so on.
Vetterli et al \cite{vetterli2002sampling,dragotti2007sampling,maravic2005sampling} proposed {\em time-domain} sampling schemes of  these FRI signals that operate at the rate of innovation 
with a provable algebraic guarantee in the noise-free scenario.
Their reconstruction scheme estimates an annihilating filter that cancels the 
Fourier series coefficients of a FRI signal at consecutive low-frequencies. 
However, due to the time domain data acquisition, the equivalent Fourier
 domain measurements 
are restricted to 
a bunched sampling pattern similar to the classical harmonic retrieval problems. 

Therefore, one of the main aims of this paper is to  generalize the  scheme by Verterli et al \cite{vetterli2002sampling,dragotti2007sampling,maravic2005sampling} to 
address 
Fourier CS problems  that recover general class of FRI signals from  irregularly subsampled   Fourier measurements.
Notably, 
we prove that the only required change is 
an additional Fourier domain interpolation step
that estimates missing Fourier measurements.
More specifically,  for general class FRI signals introduced in  \cite{vetterli2002sampling,dragotti2007sampling,maravic2005sampling}, 
we show that there always exists a low-rank Hankel structured matrix associated with the corresponding annihilating filter.
Accordingly, their missing spectral elements can be interpolated using a low-rank Hankel matrix completion algorithm.
Once a set of Fourier measurements at consecutive frequencies are interpolated, 
a FRI signal can be reconstructed using conventional methods 
including Prony's method and matrix pencil algorithms  as done in  \cite{vetterli2002sampling,dragotti2007sampling,maravic2005sampling}.
Most notably, we show that the proposed Fourier CS of FRI signals 
operates at a near optimal rate (up to a logarithmic factor) with provable performance guarantee. 
Additionally, thanks to the inherent redundancies introduced by CS sampling scheme, 
the subsequent step of retrieving a FRI signal becomes much more stable.

While a similar  low-rank Hankel matrix completion approach was used by Chen and Chi \cite{chen2014robust},   there are several important differences.
First, the low-rankness of the Hankel matrix in  \cite{chen2014robust} was shown based on the standard Vandermonde decomposition, which is true only  when the
underlying FRI signal is a stream of Diracs.  
Accordingly, in case  of  differentiated Diracs,  
 the theoretical tools in \cite{chen2014robust} cannot be used.
Second, when the underlying signal 
 can be converted to a stream of Diracs or differentiated Diracs by applying 
  a linear transform 
that acts as a diagonal operator (i.e., element-wise multiplication) in the Fourier domain, 
   we can still construct
a low rank Hankel matrix from the {\em weighted} Fourier measurements, whose weights are determined by the spectrum of the linear operator. 
For example, a total variation (TV)-sparse signal is a stream of Diracs after the differentiation, and piecewise smooth polynomials becomes a stream of differentiated Diracs by applying a differential operator.
Finally,
the advantage of the proposed approach becomes more evident when we  model the unknown signal using cardinal L-splines \cite{unser2010introduction}.
In cardinal L-splines,  the discontinuities  occur only on an integer grid, 
 which is a reasonable model to acquire signals of high but finite resolution. 
Then,  we can  show that the discretization using cardinal splines makes the reconstruction significantly more stable in the existence of noise to measurements,
and the logarithmic factor as well as the incoherence parameter for the performance guarantees can be further improved.

\begin{figure}[!htb]
    \centering
 	\begin{minipage}[b]{0.5\linewidth}
	\centerline{\includegraphics[width=11cm]{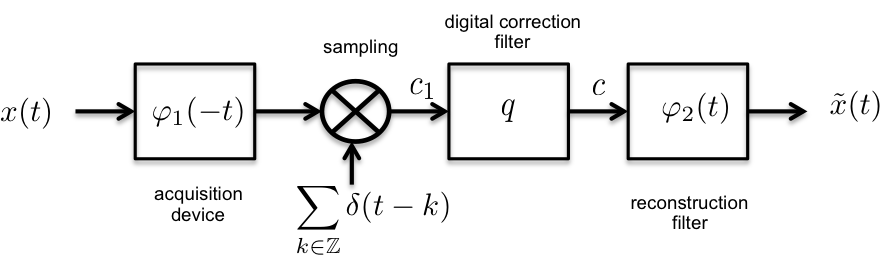}}
	\centerline{\mbox{(a)}}
		   \vspace{0.3cm}
	\centerline{\includegraphics[width=11cm]{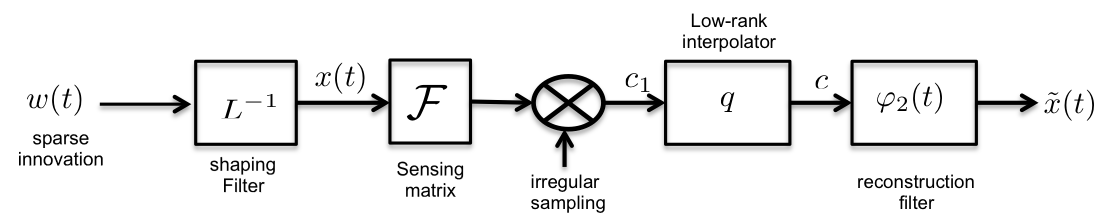}}
		\centerline{\mbox{(b)}}
	\end{minipage}
	\caption{Comparison with various sampling schemes. (a) Generalized sampling  \cite{unser2000sampling,unser1994general}: here,  a continuous input signal is filtered through a acquisition device, after which
	uniform sampling is performed. The goal of sampling is to impose consistency condition such that  if the reconstructed signal is used as an input
	to the acquisition device, it can generate the same discrete sequence
	$\{c_1\}$. This can be taken care of by the digital correction filter $q$.
		(b) Proposed sampling scheme:  Here, CS step is replaced by a discrete low-rank interpolator, and the final reconstruction
	is obtained using the reconstruction filter from fully sampled data.
		}
	\label{fig:sampling}
\end{figure}

It is important to note  that the  proposed low-rank interpolation approach is  different from classical compressed sensing approaches which  regard a sampling problem  as an inverse problem and whose focus is to directly recover the unknown  signal.
Rather,  the proposed approach is more closely related to the classical sampling theory, where
signal sampling step is  decoupled from  a signal recovery algorithm.
For example, in the sampling theory for signals in the shift-invariant spaces \cite{unser2000sampling,unser1994general}, the nature of the signal sampling  can
be fully taken care of  as a {\em digital correction filter}, after which signal recovery is performed by convolution with a reconstruction filter (see Fig.~\ref{fig:sampling}(a)).
Similarly, by introducing a {\em low rank interpolator},  the proposed  scheme in Fig.~\ref{fig:sampling}(b)  fully decouples the signal recovery from sampling step as a separate layer that can be optimized independently.
This is because the same low-rank interpolator will successfully complete missing measurements,
regardless of whether the unknown signal is either a stream of Diracs or a stream of differentiated Diracs.
In the subsquent step, analytic reconstruction methods such as
Prony's method and matrix pencil algorithms can identify the signal model  as done in  \cite{vetterli2002sampling,dragotti2007sampling,maravic2005sampling}. 

The  proposed two-layer approach   composed of Fourier domain interpolation
and the analytic reconstruction is very useful in real-world applications, because 
the low-rank interpolator can be added as a digital correction filter to existing systems, 
where the second step is already implemented. 
Moreover, in many biomedical imaging problems such as magnetic resonance imaging (MRI) or X-ray computed tomography (CT),
an accurate interpolation to fully sampled Fourier data gives an important advantage of utilising fully established  mathematical theory of analytic reconstruction.
In addition,  classical preprocessing techniques for artifact removal have been developed by assuming fully sampled measurements, so these steps can be
readily combined with the proposed low-rank interpolation approaches.
The superior advantages of the proposed  scheme  have been demonstrated in various biomedical imaging  and image processing applications such as 
compressed sensing MRI \cite{jin2015general,Lee2015MRPM},  MR artifact correction\cite{Lee2015fMRI,Jin2016MRM},  image inpainting \cite{jin2015annihilating},  super-resolution microscopy \cite{min2015fast}, image denoising \cite{jin2015sparse+}, and so on, 
which clearly confirm the practicality of the new theory.

%

Nonetheless, it is remarkable that the proposed two-layer approach using a low-rank interpolation achieves 
near optimal sample rate 
while universally applying to different signal models of the same order (e.g., stream of Diracs and stream of differentiated Diracs). 
Moreover, it may look mysterious that no explicit form of the minimum separation distance as required in
Fernandez-Granda \cite{candes2014towards,candes2013super} and Tang \cite{tang2013compressed,tang2015near} is not shown in the performance guarantee.
However, the proposed method is not free of limitations. Specifically, we will show that the incoherence parameter in our performance guarantees is 
dependent upon the type of unknown signals as well as the minimum separation between the successive spikes.
The similarity and differences of our results from
the existing theory  \cite{candes2014towards,candes2013super,tang2013compressed,tang2015near}
and the origin of the differences will be also discussed.

This paper consists of followings. Section~\ref{sec:main} first  discusses the main results
that relate an annihilating filter  and a low-rank Hankel structured matrix, and provides the performance guarantees 
of low-rank structured matrix completion, 
 which will be used throughout the paper.
 Section~\ref{sec:fri} then discusses the proposed low-rank interpolation theory for recovery of FRI signals,
 which is followed by the low rank interpolation for the case of cardinal L-splines in Section~\ref{sec:cardinal}.
 Section~\ref{sec:algorithm} explains algorithmic implementation.
Numerical results are then provided in Section~\ref{sec:result}, which is followed by conclusion in Section~\ref{sec:conclusion}, respectively.

\section{Main Results}
\label{sec:main}

\subsection{Notations }

A Hankel structured matrix 
generated from an $n$-dimensional vector
 $\xb =[x[0],\cdots, x[n-1]]^T \in \Cd^n$ has the following structure:
    \begin{eqnarray}
\hank( \xb) &=& \left[
        \begin{array}{cccc}
       x[0] &  x[1] & \cdots   &  x[d-1]   \\
      x[1]  &  x[2] & \cdots &    x[d] \\
         \vdots    & \vdots     &  \ddots    & \vdots    \\
       x[n-d] &  [n-d+1] & \cdots &  x[n-1]\\
        \end{array}
    \right] \in \Cd^{(n-d+1)\times d} \ .
    \end{eqnarray}
where $ d$ is called a matrix pencil parameter. 
We denote the space of this type of Hankel structure matrices as $\Hc(n,d)$.

An $n\times d$ {\em wrap-around}  Hankel  matrix 
generated from an $n$-dimensional vector
 $\ub=[u[0],\cdots, u[n-1]]^T \in \Cd^n$ is defined as:
  \begin{eqnarray} \label{eq:U}
\hank_c(  \ub) &=& \left[
        \begin{array}{cccc}
        u[0]  &   u[1] & \cdots   &   u[d-1]   \\
       u[1]  &   u[2] & \cdots &     u[d] \\
         \vdots    & \vdots     &  \ddots    & \vdots    \\
        u[n-d]  &   u[n-d+1] & \cdots &   u[n-1]\\ \hline 
        u[n-d+1]  &   u[n-d+2] & \cdots &   u[0] \\
           \vdots    & \vdots     &  \ddots    & \vdots    \\
              u[n-1]  &   u[0] & \cdots &   u[d-2] \\
        \end{array}
    \right] \in \Cd^{n\times d}  \  .
    \end{eqnarray}
 Note that $n\times d$ wrap-around Hankel matrix can be considered as a Hankel matrix of $(d-1)$-element augumented vector
from  $\ub \in \Cd^n$  with the periodic boundary expansion:
 $$\tilde \ub =\left[\ub^T~~ \underbrace{ u[0] ~u[1] ~ \cdots ~u[d-2]}_{(d-1)} \right]^T \in  \Cd^{n+d-1} .  $$
 We denote the space of this type of  wrap-around Hankel structure matrices as $\Hc_c(n,d)$.

\subsection{Annihilating Filter-based Low-Rank Hankel Matrix}

The Fourier CS problem of our interest is to recover the unknown
signal $x(t)$ from the Fourier measurement:
\begin{eqnarray}
\hat x(f)  = {\cal F}\{x(t)\} = \int x(t) e^{-i2\pi f t} dt \  .  
\end{eqnarray}
Without loss of generality, we assume that the support of $x(t)$ is $[0,1]$.
Then, the sampled Fourier data at the Nyquist rate is defined by
$$\hat x[k] = \left. \hat x(f)\right|_{f = k}.$$
We also define a length $(r+1)$-annihilating filter $\hat h[k]$  for $\hat x[k]$ that satisfies 
\begin{eqnarray}\label{eq:zero0}
(\hat h\ast \hat x)[k] =  \sum_{l=0}^r \hat h[l]\hat x[k-l]  = 0, \quad \forall k .
\end{eqnarray}
The existence of the finite length annihilating filter has been extensively studied for FRI signals
\cite{vetterli2002sampling,dragotti2007sampling,maravic2005sampling}.
This will be discussed in more detail later.

Suppose that the filter $\hat h[k]$ is the minimum length annihilating filter. 
Then,  for any $k_1\geq 1$ tap filter $\hat a[k]$, it is easy to see that the following filter with $d= r+k_1$ taps is also an annihilating filter for $\hat x[k]$:
\begin{eqnarray}\label{eq:zero1}
 \hat h_a[k] = (\hat a \ast \hat h ) [k] \quad \Longrightarrow  \quad \sum_{l=0}^r \hat h_a[l]\hat x[k-l]  = 0,~~\forall ~k,  
\end{eqnarray}
because $\hat h_a \ast \hat x = \hat a \ast \hat h \ast \hat x = 0$.
The matrix representation of \eqref{eq:zero1} is given by
$$\conv(\hat \xb) \bar{\hat \hb}_a = \zerob$$
where 
 $\bar{\hat\hb}_a$ denotes 
a vector that reverses  the order of the elements  in
\beq\label{eq:hba}
\hat\hb_a=\left[\hat h_a[0],\cdots, \hat h_a[d-1]\right]^T,
\eeq
and
\begin{eqnarray}\label{eq:hankorg}
\conv(\hat \xb) =\left[
        \begin{array}{cccc}
        \vdots & \vdots & \ddots & \vdots \\
        \hat x[-1]  & \hat x[0] & \cdots   & \hat x[d-2]   \\ \hline
      \hat x[0]  & \hat x[1] & \cdots   & \hat x[d-1]   \\
     \hat x[1]  & \hat x[2] & \cdots &   \hat x[d] \\
         \vdots    & \vdots     &  \ddots    & \vdots    \\
      \hat x[n-d]  & \hat x[n-d+1] & \cdots & \hat x[n-1]\\ \hline
            \hat x[n-d+1]  & \hat x[n-d+2] & \cdots & \hat x[n]\\
              \vdots & \vdots & \ddots & \vdots \\
        \end{array}
    \right]
 \end{eqnarray}
Accordingly,  by choosing $n$ such that $n-d+1> r$ and
defining an $n$-dimensional  vector composed of 
 sampled Fourier data at the Nyquist rate as:
\begin{equation}\label{eq:xb}
\hat \xb = \begin{bmatrix} \hat x[0] & \cdots & \hat x[n-1] \end{bmatrix}^T \in \Cd^{n},
\end{equation}
%
 we can construct the following matrix equation:
\begin{eqnarray}
 \hank(\hat \xb) \bar{\hat \hb}_a = \mathbf{0},
\end{eqnarray}
where the Hankel structure matrix $\hank(\hat \xb)  \in \Hc(n,d) $ is constructed as
    \begin{eqnarray}\label{eq:X2}
\hank(\hat \xb) =\left[
        \begin{array}{cccc}
      \hat x[0]  & \hat x[1] & \cdots   & \hat x[d-1]   \\
     \hat x[1]  & \hat x[2] & \cdots &   \hat x[d] \\
         \vdots    & \vdots     &  \ddots    & \vdots    \\
      \hat x[n-d]  & \hat x[n-d+1] & \cdots & \hat x[n-1]\\
        \end{array}
    \right] 
    \end{eqnarray}
Then, we can show the following key result:
\begin{theorem}\label{thm:hrank}
Let $r+1$ denote the minimum size of  annihilating filters that annihilates sampled Fourier data $\hat x[k]$. 
Assume that $\min\{n-d+1,d\}> r$.
Then,
for a given Hankel structured matrix $\hank(\hat \xb) \in \Hc(n,d)$ constructed in \eqref{eq:X2}, we have
\begin{eqnarray}\label{eq:rankr}
\rank \hank(\hat \xb)=  r, 
\end{eqnarray} 
where $\rank(\cdot)$ denotes a matrix rank.
\end{theorem}
\begin{proof}
See Appendix~\ref{ap:proof_hrank}.
\end{proof}

\subsection{Performance Guarantees for Structured Matrix Completion}

Let $\Omega$ be a multi-set consisting of random indices
from $\{0,\ldots,n-1\}$ such that $|\Omega|=m < n$.
While the standard CS approaches directly
estimate $x(t)$ from  $\hat x[k], k\in \Omega$,  here we propose a two-step approach by exploiting
Theorem~\ref{thm:hrank}. More specifically, we first interpolate $\hat x[k]$ for all $k\in \{0,\ldots, n-1\}$ from the sparse Fourier samples,
and the second step then
applies the existing spectral estimation methods
to estimate $x(t)$ as done in \cite{vetterli2002sampling,dragotti2007sampling,maravic2005sampling}.
Thanks to the low-rankness of the associated Hankel matrix, the first
step can be implemented using the following low-rank matrix completion: 

\begin{equation}
\label{eq:nucmin0}
\begin{array}{lll}
\underset{\gb \in \mathbb{C}^n}{\text{minimize}} & \rank ~\hank(\hat \gb) \\
\text{subject to} & P_\Omega (\hat\gb) = P_\Omega (\hat\xb), 
\end{array}
\end{equation}
where $P_\Omega$ is the projection operator on the sampling location $\Omega$.
Therefore, the remaining question is to verify whether  the low-rank matrix completion approach \eqref{eq:nucmin0} does not compromise 
any optimality compared to the standard Fourier CS, which is the main topic in this section.

The low-rank  matrix completion problem in \eqref{eq:nucmin0} is non-convex, which is difficult to analyze. Therefore, to provide a performance guarantee, we resort to its convex relaxation using the nuclear norm. { Chen and Chi \cite{chen2014robust} provided the first performance guarantee for robust spectral compressed sensing via structured matrix completion by nuclear norm minimization and extended the result to general low-rank Hankel/Toeplitz matrix completion \cite[Theorem~4]{chen2014robust}. However, parts of their proof (e.g., \cite[Appendix~H]{chen2014robust}) critically depend on the special structure given in the standard Vandermonde decomposition. Furthermore, they also used the standard incoherence condition, which is neither assumed nor applied by their incoherence condition \cite[eq. (27)]{chen2014robust}. Therefore, unlike their claim, the main results in \cite{chen2014robust}, in its current forms, apply only to the spectral compressed sensing.} Here, we elaborate on their results so that the performance guarantees apply to the general structured low-rank matrix completion problems which will be described in subsequent sections.
%
%
%
%
%

Recall that the notion of the incoherence plays a crucial role in matrix completion and structured matrix completion.
We recall the definitions using our notations.
Suppose that $M \in \mathbb{C}^{n_1 \times n_2}$ is a rank-$r$ matrix whose SVD is $U \Sigma V^*$ with $U\in \Cd^{n_1\times r}, \Sigma  \in \Cd^{r\times r}$ and
$V\in \Cd^{n_2\times r}$, respectively.
$M$ is said to satisfy the \emph{standard incoherence} condition with parameter $\mu$ if
\begin{equation}
\label{eq:incoherence0}
\begin{aligned}
\max_{1 \leq i \leq n_1} \norm{U^* \eb_i}_2 {} & \leq \sqrt{\frac{\mu r}{n_1}}, \\
\max_{1 \leq j \leq n_2} \norm{V^* \eb_j}_2 {} & \leq \sqrt{\frac{\mu r}{n_2}},
\end{aligned}
\end{equation}
where $\eb_i$ denotes the appropriate size standard coordinate vector with 1 on the $i$-th elements and zeros elsewhere.

To deal with two types of Hankel matrices simultaneously,
we define a {\em linear lifting} operator $\Lc: \mathbb{C}^n \to \mathbb{C}^{n_1 \times n_2}$ that
 lifts a vector $\xb \in \mathbb{C}^n$ to a structured matrix $\Lc (\xb) \in \mathbb{C}^{n_1 \times n_2}$ in a higher dimensional space.
For example, 
the dimension of $\Lc (\xb) \in \mathbb{C}^{n_1 \times n_2}$ is given as
$n_1=n-d+1$ and $n_2=d$ for a lifting to a Hankel matrix in $\Hc(n,d)$, whereas
$n_1=n$ and $n_2=d$ for a lifting to a wrap-around Hankel matrix $\Hc_c(n,d)$.
A linear lifting operator  is regarded as a synthesis operator with respect to basis $\{A_k\}_{k=1}^n,  A_k \in \Cd^{n_1\times n_2}$:
\[
\Lc (\xb) = \sum_{k=1}^n A_k \langle \eb_k, \xb \rangle,
\]
where the specific form of the basis for the case of  $\{A_k\}$ for $\Hc(n,d)$ and $\Hc_c(n,d)$  will be  explained in Appendix~\ref{subsec:basis}.

Then, the completion of a low rank structured matrix $\Lc (\xb)$ from the observation of its partial entries can be done by minimizing the nuclear norm under the measurement fidelity constraint as follows:
\begin{equation}
\label{eq:nucmin}
\begin{array}{lll}
\underset{\gb \in \mathbb{C}^n}{\text{minimize}} & \norm{\Lc (\gb)}_* \\
\text{subject to} & P_\Omega (\gb) = P_\Omega (\xb).
\end{array}
\end{equation}
 where $\|\cdot\|_*$ denotes the matrix nuclear norm.
Then, we have the following main result.

\begin{theorem}
\label{thm:uniqueness}
Let $\Omega = \{j_1,\ldots,j_m\}$ be a multi-set consisting of random indices
where $j_k$'s are i.i.d. following the uniform distribution on $\{0,\ldots,n-1\}$.
Suppose $\Lc$ correspond to one of the structured matrices in 
 $\Hc(n,d)$ and $\Hc_c(n,d)$.
Suppose, furthermore, that  $\Lc (\xb)$ is of rank-$r$ and satisfies the standard incoherence condition in \eqref{eq:incoherence0} with parameter $\mu$.
Then there exists an absolute constant $c_1$ such that
$\xb$ is the unique minimizer to \eqref{eq:nucmin} with probability $1 - 1/n^2$, provided
\begin{equation}
\label{eq:samp_comp}
m \geq c_1  \mu c_s  r \log^\alpha n,
\end{equation}
where $\alpha = 2$ if each images of $\Lc$ has the wrap-around property; $\alpha = 4$, otherwise,
and $c_s := \max\{n/n_1, n/n_2\}$.
\end{theorem}
\begin{proof}
See Appendix~\ref{subsec:pf:thm:uniqueness}.
\end{proof}


Note that  Theorem~\ref{thm:uniqueness}  provides a  generalized version of performance guarantee compared to  the previous work \cite{chen2014robust}.
Specifically, Theorem~\ref{thm:uniqueness} holds for other structured matrices, 
 if the associated basis matrix
$\{A_k\}$ satisfies the specific condition described in detail in Eq.~\eqref{eq:rowcolsp}.
In addition, for the case of signals with wrap-around Hankel matrix (which will be explained later),
the log exponentional factor $\alpha$  in \eqref{eq:samp_comp} becomes 2, which reduces the sampling rate.
This sampling rate reduction is novel and was not observed in the previous work \cite{chen2014robust}.

Next, we consider the recovery of $\xb$ from its partial entries with noise.
Let $\yb$ denote a corrupted version of $\xb$.
The unknown structured matrix $\Lc (\xb)$ can be estimated via
\begin{equation}
\label{eq:nucmin_noisy}
\begin{array}{lll}
\underset{\gb \in \mathbb{C}^n}{\text{minimize}} & \norm{\Lc (\gb)}_* \\
\text{subject to} & \norm{P_\Omega (\gb - {\yb})}_2 \leq \delta.
\end{array}
\end{equation}
Then, we have the following stability guarantee:
\begin{theorem}
\label{thm:stability}
Suppose the noisy data ${\yb} \in \Cd^n$ satisfies $\norm{P_\Omega (\yb - {\xb})}_2 \leq \delta$ and $\xb \in \Cd^n$ is the  noiseless data.
Under the hypotheses of Theorem~\ref{thm:uniqueness}, there exists an absolute constant $c_1,c_2$ such that
with probability $1 - 1/n^2$, the solution $\gb$ to \eqref{eq:nucmin_noisy} satisfies
\[
\norm{\Lc ({\xb}) - \Lc (\gb)}_{\mathrm{F}} \leq c_2 n^2 \delta,
\]
provided that \eqref{eq:samp_comp} is satisfied with $c_1$.
\end{theorem}

\begin{proof}[Proof of Theorem~\ref{thm:stability}]
Theorem~\ref{thm:uniqueness} extends to the noisy case similarly to the previous work \cite{chen2014robust}.
We only need to replace \cite[Lemma~1]{chen2014robust} by our Lemma~\ref{lemma:uniqueness}.
\end{proof}

Note that  Theorem~\ref{thm:stability}  provides an improved performance guarantee
with significantly smaller noise amplification factor, compared to $n^3$ dependent noisy amplification factor in  the previous work \cite{chen2014robust}.

\section{Guaranteed Reconstruction of FRI Signals}
\label{sec:fri}

The explicit derivation of the minimum length finite length annihilating filter was one of the most important contributions of
the sampling theory of FRI signals  \cite{vetterli2002sampling,dragotti2007sampling,maravic2005sampling}.
Therefore, by combing the results in the previous section,
we can provide performance guarantees for the recovery of FRI signals from partial Fourier measurements.

\subsection{Spectral Compressed Sensing: Recovery of Stream of Diracs}

Consider the periodic stream of Diracs
  described by the superposition of  $r$ impulses
\begin{equation}\label{eq:signal3}
x(t) = \sum_{j=0}^{r-1} a_j \delta \left( t- t_j \right) \, \quad t_j \in [0, 1].
\end{equation}
Then, the discrete Fourier data are given by
\begin{eqnarray}\label{eq:fs}
\hat x[k] = \sum_{j=0}^{r-1} a_j e^{-i2\pi k t_j } \ .
\end{eqnarray}
As mentioned before, the spectral compressed sensing by Chen and Chi \cite{chen2014robust} or Tang \cite{{tang2013compressed}} correspond to this case, in which
 they are interested in recovering \eqref{eq:signal3} from a subsampled spectral measurements.

One of the important contributions of this section is to show that the spectral compressed sensing can be equivalently explained using  annihilating filter-based low-rank Hankel matrix.
Specifically,    for the stream of Diracs,  the minimum length annihilating filter $\hat h[k]$ has the following z-transform representation  \cite{vetterli2002sampling}: 
\begin{eqnarray}\label{eq:afilter}
\hat h(z)  &=& \sum_{l=0}^r \hat h[l] z^{-l} = \prod_{j=0}^{r-1} (1- e^{-i2\pi t_j} z^{-1}) \ ,
\end{eqnarray}
 because
\begin{eqnarray}
(\hat h\ast \hat x)[k] &=& \sum_{l=0}^k \hat h[l]\hat x[k-l] \nonumber \\
&=& \sum_{l=0}^r \sum_{j=0}^{r-1} a_j \hat h[l]u_j^{k-l}  \nonumber \\
&=& \sum_{j=0}^{r-1}a_j\underbrace{\left( \sum_{l=0}^r \hat h[p]u_j^{-l} \right)}_{\hat h(u_j)}u_j^k = 0 \label{eq:fri}
\end{eqnarray}
where $u_j = e^{-i2\pi t_j }$ \cite{vetterli2002sampling,dragotti2007sampling,maravic2005sampling}.
Accordingly, the filter length is $r+1$, which is low ranked if $\min\{n-d+1,d\} >r$.

Therefore,  by utilizing Theorem~\ref{thm:hrank} and Theorem~\ref{thm:uniqueness},  we can provide the performance guarantee
of the following nuclear norm minimization to estimate the Fourier samples:
\begin{eqnarray}\label{eq:EMaC}
\quad & \min_{\gb\in \Cd^{n} } & \|\hank(\gb)\|_* \\
&\mbox{subject to } & P_\Omega(\gb) = P_\Omega(\hat \xb) \nonumber
\end{eqnarray}
where $\hank(\gb) \in \Hc(n,d)$.
\begin{theorem}
\label{thm:unique_dirac} 
For a given stream of Diracs in Eq.~\eqref{eq:signal3}, $\hat\xb$ denotes the noiseless discrete  Fourier data in \eqref{eq:xb}.
Suppose, furthermore,  $d$ is given by $\min\{n-d+1,d \}> r$.
Let $\Omega = \{j_1,\ldots,j_m\}$ is a multi-set consisting of random indices
where $j_k$'s are i.i.d. following the uniform distribution on $\{0,\ldots,n-1\}$.
Then,  there exists an absolute constant $c_1$ such that
$\hat\xb$ is the unique minimizer to \eqref{eq:EMaC} with probability $1 - 1/n^2$, provided
\begin{equation}
\label{eq:samp_comp_dirac}
m \geq c_1  \mu c_s  r \log^4n,
\end{equation}
where  $c_s := \max\{n/(n-d+1), n/d\}$.
\end{theorem}
\begin{proof}
This is a simple consequence of Theorem~\ref{thm:hrank} and Theorem~\ref{thm:uniqueness}, because 
the minimum annihilating filter size from \eqref{eq:afilter} is $r+1$.
\end{proof}

This result appears identical to that of Chen and Chi  \cite{chen2014robust}.
However, they explicitly utilized  the standard Vandermonde decomposition.
On the contrary, the annihilating filter-based construction of low-rank Hankel matrix is more general that can cover all  FRI signal
models as will be shown later.

For the noisy measurement,  we interpolate the missing Fourier data using the following low-rank matrix completion:
\begin{eqnarray}
 &\quad &\min  \|\hank(\gb)\|_* \quad  \label{eq:nucmin_noisy_dirac}\\
 &\mbox{subject to}& \|P_{\Omega}(\gb)- P_\Omega(\hat\yb) \|\leq \delta \notag 
\end{eqnarray}
where $\hat\yb$ is the noisy Fourier data. Then,
Theorem~\ref{thm:stability} informs us that we can improve upon the results by Chen and Chi  \cite{chen2014robust} (from $n^3$ to $n^2$):
\begin{theorem}
\label{thm:stability_dirac}
Suppose that the noisy Fourier data $\hat{\yb}$ satisfies $\norm{P_\Omega (\hat\yb - \hat{\xb})}_2 \leq \delta$, where $\hat\xb$ denotes the noiseless discrete  Fourier data in \eqref{eq:xb}.
Under the hypotheses of Theorem~\ref{thm:unique_dirac}, there exists an absolute constant $c_1,c_2$ such that
with probability $1 - 1/n^2$, the solution $\gb$ to \eqref{eq:nucmin_noisy_dirac} satisfies
\[
\norm{\hank (\hat{\xb}) - \hank(\gb)}_{\mathrm{F}} \leq c_2 n^2 \delta,
\]
provided that \eqref{eq:samp_comp_dirac} is satisfied with $c_1$.
\end{theorem}

\subsection{Stream of Differentiated Diracs}

Another important class of FRI signal is  a stream of differentiated Diracs:
\begin{equation}\label{eq:ddiracs}
x(t) = \sum_{j=0}^{r-1}  \sum_{l=0}^{l_j-1}a_{j,l}\delta^{(l)} (t-t_j) \ ,
\end{equation}
where
$\delta^{(l)}$ denotes the $l$-th derivative of Diracs in the distributions sense.
Thus, its
 Fourier transform is given by
\begin{eqnarray}\label{eq:ddiracf}
\hat x(f) =   \sum_{j=0}^{r-1}   \sum_{l=0}^{l_j-1}a_{j,l}   (i2\pi f)^{l} e^{-i2\pi f t_j} \,
\end{eqnarray}
whose discrete samples are given by
\begin{eqnarray}\label{eq:spec}
\hat x[k]:=  \hat x(k)  =     \sum_{j=0}^{r-1}   \sum_{l=0}^{l_j-1}a_{j,l}  \left({i2\pi k}\right)^{l} e^{-i2\pi k t_j }  .
\end{eqnarray} 
Then, there exists an associated minimum length annihilating filter  whose z-transform is given by:
\begin{eqnarray}\label{eq:dann0}
\hat h(z) = \prod_{j=0}^{r-1}(1-u_j z^{-1})^{l_j}
\end{eqnarray}
where $u_j = e^{-i2\pi t_j}$  \cite{vetterli2002sampling}.
Therefore, we can provide the following performance guarantees:
\begin{theorem}
\label{thm:unique_ddirac}
For a given stream of differentiated Diracs in Eq.~\eqref{eq:ddiracs},  $\hat\xb$ denotes the noiseless discrete  Fourier data in \eqref{eq:xb}.
Suppose, furthermore,  $d$ is given by $\min\{n-d+1,d\} > \sum_{j=0}^{r-1}l_j $.
Let   $\Omega = \{j_1,\ldots,j_m\}$ be a multi-set consisting of random indices
where $j_k$'s are i.i.d. following the uniform distribution on $\{0,\ldots,n-1\}$.
Then,  there exists an absolute constant $c_1$ such that
$\hat\xb$ is the unique minimizer to \eqref{eq:EMaC} with probability $1 - 1/n^2$, provided
\begin{equation}
\label{eq:samp_comp_ddirac}
m \geq c_1  \mu c_s  \left( \sum_{j=0}^{r-1}l_j \right) \log^4n,
\end{equation}
where  $c_s := \max\{n/(n-d+1), n/d\}$.
\end{theorem}
\begin{proof}
This is a simple consequence of Theorem~\ref{thm:hrank} and Theorem~\ref{thm:uniqueness}, because 
the minimum annihilating filter size from \eqref{eq:dann0} is $\left( \sum_{j=0}^{r-1}l_j \right)+1$.
%
%
%
%
\end{proof}

\begin{theorem}
\label{thm:stability_ddirac}
Suppose  the noisy Fourier data $\hat{\yb}$ satisfies $\norm{P_\Omega (\hat\yb - \hat{\xb})}_2 \leq \delta$, where $\hat\xb$ denotes the noiseless discrete  Fourier data in \eqref{eq:xb}.
Under the hypotheses of Theorem~\ref{thm:unique_ddirac}, there exists an absolute constant $c_1,c_2$ such that
with probability $1 - 1/n^2$, the solution $\gb$ to \eqref{eq:nucmin_noisy_dirac} satisfies
\[
\norm{\hank (\hat{\xb}) - \hank(\gb)}_{\mathrm{F}} \leq c_2 n^2 \delta,
\]
provided that \eqref{eq:samp_comp_ddirac} is satisfied with $c_1$.
\end{theorem}

\subsection{Non-uniform Splines}

Note that signals may not be sparse in the image domain, but can be sparsified in a transform domain. Our goal is to find a generalized framework,  whose sampling rate can be reduced down to the transform domain sparsity level.
Specifically, the signal $x$ of our interest is a non-uniform spline that can be represented by : 
\begin{equation}\label{eq:ssp}
\mathrm{L} x = w
\end{equation}
where $\mathrm{L}$ denotes a constant coefficient linear  differential equation  that is often called the continuous domain whitening operator in \cite{unser2014unified,unser2014unified2}: 
\begin{equation}\label{eq:L0}
\mathrm L := b_K\mathrm \partial^K+b_{K-1}\mathrm \partial^{K-1}+\ldots+b_1\mathrm \partial+b_0
\end{equation}
and 
$w$ is a continuous sparse innovation:
\begin{eqnarray}\label{eq:w}
w(t) = \sum_{j=0}^{r-1} a_j \delta \left( t-t_j \right) \, \quad . 
\end{eqnarray}
For example, if the underlying signal is piecewise constant, we can set $\mathrm L$ as the first differentiation. In this case, $x$ corresponds to the total variation signal model.
Then, by taking the Fourier transform of \eqref{eq:ssp}, we have
\begin{eqnarray}\label{eq:y}
\hat z(f): = \hat l (f) \hat x(f)  = \sum_{j=0}^{r-1} a_j e^{-i2\pi f t_j}
\end{eqnarray}
where 
\begin{eqnarray}
\hat l(f) =  b_K (i2\pi f)^K +b_{K-1} (i2\pi f)^{K-1}+\ldots+b_1(i2\pi f)+b_0
\end{eqnarray}
Accordingly, the same  filter  $\hat h[n]$  whose z-transform is given by
\eqref{eq:afilter}
can annihilate the discrete samples of the weighted spectrum $\hat z(f) = \hat l(f)\hat x(f)$,
and the Hankel matrix  $\hank(\hat\zb)\in \Hc(n,d)$  from the weighted spectrum $\hat z(f)$ satisfies the following rank condition:
\begin{eqnarray*}
\rank \hank(\hat \zb) = r.  
\end{eqnarray*} 
Thanks to the low-rankness, the missing  Fourier data can be interpolated using the following matrix completion problem:
\begin{eqnarray}\label{eq:EMaC2}
(P_{w})
 &\min_{\gb\in \Cd^{n} } & \| \hank (\gb)\|_* \\
&\mbox{subject to } & P_\Omega(\gb) = P_\Omega(\hat \lb \odot \hat \xb) \nonumber  \  ,
\end{eqnarray}
or, for noisy  Fourier measurements $\hat\yb$, 
\begin{eqnarray}\label{eq:EMaC2_noisy}
(P_{w}')
 &\min_{\gb\in \Cd^{n} } & \| \hank (\gb)\|_* \\
&\mbox{subject to } & \|P_\Omega(\gb) - P_\Omega(\hat \lb \odot \hat \yb) \|\leq \delta  \nonumber  \  ,
\end{eqnarray}
where $\odot$ denotes the Hadamard product, and $\hat\lb$ and $\hat \xb$ denotes the vectors composed of full samples of $\hat l[k]$ and $\hat x[k]$, respectively.
After solving $(P_w)$, 
the missing  spectral data $\hat x[k]$ can be obtained by dividing by the weight, i.e. $\hat x[k] =  g[k]/\hat l[k]$ assuming that $\hat l[k] \neq 0$, where
$g[k]$ is the estimated Fourier data using $(P_w)$.
As for the  sample  $\hat x[k]$ at the spectral null of the filter $\hat l[k]$,  the corresponding elements should be included as  measurements. 

Now, we can provide the following performance guarantee:
\begin{theorem}
\label{thm:unique_spline}
For a given non-uniform splines in Eq.~\eqref{eq:ssp}, $\hat\xb$ denotes the noiseless discrete  Fourier data in \eqref{eq:xb}.
Suppose, furthermore,  $d$ is given by $\min\{n-d+1,d\} > r $ and  $\Omega = \{j_1,\ldots,j_m\}$ be a multi-set consisting of random indices
where $j_k$'s are i.i.d. following the uniform distribution on $\{0,\ldots,n-1\}$.
Then,  there exists an absolute constant $c_1$ such that
$\hat\xb$ is the unique minimizer to \eqref{eq:EMaC2} with probability $1 - 1/n^2$, provided
\begin{equation}
\label{eq:samp_comp_spline}
m \geq c_1  \mu c_s r \log^4n,
\end{equation}
where  $c_s := \max\{n/(n-d+1), n/d\}$.
\end{theorem}

\begin{theorem}
\label{thm:stability_spline}
Suppose  that  the noisy Fourier data $\hat{\yb}$ satisfies $\norm{P_\Omega (\hat \lb \odot \hat\yb - \hat \lb \odot \hat{\xb})}_2 \leq \delta$, where $\hat\xb$ denotes the noiseless discrete  Fourier data in \eqref{eq:xb}.
Under the hypotheses of Theorem~\ref{thm:unique_spline}, there exists an absolute constant $c_1,c_2$ such that
with probability $1 - 1/n^2$, the solution $\gb$ to \eqref{eq:EMaC2_noisy} satisfies
\[
\norm{\hank (\hat \lb \odot \hat{\xb}) - \hank(\gb)}_{\mathrm{F}} \leq c_2 n^2 \delta,
\]
provided that \eqref{eq:samp_comp_spline} is satisfied with $c_1$.
\end{theorem}

\subsection{Piecewise Polynomials}

A
signal is a periodic piecewise polynomial with $r$ pieces
each of maximum degree $q$ if and only if its $(q+1)$ derivative
is a stream of differentiated Diracs given by 
\begin{eqnarray}\label{eq:ppoly}
x^{(q+1)}(t) = \sum_{j=0}^{r-1}  \sum_{l=0}^{q}a_{j,l}\delta^{(l)} (t-t_j) \ .
\end{eqnarray}
In this case,  the corresponding Fourier transform relationship is given by
\begin{eqnarray}\label{eq:poly}
\hat z(f) := (i2\pi f)^{(q+1)} \hat x(f)  = \sum_{j=0}^{r-1} \sum_{l=0}^q a_{j,l} (i2\pi f)^l e^{-i2\pi  t_j}.
\end{eqnarray}
Since the righthand side of \eqref{eq:poly} is a special case of  \eqref{eq:ddiracf},  the associated minimum length annihilating filter has the following z-transform representation:
\begin{eqnarray}\label{eq:dann}
\hat h(z) = \prod_{j=0}^{r-1}(1-u_j z^{-1})^{q} \  .
\end{eqnarray}
whose  filter length is given by $(q+1)r+1$.
Therefore, we can provide the following performance guarantee:
\begin{theorem}
\label{thm:unique_ppoly}
For a given piecewise smooth polynomial  in Eq.~\eqref{eq:ppoly}, let $\hat \zb$ denotes the discrete spectral samples of $\hat z(f)=\hat l(f)\hat x(f)$
with $\hat l(f)=(i2\pi f)^{(q+1)}$.
Suppose, furthermore,   $d$ is given by $\min\{n-d+1,d\} > (q+1)r $ and  $\Omega = \{j_1,\ldots,j_m\}$ be a multi-set consisting of random indices
where $j_k$'s are i.i.d. following the uniform distribution on $\{0,\ldots,n-1\}$.
Then,  there exists an absolute constant $c_1$ such that
$\hat\xb$ is the unique minimizer to \eqref{eq:EMaC2} with probability $1 - 1/n^2$, provided
\begin{equation}
\label{eq:samp_comp_ppoly}
m \geq c_1  \mu c_s (q+1)r \log^4n,
\end{equation}
where  $c_s := \max\{n/(n-d+1), n/d\}$.
\end{theorem}
\begin{proof}
This is a simple consequence of Theorem~\ref{thm:hrank} and Theorem~\ref{thm:uniqueness}, because 
the minimum annihilating filter size from \eqref{eq:dann} is $ (q+1)r+1$.
\end{proof}

\begin{theorem}
\label{thm:stability_ppoly}
Suppose  that noisy Fourier data $\hat{\yb}$ satisfies $\norm{P_\Omega (\hat \lb \odot \hat\yb - \hat \lb \odot \hat{\xb})}_2 \leq \delta$, where $\hat\xb$ denotes the noiseless discrete  Fourier data in \eqref{eq:xb}.
Under the hypotheses of Theorem~\ref{thm:unique_ppoly}, there exists an absolute constant $c_1,c_2$ such that
with probability $1 - 1/n^2$, the solution $\gb$ to \eqref{eq:EMaC2_noisy} satisfies
\[
\norm{\hank (\hat \lb \odot \hat{\xb}) - \hank(\gb)}_{\mathrm{F}} \leq c_2 n^2 \delta,
\]
provided that \eqref{eq:samp_comp_ppoly} is satisfied with $c_1$.
\end{theorem}

\subsection{Incoherence and the Minimum Separation}
\label{sec:incoherence_min}

Note that  the proposed low-rank interpolation achieves 
near optimal sample rate 
while universally applying to different FRI signal models of the same order (e.g., stream of Diracs and stream of differentiated Diracs). 
Moreover,  even though the concept of the minimum separation between the successive spikes was essential for the performance guarantee of super-resolution 
in  Candes and Fernandez-Granda \cite{candes2014towards,candes2013super}, Tang \cite{tang2013compressed,tang2015near}, etc,  
similar expression is not observable in Theorem~\ref{thm:unique_dirac}-Theorem~\ref{thm:stability_ppoly}.
This looks mysterious. 
Therefore, the main goal of this section is  
to show that these information are still required but hidden in the incoherence parameter $\mu$.

Note that the proof in Theorem~\ref{thm:hrank} implies that  the explicit form  $\hat x[k]$ given by 
\begin{equation}\label{eq:xk}
\hat x[k]:=\sum\limits_{j=0}^{p-1}\sum\limits_{l=0}^{l_j-1}a_{j,l}k^l{\lambda_j}^k,~\quad\quad\quad {\rm where}\quad~r = \sum_{j=0}^{p-1} l_j,
\end{equation}
 is a necessary and sufficient condition to have the
low-rank Hankel matrix.
Here, $\lambda_j = e^{-i2\pi t_j}$  for FRI signals.
Furthermore, for a Hankel matrix constructed using the signal model in \eqref{eq:xk},  there exist an exponentional decomposition of Hankel matrix using {\em confluent} Vandermonde matrix \cite{boley1998vandermonde}.
Specifically, define a confluent Vandermonde matrix $\mathcal{V}_{n-d+1}  \in  \Cd^{(n-d+1)\times r}$ (resp. $\mathcal{V}_d  \in  \Cd^{d\times r}$):
\begin{eqnarray}
\mathcal{V}_{n-d+1}  = \begin{bmatrix}  C_{n-d+1}^{l_0}(\lambda_0) & C_{n-d+1}^{l_1}(\lambda_1) & \cdots & C_{n-d+1}^{l_{p-1}}(\lambda_{m_{p-1}})  \end{bmatrix} ,
\end{eqnarray}
where  the $(m,l)$ element of the sub-matrix $C_{n-d+1}^{l_i}(\lambda) \in \Cd^{(n-d+1)\times l_i}$ is given by
\begin{equation}\label{eq:CC}
\left[C_{n-d+1}^{l}(\lambda)\right]_{m,l}= \begin{cases} 0, &m<l\\   
\frac{(m-1)!}{(m-l)!} \lambda^{i-j} , & \mbox{otherwise} \end{cases} \ . 
\end{equation}
 Then,  the associated Hankel matrix $\hank(\hat\xb) \in \Hc(n,d)$ with $\min\{n-d+1,d\}>r$ has the following {\em generalized Vandermonde decomposition} \cite{badeau2006high,batenkov2013accuracy,boley1998vandermonde,batenkov2013decimated}:
\begin{eqnarray}\label{eq:genVan}
\hank(\hat\xb) =  \mathcal{V}_{n-d+1} ~\mathcal{B}~ \mathcal{V}_d^T  ,\quad 
\end{eqnarray}
where $\mathcal{V}_{n-d+1} \in \Cd^{(n-d+1)\times r}$ and $\mathcal{V}_d\in \Cd^{d\times r}$ are the confluent Vandermonde matrices and
$\mathcal{B}$ is a $r\times r$ block diagonal matrix given by
$$\mathcal{B} =  \begin{bmatrix} H_{0} & 0 & \cdots & 0 \\ 0 & H_{1} & \ddots \\ \vdots & \ddots & \ddots &  0  \\ 0 &  \cdots & 0 & H_{{p-1}} \end{bmatrix},$$
where $H_{i}$ is the $m_i \times m_i$ upper anti-triangular Hankel matrix \cite{badeau2006high}.
Because $\mathcal{B}$ in \eqref{eq:genVan} is a full rank block diagonal matrix,
for a given SVD of $\hank(\hat\xb)=U\Sigma V^H$ with $U \in \Cd^{(n-d+1)\times r}, V \in \Cd^{d\times r}$ and $\Sigma \in \Cd^{r\times r}$, we have
$$\Ran~ U = \Ran~ \mathcal{V}_{n-d+1},\quad \Ran~ V = \Ran~ \mathcal{V}_{d}. $$
Accordingly, we can derive the following upper bound of the standard coherence:
\begin{lemma}\label{lem:sigma}
For a Hankel matrix $\hank(\hat\xb)\in \Hc(n,d)$ with the decomposition  in \eqref{eq:genVan},  the standard coherence $\mu$  in \eqref{eq:incoherence0} satisfies:
\begin{eqnarray}\label{eq:kappa}
\mu &\leq & \max\left\{\frac{\zeta_{n-d+1}}{\sigma_{\min}\left(\Vc_{n-d+1}^*\Vc_{n-d+1}\right)},~\frac{\zeta_{d}}{\sigma_{\min}\left(\Vc_{d}^*\Vc_{d}\right)}\right\}
\end{eqnarray}
where $\sigma_{\min}(\cdot)$  denotes the least singular value and the constant $\zeta_N, N\in \mathbb{N}$ is defined by
\begin{eqnarray}\label{eq:eta}
\zeta_{N} &=& N   \left[\frac{(N-1)!}{(N-l_{\max})!}\right]^2 ~ \ , 
\end{eqnarray}
and $l_{\max} := \max\limits_{0\leq j\leq p-1} l_j$.
\end{lemma}
\begin{proof}
Since $U$ (resp. $\Vc_{n-d+1}$) and $V$ (resp. $\Vc_d$) determine the same
column (resp. row) space, we can write
\begin{eqnarray*}
UU^* &=& \Vc_{n-d+1}\left(\Vc_{n-d+1}^*\Vc_{n-d+1}\right)^{-1}\Vc_{n-d+1}^*\\
VV^* &=& \Vc_{d}\left(\Vc_{d}^*\Vc_{d}\right)^{-1}\Vc_{d}^* 
\end{eqnarray*}
Thus, we have
\begin{eqnarray*}
\max_{1\leq i \leq d}\|V^*\eb_i\|_2^2 &= & \max_{1\leq i \leq d} \eb_i^* \Vc_{d}\left(\Vc_{d}^*\Vc_{d}\right)^{-1}\Vc_{d}^*\ \eb_i \\
&\leq &  \frac{1}{\sigma_{\min}\left(\Vc_{d}^*\Vc_{d}\right)} \max_{1\leq i \leq d}\|\Vc_{d}^*\eb_i\|^2 
\end{eqnarray*}
Moreover, we have
\begin{eqnarray*}
 \max_{1\leq i \leq d} \|\Vc_{d}^*\eb_i\|^2  &=&  \|\Vc_{d}^*\eb_{d}\|^2  \\
 &=& \sum_{j=0}^{p-1} \sum_{l=1}^{l_j} \left[ \frac{(d-1)!}{(d-l)!}\right]^2 \\
 &\leq & r  \left[ \frac{(d-1)!}{(d-l_{\max})!}\right]^2
\end{eqnarray*}
where we use \eqref{eq:CC} and  $r=\sum_{j=0}^{p-1} l_j$.
Similarly, 
\begin{eqnarray*}
\max_{1\leq i \leq {n-d+1}}\|U^*\eb_i\|_2^2   &\leq & \frac{r}{\sigma_{\min}\left(\Vc_{n-d+1}^*\Vc_{n-d+1}\right)}   \left[ \frac{(n-d)!}{(n-d+1-l_{\max})!}\right]^2
\end{eqnarray*}
Therefore, using  the definition of $\mu$ in \eqref{eq:incoherence0},  we can obtain \eqref{eq:kappa}.  This concludes the proof.
\end{proof}

Note that this is  an extension of the approach in \cite[Appendix C, Section III.A]{chen2014robust} which tried to bound the mutual coherence for the standard Vandermonde decomposition,
(i.e. $l_0=\cdots= l_{p-1}=1$)  by a small number.
Specifically, for the cases of random frequency locations 
or  small perturbation  off the grid, they showed that the incoherence parameters become small  \cite{chen2014robust}. However, the dependency of $\mu$
on the minimum separation was not explicit in their discussion.

Recently,   Moitra  \cite{moitra2015super} discovered a very intuitive relationship between  the least/largest singular values 
of Vandermonde matrix  and  the minimum separation distance 
$\Delta = \min_{i\neq j} |t_i-t_j|.$  
Specifically, by making novel connections between extremal functions and the spectral properties of Vandermonde
matrices $\Vc_N$, Moitra  showed that  if $N> 1/\Delta +1$, then  the least singular value is bounded as
$$\sigma_{\min}(\Vc_N^*\Vc_N)  \geq {N-1/\Delta-1}.$$
If applied in our problem  involving  $\Vc_{n-d+1}$ and $\Vc_d$, 
 the resulting upper bound of the coherence parameter  for standard Vandermonde matrix is given by
\begin{equation}\label{eq:muup}
\mu \leq \frac{n/2}{n/2-1/\Delta-1} ,
\end{equation}
which approaches to one  with sufficiently large $n$.
 Eq.~\eqref{eq:muup}  is obtained  because  the matrix pencil size $d=n/2$ 
gives the optimal trade-off between $\sigma_{\min}\left(\Vc_{n-d+1}^*\Vc_{n-d+1}\right)$ and $\sigma_{\min}\left(\Vc_{d}^*\Vc_{d}\right)$, and
$\zeta_{n/2}=n/2$ owing to $l_{\max}=1$.
 Note that this  coincides with the minimum
separation in Fernandez-Granda \cite{candes2014towards,candes2013super} and Tang \cite{tang2013compressed,tang2015near}.
However, compared to these approaches \cite{candes2014towards,candes2013super,tang2013compressed,tang2015near} that
require the minimum separation as a {\em hard} constraint,  
our approach requires it as a {\em soft} oversampling factor in terms of the incoherence parameter $\mu$.  


Then, where is the difference originated ?  We argue that this comes from  different uses of interpolation functions. 
Specifically, in  Candes, Fernandez-Granda \cite{candes2014towards,candes2013super} and Tang \cite{tang2013compressed,tang2015near},
dual polynomial function that interpolates the sign at the singularity locations should be found to construct a dual certificate. 
On other hand,  in the proposed annihilating filter based approach, the interpolating function is a smooth function that has zero-crossings at the singularity locations.
To see this,  let $\hat h[k]$ is an annihilating filter that annihilates $\hat x[k]$.  Then there exists an {\em annihilating function} $h(t)$ such that
\begin{eqnarray*}
\hat x[k] \ast \hat h[k] = 0, \forall k &\Longleftrightarrow &  x(t)h(t)= 0, \quad\forall t,
\end{eqnarray*}
so $h(t)=0$ whenever $x(t) \neq 0$. The construction of the annihilating function $h(t)$ is extremely easy and can be readily obtained by the multiplications of  sinusoids
(for example,  to null out $r$-periodic stream of Diracs within $[0,1]$, we set $f(t) = \prod_{j=0}^{r-1}(e^{i2\pi t}-e^{i2\pi t_j})$).
Moreover, this approach can be easily extended to have multiple roots, which is required for differentiated Diracs.
We believe that the ``soft constraint'' originated from annihilating function is one of the  key ingredients that
enables recovery of general FRI signals which was not possible by the existing super-resolution methods  \cite{candes2014towards,candes2013super,tang2013compressed,tang2015near}.

%


The derivation of the least singular value for the confluence Vandermonde matrix  have been also an important topic of researches \cite{batenkov16,batenkov2013decimated,batenkov2013accuracy,gautschi1974norm,gautschi1962inverses,gautschi1963inverses,gautschi1978inverses}.
In general, it will
also depend on the minimum separation distance \cite{batenkov16}.
However, the explicit tight bound of the least singular value is not available in general, so we leave this for future work.

%

\subsection{Recovery of Continuous Domain FRI Signals After Interpolation}
\label{sec:matrix_pencil}

Regardless of the unknown signal type (the stream of Diracs or a stream of differentiated Diracs),
note that 
an identical low-rank interpolator can be used.
Once the spectrum $\hat x[k]$ is fully interpolated, 
in the subsequent step,  
Prony's method and matrix pencil algorithm can identify the signal model  from 
the roots of the estimated annihilator filter as done in  \cite{vetterli2002sampling,dragotti2007sampling,maravic2005sampling}. 
Accordingly, our robustness guarantees on the low-rank matrix entries can be translated in terms of the actual signal that is recovered (for example,  on the support or amplitudes of the
spike in the case of recovery of spike superpositions).
In fact, this has been also an active area of researches \cite{batenkov16,batenkov2013decimated,batenkov2013accuracy,moitra2015super,DBLP:conf/isit/AubelB16,badeau2006high}, and we again exploit these findings
for our second step of signal recovery.
For example, see Moitra \cite{moitra2015super}  for more details on the error bound for the case of modified matrix pencil approach for recovery of Diracs.
In addition,  Batenkov
has recently generalized this for the recovery of general signals in \eqref{eq:xk} \cite{batenkov16}.
The common findings are that
the estimation error for the location parameter $\{t_j\}_{j=0}^{p-1}$ and the magnitude $a_{j,l}$ 
are bounded by the condition number of the confluent Vandermonde matrix  as well as the minimum separation distance $\Delta$.
Moreover,  matrix pencil approaches  such as  Estimation of Signal Parameters via Rotational Invariance
Techniques (ESPRIT) method \cite{roy1989esprit} is shown to stably recovery the locations \cite{DBLP:conf/isit/AubelB16,batenkov16}.

%
%
%
%
%
%
%
%
%
Here, we briefly review the matrix pencil approach for the signal recovery  \cite{badeau2006high,sarkar1995using}.
Specifically, for a given confluent Vandermonde matrix $\Vc_{n-d+1}$, let $\Vc_{n-d+1}^{\downarrow}$ be the matrix  extracted from $\Vc_{n-d+1}$ by deleting the last row.
Similarly,   let $\Vc_{n-d+1}^\uparrow$ be the matrix extracted from $\Vc_{n-d+1}$ by deleting the first row.
Then, $\Vc_{n-d+1}^{\downarrow}$  and $\Vc^\uparrow$ span the same signal subspace  and
$$\Vc_{n-d+1}^\uparrow =  \Vc_{n-d+1}^{\downarrow}J$$
where $J$ is the $r\times r$ block diagonal matrix
$$J = \begin{bmatrix} J_{m_0}(\lambda_0) &  0 & \cdots & 0 \\ 0 & J_{m_1}(\lambda_1) &  \ddots &  \vdots \\ \vdots & \ddots & \ddots & \vdots & \\
0 & \cdots &  0 & J_{m_{p-1}}(\lambda_{p-1})\end{bmatrix},$$
where $J_{m_i}(\lambda_i)$ denotes the $m_i \times m_i$ Jordan block \cite{boley1998vandermonde,badeau2006high}:
$$
J_{m_i}= \begin{bmatrix} \lambda_i & 1 &  0 & \cdots & 0 \\  0 & \lambda_i &  1 & \ddots &  \vdots \\ 
0 & 0 &\lambda_i & \ddots & 0 \\ \vdots &  \ddots  & \ddots & \ddots & 1 \\
0 & \cdots &  0 & 0 &\lambda_i \end{bmatrix} \ .
$$
In practice, the confluence Vandermonde matrix $\Vc_{n-d+1}$ is unknown, but a  $(n-d+1)\times r$ matrix
$W$ that spans the signal subspace can be estimated using singular value decomposition (SVD).
Then, we can easily see that
$$W^\uparrow = W^\downarrow \Phi$$
where $r\times r$ spectral matrix $\Phi$ is given by
$$\Phi = GJ G^{-1}$$
for some matrix $G$. 
Finally, the matrix pencil algorithm computes the eigenvalues of $\Phi$ matrix  from which the estimated poles and
their multiplicities are estimated.

\section{Guaranteed Reconstruction of Cardinal L-Splines}
\label{sec:cardinal}

\subsection{Cardinal L-Spline Model}
A cardinal spline is a special case of a  non-uniform spline where the knots are located on the integer grid  \cite{unser2010introduction,unser2014unified,unser2014unified2}.
More specifically, a function $x(t)$ is called  a {\em cardinal $\mathrm{L}$-spline} if and only if 
\beq\label{eq:Lspline}
\mathrm{L} x(t) = w(t), 
\eeq
where the operator $\mathrm{L}$ is  continuous domain whitening operator and
and the continuous domain {\em innovation}  signal  $w(t)$ is given by
\beq \label{eq:innovation}
 \quad  w(t):=\sum_{p\in \Zd} a[p] \delta(t-p)  \ ,
 \eeq
 whose singularities are located on integer grid.
 
 Even though the recovery of cardinal L-splines can be considered as special instance of that of non-uniform splines,
  the cardinal setting allows high but finite resolution, so it is closely related to standard compressed sensing framework in discrete framework.
Therefore,  this section provides more detailed discussion of recovery of cardinal L-splines from partial Fourier measurements.
The analysis in this section is significantly influenced by the theory of sparse stochastic processes  \cite{unser2010introduction}, so we follow the original authors's notation.

\subsection{Construction of  Low-Rank Wrap-around Hankel Matrix}

The main advantage of using cardinal setup is that 
 we can recover signals by exploiting the sparseness of {\em discrete innovation} rather than exploiting off the grid singularity.
So, we are now interested in deriving the discrete counterpart of the whitening operator $\mathrm{L}$, which is denoted by $\mathrm{L}_d$:
\begin{equation}\label{eq:ld}
\mathrm{L}_d \delta(t) = \sum_{p \in \Zd} l_d[p] \delta(t-p).
\end{equation}
Now,  by applying the discrete version of whitening operator $\mathrm{L}_d$ to $x(t)$,  we have
\begin{eqnarray}\label{eq:uc}
u_c(t) &:=& \mathrm{L}_d x(t) =  \mathrm{L}_d \mathrm{L}^{-1} w(t) 
= \left(\beta_L \ast w\right)(t) \\
&=& \sum_{p\in \Zd} a[p] \beta_L(t-p) \notag \ .
\end{eqnarray}
where $\beta_L(t) $ denotes  a generalized B-spline associated with the operator $\mathrm{L}$  \cite{unser2010introduction},
which is defined by
\beq
\beta_L(t) =  \mathrm{L}_d \mathrm{L}^{-1}\delta(t)  = {\cal F}^{-1} \left\{ \frac{\sum_{p\in \Zd} l_d[j] e^{-i\omega p}}{\hat l(\omega)}\right\}(t) \  ,
\eeq
where we now use  $\omega=2\pi f$ for Fourier transform to follow the notation in \cite{unser2010introduction}.
As shown in Fig.~\ref{fig:innovation}, $u_c(t)$  is indeed a {\em smoothed} version of continuous domain innovation  $w(t)$ in \eqref{eq:innovation}, because all the sparsity information of the innovation $w(t)$ is encoded in its coefficients
$\{a[p]\}$, and  aside from the interpolant $\beta_L(t)$,  $u_c(t)$ in \eqref{eq:uc} still retains the same coefficients. 
Moreover,  the sparseness of  sampled discrete innovation on the integer grid can be  identified from the discrete samples of  $u_c(t)$:
\begin{eqnarray}
u(t) &:=& u_c(t)\sum_{p\in \Zd} \delta(t-p) \notag\\
&=&  \sum_{p\in \Zd} u_d[p] \delta (t-p)  \label{eq:un}\\
&=&\sum_{p\in \Zd}(a\ast b_L)[p] \delta(t-p)  
\end{eqnarray}
where  
\beq\label{eq:bL}
b_L[p] := \left.\beta_L(t)\right|_{t=p} .
\eeq
To make the discrete sample $u_d[p]$ sparse,  
 the discrete filter $b_L[p]$ should be designed to have the minimum non-zero support.
Due to the relationship \eqref{eq:bL}, this can be achieved if
 $\beta_L(t)$ is maximally localized.
The  associated DFT spectrum of the discrete innovation is given by
\beq\label{eq:uk}
\hat u_d[k]=  \left.  \hat u(\omega) \right|_{\omega = \frac{2\pi k}{n}}=   \sum_{j=0}^{r-1}  u_j e^{-i\frac{2\pi k i_j}{n}}
\eeq
where $\{u_j\}$ denotes the non-zero coefficient of $u[p]$ and $i_j$ refers the corresponding index.

\begin{figure}[!bt]
    \centering
	\centerline{\includegraphics[width=10cm]{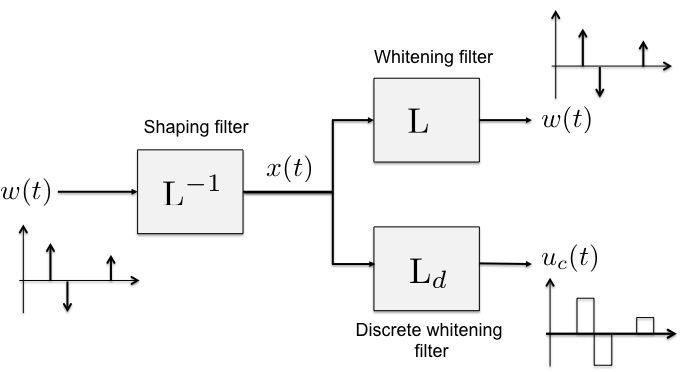}}
	\caption{
	The discrete innovation and continuous domain innovations generated by $\mathrm{L}_d$ and $\mathrm{L}$, respectively.}
	\label{fig:innovation}
\end{figure}

To exploit the sparseness of discrete innovation using the low-rank Hankel matrix, we should relate the discrete innovation to the discrete samples
of the unknown cardinal L-spline $x(t)$. 
This can be done using  an equivalent 
 B-spline representation of $x(t)$  \cite{unser2010introduction}:
\begin{eqnarray}\label{eq:bspline}
x(t) 
&=&  \sum_{p\in \Zd} c[p] \beta_L(t-p)  \  ,
\end{eqnarray}
where  $c[p]$ satisfies
$$a[p] = (c \ast l_d)[p]$$
for $a[p]$ and $l_d[p]$  in \eqref{eq:Lspline} and \eqref{eq:ld}, respectively.
Here, the equivalent B-spline representation in \eqref{eq:bspline} can be shown by: 
\beq
\mathrm{L}x(t) = \sum_{p\in \Zd} c[p] L\beta_L(t-p) = \sum_{l\in \Zd} \underbrace{(c \ast l_d)[p]}_{a[p]} \delta(t-p) \ ,
\eeq
because $\mathrm{L}\mathrm{L}_d \mathrm{L}^{-1} \delta(t) = \mathrm{L}_d \delta(t)$.
So,  we have
\begin{eqnarray}
u(t) 
&=&\sum_{p\in \Zd}(a\ast b_L)[p] \delta(t-p)  \\
&=&\sum_{p\in \Zd}({l_d\ast c} \ast b_L)[p] \delta(t-p)    \\
&=&\sum_{p\in \Zd}(l_d \ast x_d)[p] \delta(t-p)    
\end{eqnarray}
where
\beq\label{eq:xL}
 x_d[p] := \left. x(t)\right|_{t=p} = \sum_{l\in \Zd} c[l] \beta(p-l) = (c\ast b_L)[p] \ .
\eeq
Therefore, $u_d[p]=(l_d \ast x_d)[p]$ and 
the corresponding DFT spectrum is given by
\beq\label{eq:uk2}
\hat u_d[k] =  \hat l_d[k] \hat x_d[k], \quad k=0,\cdots, n-1   . 
\eeq
Because the DFT data $\hat x_d[k]$ can be computed and $\hat l_d[k]$ is known,
we can construct a Hankel matrix $\hank(\hat \ub_d) = \hank(\hat \lb_d\odot \hat \xb_d) \in \Hc(n,d)$.
Thanks to \eqref{eq:uk},   the  associated minimum size annihilating filter $\hat h[k]$  that cancels $\hat x_d[k]$  can be obtained from the following z-transform expression
\begin{eqnarray}
\hat h(z) = \prod_{j=0}^{r-1}(1-e^{-i\frac{2\pi k l_j}{n}} z^{-1})
\end{eqnarray}
whose length is $r+1$.
Therefore, we have
\begin{eqnarray}
\rank \hank(\hat \ub_d) = \rank \hank(\hat \lb_d\odot \hat \xb_d)   = r . 
\end{eqnarray} 
Moreover, due to the periodicity of DFT spectrum, 
we can use the following wrap-around Hankel matrix:
  \begin{eqnarray} \label{eq:Umy}
\hank_c(\hat \ub_d) &=& \left[
        \begin{array}{cccc}
      \hat u_d[0]  & \hat u_d[1] & \cdots   & \hat u_d[d-1]   \\
     \hat u_d[1]  & \hat u_d[2] & \cdots &   \hat u_d[d] \\
         \vdots    & \vdots     &  \ddots    & \vdots    \\
      \hat u_d[n-d]  & \hat u_d[n-d+1] & \cdots & \hat u_d[n-1]\\ \hline 
      \hat u_d[n-d+1]  & \hat u_d[n-d+2] & \cdots & \hat u_d[0] \\
           \vdots    & \vdots     &  \ddots    & \vdots    \\
            \hat u_d[n-1]  & \hat u_d[0] & \cdots & \hat u_d[d-2] \\
        \end{array}
    \right] \in \Cd^{n\times d} \  
    \end{eqnarray}
where the bottom block is an augmented block. 
Since the bottom block can be also annihilated using the same annihilating filter, we can see the rank of the wrap-around Hankel expansion
is the same as the original Hankel structured matrix:
$$\rank\hank_c(\hat \ub_d) =\rank \hank(\hat \ub_d)  = r .$$
Then,  the missing DFT coefficients can be interpolated using the following low-rank matrix completion:
\begin{eqnarray}\label{eq:EMaC3}
 &\min_{\gb\in \Cd^{n} } & \| \hank_c (\gb)\|_* \\
&\mbox{subject to } & P_\Omega(\gb) = P_\Omega(\hat \lb_d \odot \hat \xb_d) \nonumber  \  ,
\end{eqnarray}
or 
\begin{eqnarray}\label{eq:EMaC3_noisy}
 &\min_{\gb\in \Cd^{n} } & \| \hank_c (\gb)\|_* \\
&\mbox{subject to } & \|P_\Omega(\gb) - P_\Omega(\hat \lb_d \odot \hat \yb_d) \|\leq \delta  \nonumber  \  ,
\end{eqnarray}
for noisy DFT data $\hat \yb_d$.
Then, we have the following performance guarantee:
\begin{theorem}
\label{thm:unique_cardinal}
For a given cardinal L-spline $x(t)$ in Eq.~\eqref{eq:Lspline}, let 
 $\hat l_d[k]$ denotes the DFT of discrete whitening operator and
$\hat x_d[k]$ is the DFT of the discrete sample $x_d[p]$ in \eqref{eq:xL}.
Suppose, furthermore,   $d$ is given by $\min\{n-d+1,d\} > r $ and  $\Omega = \{j_1,\ldots,j_m\}$ be a multi-set consisting of random indices
where $j_k$'s are i.i.d. following the uniform distribution on $\{0,\ldots,n-1\}$.
Then,  there exists an absolute constant $c_1$ such that
$\hat\xb$ is the unique minimizer to \eqref{eq:EMaC3} with probability $1 - 1/n^2$, provided
\begin{equation}
\label{eq:samp_comp_cardinal}
m \geq c_1c_s \mu r  \log^2n .
\end{equation}
where $c_s=n/d$ and $\mu$ is the incoherence parameter.
\end{theorem}
\begin{proof}
The associated Hankel matrix has wrap-around property, so the log power factor is reduced to 2, and $c_s=\max\{n/n, n/d\}=n/d$. Q.E.D.
\end{proof}

\begin{theorem}
\label{thm:stability_cardinal}
Suppose  that noisy DFT data $\hat{\yb}_d$ satisfies $\norm{P_\Omega (\hat \lb_d \odot \hat\yb_d - \hat \lb_d \odot \hat{\xb}_d)}_2 \leq \delta$, where $\hat\xb_d$ is noiseless DFT
data $\hat x_d[k]$ of $x_d[p]$ in \eqref{eq:xL}. 
Under the hypotheses of Theorem~\ref{thm:unique_cardinal}, there exists an absolute constant $c_1,c_2$ such that
with probability $1 - 1/n^2$, the solution $\gb$ to \eqref{eq:EMaC3_noisy} satisfies
\[
\norm{\hank (\hat \lb_d \odot \hat{\xb}_d) - \hank(\gb)}_{\mathrm{F}} \leq c_2 n^2 \delta,
\]
provided that \eqref{eq:samp_comp_cardinal} is satisfied with $c_1$.
\end{theorem}

\subsection{Incoherence Condition}

Another advantage of using a cardinal set-up is that the coherence condition can be optimal
 even in the finite sampling regime.
Specifically, due to the the wrap-around property, when $d=n$,
 the singular vectors $U$  (resp. $V$) of $\Hc_c (\xb)$ are composed of $r$ columns of a normalized DFT matrix.
Thus, the standard incoherence condition is 
\begin{eqnarray}\label{eq:mucardinal}
\mu= \max\left\{  \frac{n}{r} \max_{1\leq i \leq n}\|U^*\eb_i\|_2^2,   \frac{n}{r} \max_{1\leq i \leq n}\|V^*\eb_i\|_2^2\right\}  &= & 1 \  ,
\end{eqnarray}
which is optimal. 
Note that compared to the off the grid cases in Section~\ref{sec:incoherence_min}, the optimal mutual coherence can be obtained even with finite $n$.

It is also interesting to see that the corresponding  separation is equal to the Nyquist sampling distance $\Delta = 1/n$, which appears smaller than the minimum separation condition
$2/n$ in Section~\ref{sec:incoherence_min}.
Recall that in off the grid signal reconstruction,  there always exists a limitation in choosing the matrix pencil size $d$ due to  trade-off between the condition number of $\Vc_{n-d+1}$ and $\Vc_d$.
However,  for the cardinal set-up,  thanks to the periodic boundary condition, the limitation does not exist anymore, and the net effect is doubling the effective aperture
size from $n$ to $2n$. This results in the reduction of the minimum separation in the cardinal setup.



\subsection{Regularization Effect in  Cardinal Setup}

Note that in the proposed low-rank interpolation approach for the recovery of general FRI signals,
the weighting factor $\hat l(\omega)$ used in $(P_w)$ or $(P_w')$ is basically a high pass filter that can boost up the noise contribution.
This may limit the performance of the overall low-rank matrix completion algorithm.
In fact, another important advantage of the cardinal setup is to provide a natural regularization.
More specifically, in constructing the weighting matrix  for  the low-rank matrix completion problem, instead of using the spectrum of the continuous domain
whitening operator $\mathrm{L}$,  we should use  $\hat l_d(\omega)$  of the discrete counterpart $\mathrm{L}_d$.
As will be shown in the following examples, this helps to limit the noise amplification in the associated low-rank matrix completion problem.

\begin{figure}[!bt]
    \centering
	\centerline{\includegraphics[width=8cm]{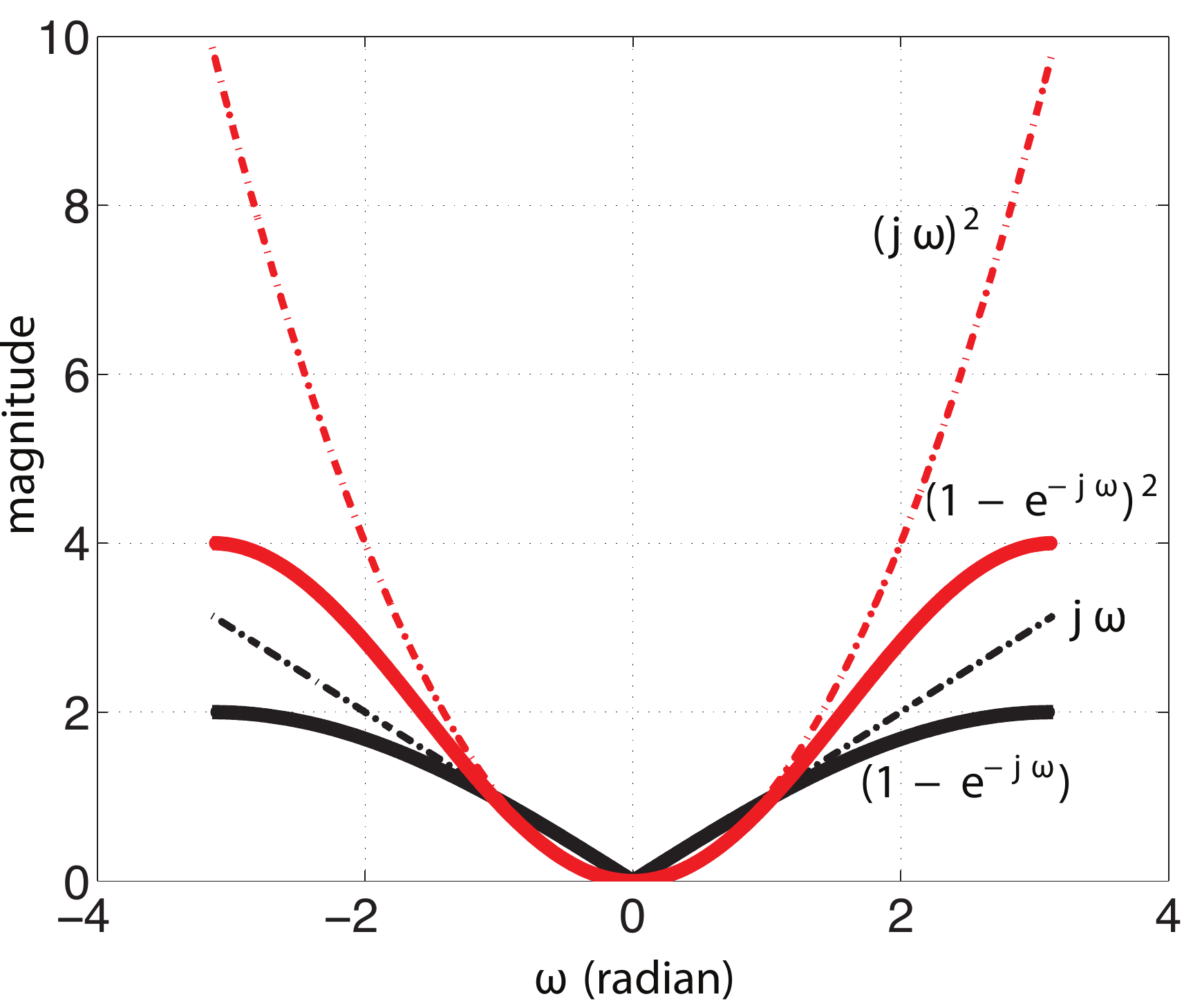}}
	\caption{
	Comparison of first- and second-order weights from whitening operator $\mathrm{L}$ and discrete counter-part $\mathrm{L}_d$.}
	\label{fig:ld}
\end{figure}

\subsubsection{Signals with Total Variation}

A signal with  total variation can be considered as a special case of \eqref{eq:Lspline} with 
$\mathrm{L} = \frac{d}{dt}$.  Then, the discrete whitening operator  $\mathrm{L}_d$ is the finite difference operator $\mathrm{D}_d$ given by
$$\mathrm{D}_d x(t) =  x(t)-x(t-1).$$
In this case, the associated L-spline is given by
\beq
\beta_L(t) = \beta_+^0(t) =  \Fc^{-1} \left\{ \frac{1-e^{-i\omega}}{i\omega} \right\}(t) = \begin{cases} 1, & \mbox{for $0\leq t<1$}\\ 0, &\mbox{otherwise} \end{cases}
\eeq
Note that this is maximally localized spline because
$b_L[p]=  \left. \beta_L(t)\right|_{t=p}= \delta[p]$ is a single-tap filter.
Therefore,  the sparsity level of the discrete innovation is equal to the number of underlying Diracs.
Moreover, the weighting function for the low-rank matrix completion problem  is given by
$$\hat l_d(\omega)  =  1-e^{-i\omega}.$$
Figure~\ref{fig:ld} compared the weighting functions that corresponds to the original whitening operator  $\hat l(\omega)=i\omega$ and the discrete counterpart $\hat l_d(\omega) = 1-e^{-i\omega}.$
We can clearly see that  high frequency boosting is reduced by the discrete whitening operator, which makes the low-rank matrix completion much more robust.

\subsubsection{Signals with Higher order Total Variation}

Consider a signal $x(t)$ that is represented by \eqref{eq:Lspline} with 
$\mathrm{L} = \frac{d^{m+1}}{dt^{m+1}}$.
Then, the corresponding discrete counterpart $\mathrm{L}_d$ should be constructed by 
$$\mathrm{L}_d \delta(t) = \underbrace{\mathrm{D}_d \mathrm{D}_d \cdots \mathrm{D}_d}_{m}\delta(t) \  . $$
In this case, the associated L-spline is given by \cite{unser2010introduction}
\begin{eqnarray}
\beta_+^m(t)  &=&  \underbrace{\left(\beta_+^0 \ast \beta_+^0  \ast \cdots \ast \beta_+^0\right)}_{m+1}(t)  \notag\\
&=&  \Fc^{-1} \left\{ \left(\frac{1-e^{-i\omega}}{i\omega}\right)^{m+1} \right\}(t) = \sum_{k=0}^{m+1} (-1)^k   \binom{m+1}{k} \frac{(t-k)^m_+}{m!}
\end{eqnarray}
with $(t)_+ = \max(0,t)$.
We can see that the length of the corresponding filter $b_L[n]$ is now given by $ m+1$.
Hence, when the underlying signal is $r$-Diracs, then the sparsity level of the discrete innovation is upper bounded by
\begin{equation}\label{eq:sparsitydouble}
(m+1) r
\end{equation}
 and the corresponding  weighting function is given by
\beq\label{eq:mld}
\hat l_d(\omega)  =  (1-e^{-i\omega})^{m+1}.
\eeq
Again, Figure~\ref{fig:ld} clearly showed that this weighting function is much more robust against noises compared to the original weighting $(i\omega)^{m+1}$.

Note that the relationship between the sparsity in \eqref{eq:sparsitydouble} and the noise reduction by \eqref{eq:mld}   clearly demonstrate the trade-off between regularization and the 
resolution in signal recovery.
Specifically, to recover high order splines,  rather than imposing the higher order weighting that is prone to noise boosting,
we can use regularised weighting  \eqref{eq:mld}  that comes from discrete whitening operator.
The catch, though,  is the necessity for additional spectral samples originated from the sparsity increase.

\subsection{Recovery of Continuous Domain Signals After Interpolation}

%

In contrast to the recovery of general FRI signals from its spectral measurements,  the reconstruction of cardinal L-spline
can be done using standard B-spline signal processing tools \cite{unser1993a,unser1993b}.
Specifically,
after recovering the DFT spectrum $\hat x[k]$ using the Hankel structured matrix completion,  a trivial application of an inverse DFT can obtain $x_d[n]$.
Then, to recover $x(t)$, we use the equivalent representation Eq.~\eqref{eq:bspline}.
More specifically, the coefficient $c[n]$ in \eqref{eq:bspline} can be computed by \eqref{eq:xL}:
$$x_d[n] = (c \ast b_L)[n].$$
Because $x_d[n]$ are already computed and $b_L[n]$ is known, the unknown coefficient $c[n]$ can be obtained using the
standard method in \cite{unser1993a,unser1993b} using recursive filtering without computationally expensive matrix inversion.
In case the operator $\mathrm{L}$ is the first differentiation, $b_L[n]=\delta[n]$, so  $c[n]$ can be readily obtained as $x_d[n]$.

\section{Algorithm Implementation}
\label{sec:algorithm}

\subsection{Noiseless  structured matrix completion algorithm} \label{sec:reconstruction} 	

In order to solve  structured matrix completion problem from noise free measurements,
we employ an SVD-free structured rank minimization algorithm \cite{signoretto2013svd} with an initialization using the
 low-rank factorization model (LMaFit) algorithm \cite{wen2012solving}.  This algorithm does not use the singular value decomposition (SVD), so the computational complexity can be significantly reduced. %
Specifically, the algorithm is based on the following observation \cite{signoretto2013svd}:
 \begin{eqnarray}\label{eq:relaxation_nuclear}
 \|A\|_* = \min\limits_{U,V: A=UV^H} \|U\|_F^2+ \|V\|_F^2 \quad \ .
 \end{eqnarray}
Hence,   it can be reformulated as 
 the nuclear norm minimization problem under  the matrix factorization constraint:
\begin{eqnarray}
\min_{U,V:\hank(\gb)=UV^H}  && \|U\|_F^2+ \|V\|_F^2 \nonumber \\
\mbox{subject to}
&&P_\Omega(\gb)=  P_\Omega(\hat \xb)
 . \label{eq:data}
\end{eqnarray}
By combining the two constraints, we have the following cost function for   an alternating direction method of multiplier  (ADMM) step \cite{boyd2011distributed}:
\begin{eqnarray}\label{eq:ADMM}
L(U,V,\gb,\Lambda) & := & \iota(\gb) + \frac{1}{2} \left( \|U\|_F^2+\|V\|_F^2\right)  \nonumber\\
&&+  \frac{\mu}{2}  \|\hank(\gb)- UV^H+\Lambda\|^2_F
\end{eqnarray}
where $\iota(\gb)$  denotes an indicator function:
$$\iota(\gb) = \left\{\begin{array}{ll} 0, & \mbox{if $P_\Omega(\gb)=  P_\Omega(\hat \xb)$} \\ \infty, & \mbox{otherwise} \end{array} \right.  \  .$$
One of the advantages of the ADMM formulation is that each subproblem is simply obtained from \eqref{eq:ADMM}. More specifically,
 $\gb^{(n+1)},  U^{(n+1)}$ and  $V^{(n+1)}$ can be obtained, respectively, by applying the following optimization problems sequentially:
 \begin{equation}\label{eq:update}
 \begin{array}{ll}
\gb^{(n+1)} =& \arg
 \min_\gb  \iota(\gb)  +  \frac{\mu}{2}  \|\hank(\gb)- U^{(n)}V^{(n)H}+\Lambda^{(n)}\|^2_F \\
 U^{(n+1)} =& \arg
 \min_U  \frac{1}{2} \|U\|_F^2 + \frac{\mu}{2}  \|\hank(\gb^{(n+1)})- UV^{(n)H}+\Lambda^{(n)}\|^2_F \\ 
  V^{(n+1)} =&\arg
  \min_V  \frac{1}{2} \|V\|_F^2 + \frac{\mu}{2}  \|\hank(\gb^{(n+1)})- U^{(n+1)}V^{H}+\Lambda^{(n)}\|^2_F 
  \end{array}
  \end{equation}
  and the Lagrangian update is given by
  \begin{eqnarray}\label{eq:lagrangian}
  \Lambda^{(n+1)} =& \Yc^{(n+1)}- U^{(n+1)}V^{(n+1)H}+\Lambda^{(n)} \ ,
 \end{eqnarray}
 where $\Yc^{(n+1)}=\hank(\gb^{(n+1)})$.
 It is easy to show that the first step in \eqref{eq:update} can be reduced to
 \begin{eqnarray}\label{eq:gb}
 \gb^{(n+1)}  = 
 P_{\Omega^c}\hank^{\dag}\left\{  U^{(n)}V^{(n)H}- \Lambda^{(n)} \right\} + P_\Omega (\hat \xb) ,
 \end{eqnarray}
 where $P_{\Omega^c}$ is a projection mapping on the set $\Omega^c$  (the complement set of $\Omega$) and
  $\hank^{\dag}$ corresponds to the Penrose-Moore pseudo-inverse mapping from our structured matrix to a vector.
Hence, the role of the pseudo-inverse is taking the average value and putting it back
to the original coordinate. 
Next, the subproblem for $U$ and $V$ can be easily calculated by taking the derivative with respect to each matrix, and we have
  \begin{equation}\label{eq:UV}
  \begin{array}{l}
 U^{(n+1)} = \mu \left(\Yc^{(n+1)}+\Lambda^{(n)}\right)V^{(n)} \left(I+\mu V^{(n)H}V^{(n)}\right)^{-1}    \\
  V^{(n+1)} = \mu\left(\Yc^{(n+1)}+\Lambda^{(n)}\right)^HU^{(n+1)} \left(I+\mu U^{(n+1)H}U^{(n+1)}\right)^{-1}   
 \end{array}
 \end{equation}

  Now, for faster convergence, the remaining issue is how to initialize $U$ and $V$. For this, we employ an algorithm called the low-rank factorization model (LMaFit) \cite{wen2012solving}.
 More specifically, for a low-rank matrix $Z$, LMaFit solves the following optimization problem:
\begin{equation}\label{eq:lmafit}
\min_{U,V,Z} \frac{1}{2} \|UV^H-Z\|^2_F~  \mbox{subject to } P_I(Z) = P_I(\hank(\hat \xb))
\end{equation}
and $Z$ is initialized with $\hank(\hat \xb)$ and the index set $I$ denotes the positions where the elements of $\hank(\hat \xb)$ are known.
LMaFit solves a linear equation with respect to  $U$ and $V$ to find their updates and relaxes the updates by taking the average between the previous iteration and the current iteration.
Moreover, the rank estimation can be done automatically. LMaFit uses QR factorization instead of SVD, so it is also computationally efficient.

\subsection{Noisy  structured matrix completion algorithm}

Similarly, the noisy matrix completion problem  can be solved by mininimizing the following Lagrangian function:
\begin{eqnarray}
L(U,V,\gb,\Lambda) & := & \frac{\lambda}{2} \|P_\Omega(\hat\yb) - P_\Omega(\gb) \|_2^2 + \frac{1}{2} \left( \|U\|_F^2+\|V\|_F^2\right)  \nonumber\\
&&+  \frac{\mu}{2}  \|\hank(\gb)- UV^H+\Lambda\|^2_F
\end{eqnarray}
where $\lambda$ denotes an appropriate regularization parameter. Compared to the noiseless cases, the  only difference is the update step of $\gb$. 
More specifically, we have
 \begin{eqnarray}
\gb^{(n+1)} =& \arg
 \min_\gb  \frac{\lambda}{2} \|P_\Omega(\hat\yb) - P_\Omega(\gb) \|_2^2  +  \frac{\mu}{2}  \|\hank(\gb)- U^{(n)}V^{(n)H}+\Lambda^{(n)}\|^2_F 
\end{eqnarray} 
which can be reduced to 
 \begin{eqnarray}
 \gb^{(n+1)}  = 
 P_{\Omega^c}\hank^{\dag}\left\{  U^{(n)}V^{(n)H}- \Lambda^{(n)} \right\} +  P_\Omega(\zb)  
 \end{eqnarray} 
 where  $\zb=[z[0],\cdots, z[n-1]]^T$ such that
 \begin{eqnarray}
 z[i]   = 
 \frac{\lambda y[i]+ \mu P_i\left(\hank^*\left(U^{(n)}V^{(n)H}- \Lambda^{(n)} \right)\right)}{\lambda + \mu P_i(\hank^*\hank(\eb_i))},  
 \end{eqnarray}
 where $\eb_i$ denotes the unit coordinate vector where the $i$-th element is 1, and $P_i$ is the projection operator to the $i$-th coordinate.

\subsection{Implemetation of ADMM}


The alternating direction method of multipliers (ADMM) described above is widely used to solve
large-scale linearly constrained optimization problems, convex or nonconvex, in many engineering
fields.
The convergence of ADMM algorithm for minimizing the sum of two or more nonsmooth convex separable functions have been well-studied,
and Hong and Luo \cite{hong2012linear} proved the linear convergence of a general ADMM algorithm with any number of blocks under linear constraints.
Even for the nonconvex problems, Hong et al \cite{hong2016convergence} further showed that   the classical ADMM converges to
the set of stationary solutions, provided that the penalty parameter  in ADMM  ($\mu$  in our case)
is chosen to be sufficiently large. 
Accordingly, to ensure the convergence of ADMM,  it is usually recommended to  use a sufficiently large $\mu$; so, we  chose $\mu=10^3$ in our implementation.

Note that the computational complexity of our ADMM algorithm is crucially determined by the matrix inversion in \eqref{eq:UV}.
  More specifically,  the computational complexity  in
   \eqref{eq:UV} in terms of multiplication is
   $\mathcal{O}((n-d+1)rd+r^3)$, whereas
 the number of multiplication required for \eqref{eq:gb} and \eqref{eq:lagrangian} is $\mathcal{O}((n-d+1)rd)$.
Thus, if the underlying signal is sparse, then
we can choose sufficiently small rank estimate $r$ and the matrix pencil size $d$ to reduce the  overall computational complexity.
Another  important issue in practice is the memory usage. Note that the $ U, V, \Yc$ as well as the Lagrangian parameter $\Lambda$ should be stored throughout the iterations
in our ADMM implementation.  The associated memory requirement is at least $(n-d+1)r+rd+2(n-d+1)d$. 
This is not an issue in our 1-D problems, but for  large size problems (especially originated for three dimensional recovery problems in medical imaging applications),
the memory requirement quickly grows, which can become a dominating computational bottleneck  in parallel implementation using memory limited graphic processor unit (GPU).

\section{Numerical Results}
\label{sec:result}

In this section,
we first perform comparative numerical study for recovering FRI signals on integer grid.
Then, we provide numerical experiments of reconstructing piecewise polynomials where the discontinuities are located in arbitrary positions.

\subsection{Recovery of On-grid Signals}

First,  we perform numerical simulations using noiseless measurements.
Specifically, we consider three scenario: 1) streams of Diracs,
 2) piecewise constant signals, and 3) super-position of Diracs and piecewise constant signals.  
  As a reference for comparsion,  the SPGL1 
   implementation of the basis pursuit (BP) algorithm  \cite{van2008probing} was used  for recovering a stream of Diracs,
whereas  the split Bregman method of $l_1$ total variation reconstruction \cite{goldstein2009split} was used  for recovering  signals in 2) and 3) scenario.
We assume that all  the singularities are located on integer grid.
 To quantify recovery performances, phase transition plots were calculated
 using  300 Monte Carlo runs.

\subsubsection{Diracs streams}

To simulate Diracs stream signals, we generated  one-dimensional vectors with the length of $100$,
where the location of Diracs are constrained on integer grid.
The spectral measurements were randomly sampled  with uniform random distribution, where the zero frequency component was always included.
This made the Fourier sensing matrix become a DFT matrix, so we can use basis pursuit using partial DFT sensing matrix.
We used the SPGL1 basis pursuit algorithm which was obtained from the original author's homepage \cite{van2008probing}. Only thing we need to set for SPGL1 was the number of iteration to 500.
For the proposed method,  
 $d$ was set to be $\lfloor n/2\rfloor+1=51$. The  other hyper-parameters for the proposed method were as following: $\mu=10^{3}$, $500$ iterations, $tol=10^{-4}$ for LMaFit.
For a fair comparison, we used the same iteration numbers and sampling patterns for  both  basis pursuit and the proposed algorithm. The phase transitions 
in  Fig.~\ref{fig:deltas} show the success ratio calculated from 300 Monte Carlo trials. Each trial from Monte Carlo simulations was considered as a success when the normalized mean square error (NMSE) is below $10^{-3}$. 
In Fig. \ref{fig:deltas}, the proposed approach provided a sharper transition curve between success and failure than that of the basis pursuit. Furthermore, a transition curve of the proposed method
(red dotted line) is higher than that of basis pursuit (blue dotted line).  

\begin{figure}[!t]
    \centering
	\centerline{\includegraphics[trim = 0mm 60mm 0mm 50mm,width=7 in]{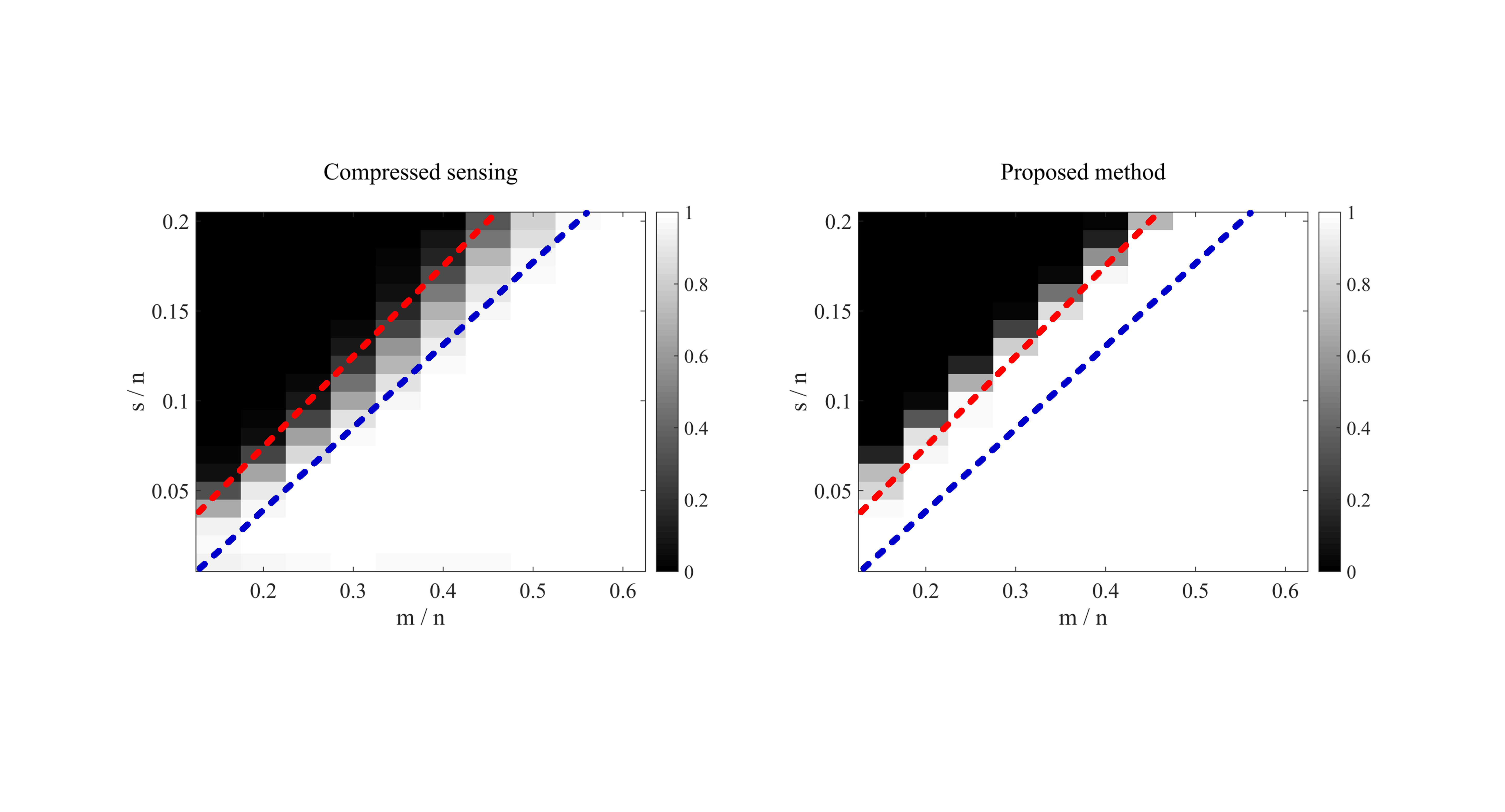}}
	\caption{
	Phase transition diagrams  for recovering stream of Diracs   from $m$ random sampled Fourier samples. The size of target signal ($n$) is $100$ and the annihilating filter size 
	$d$ was set to be $51$.  $s$ denotes the number of Diracs. The left and right graphs correspond to the phase transition diagram of the basis pursuit \cite{van2008probing} compressed sensing approach and the proposed low-rank interpolation approach,  respectively. The success ratio is obtained by the success ratio from 300 Monte Carlo runs. Two transition lines from compressed sensing (blue) and low-rank interpolator (red) are overlaid. }
	\label{fig:deltas}
\end{figure}


\subsubsection{Piecewise-constant signals}

To generate the piecewise constant signals, we first generated Diracs signal at random  locations on integer grid and added steps in between the Diracs.
The length of the unknown one-dimensional vector was again set to  $100$. To avoid boundary effect, the values at the end of both boundaries were set to zeros.
As a  conventional compressed sensing approach, we employed the  1-D version of $l_1$-total variation reconstruction  ($l_1$-TV) using the  split Bregman  method \cite{goldstein2009split}, which was modified from the original  2-D  version of $l_1$-TV  from author's homepage. 
We found that the optimal parameters for $l_1$-TV were $\mu=10^2,\lambda=1$, and the outer and inner loop iterations  of 5 and 40, respectively. 
The  hyper-parameters for the proposed method are as follows; $\mu=10^{3}$, $200$ iterations, $tol=10^{-3}$ for LMaFit.
Note that  we need $1-e^{-i\omega}$ weighting for low-rank Hankel matrix completion as a discrete whitening operator for TV signals.
The phase transition plots were calculated using  averaged success ratio  from 300 Monte Carlo trials. Each trial from Monte Carlo simulations was  considered as a success when the NMSE was below $10^{-2}$.
Because the actual non-zero support of the piecewise constant signals was basically entire domain,  the threshold was set larger than the previous Dirac experiments.

As shown in Fig. \ref{fig:pw}, a transition curve from the proposed method (red dotted line) provided a sharper and improved transition than $l_1$ total variation approach (blue dotted line). Furthermore, even in the area of success status, there were some  unsuccessful recoveries for the case of conventional method,
whereas
the proposed method succeeded nearly all the times. 
 In Fig. \ref{fig:single},  we also illustrated  sample recovery results from the same locations in the phase transition diagram, which are at the yellow star marked position in Fig.~\ref{fig:pw}.
We observed near perfect reconstruction from the proposed method, whereas severe blurring was observed in $l_1$-TV reconstruction.

\begin{figure}[!htb]
    \centering
	\centerline{\includegraphics[trim = 0mm 60mm 0mm 50mm,clip=true,width=7 in]{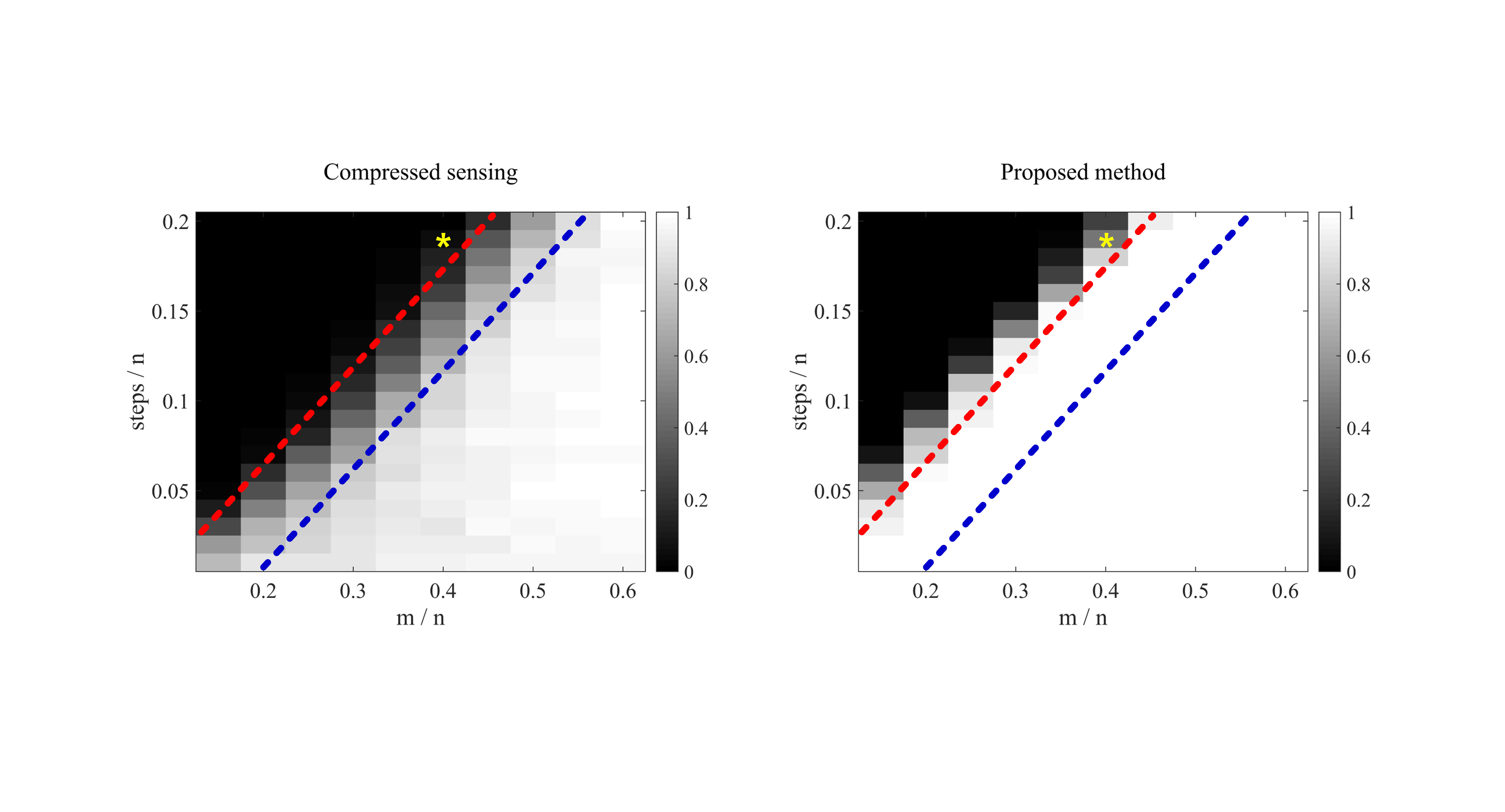}}
	\caption{
	Phase transition diagrams  for piecewise constant signals   from $m$ random sampled Fourier samples. The size of target signal ($n$) is $100$ and the annihilating filter size 
	$d$ was set to be $51$.   The left and right graphs correspond to the phase transition diagram of the $l_1$-TV  compressed sensing approach and the proposed low-rank interpolation approach,  respectively. The success ratio is obtained by the success ratio from 300 Monte Carlo runs. Two transition lines from compressed sensing (blue) and low-rank interpolator (red) are overlaid.  }
	\label{fig:pw}
\end{figure}

\begin{figure}[!htb]
    \centering
	\centerline{\includegraphics[width=8cm]{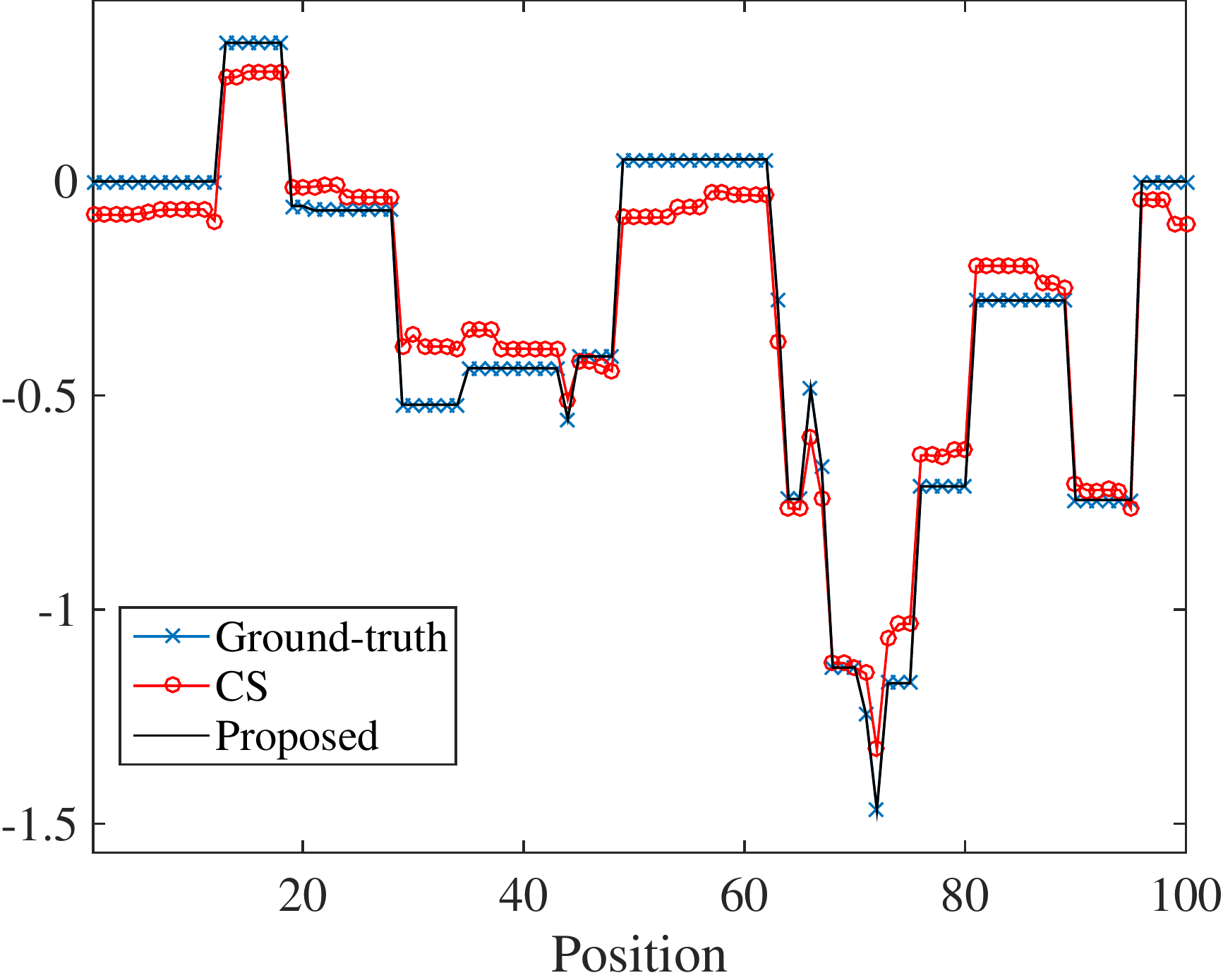}}
	\caption{
	Sample reconstruction results at the yellow star position in Fig.~\ref{fig:pw}.   
	Ground-truth signal (original),  
	$l_1$-TV  (compressed sensing) and the proposed method (low-rank interpolator) were illustrated.
	The parameters for the experiments are:  $n=100, d=51,m=40$ and the number of steps was 19.
}
	\label{fig:single}
\end{figure}

\subsubsection{Piecewise-constant signal + Diracs}

\begin{figure}[!htb]
    \centering
	\centerline{\includegraphics[trim = 0mm 60mm 0mm 50mm,clip=true,width=7 in]{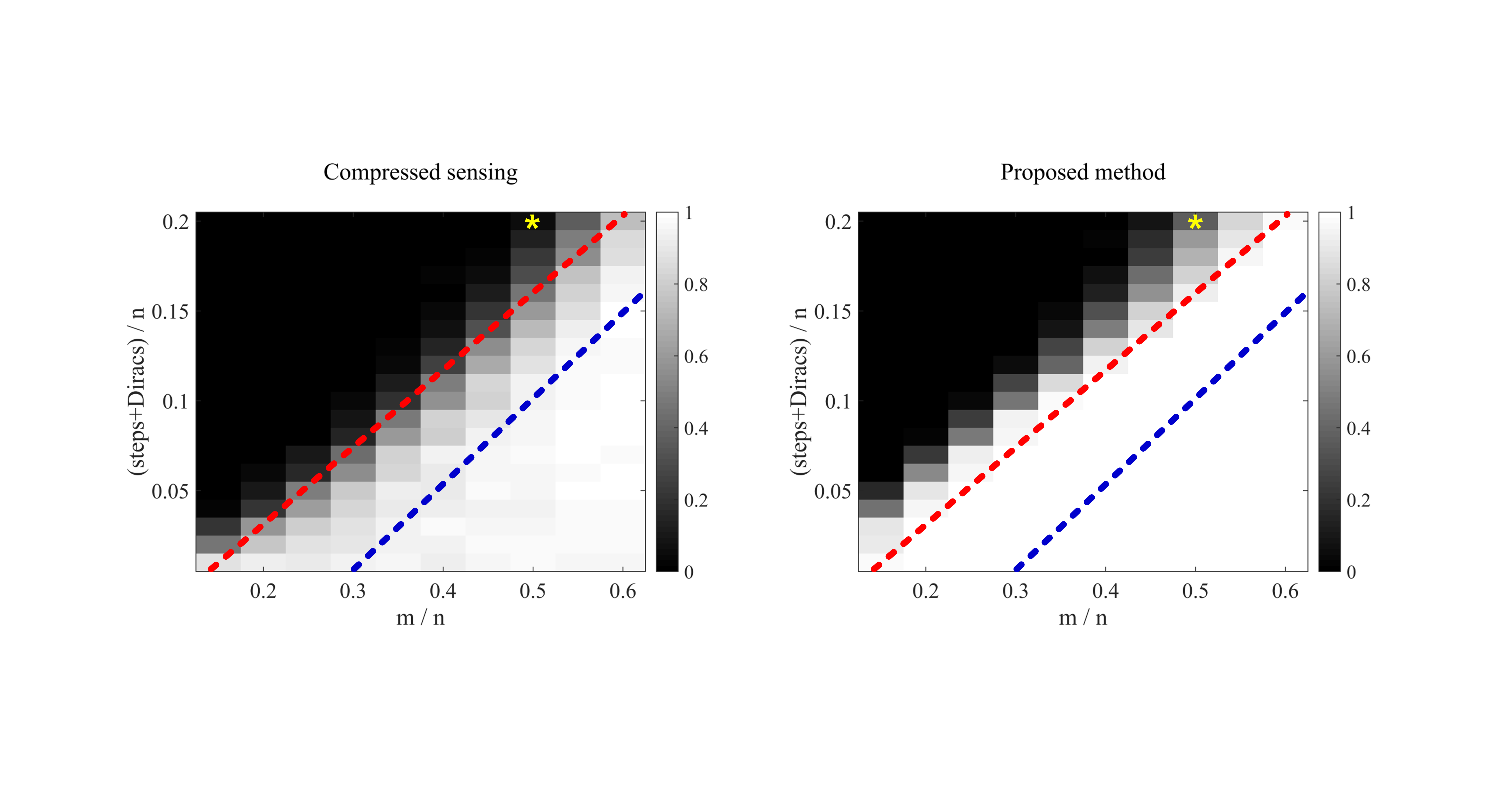}}
	\caption{
	Phase transition diagrams  for recovering super-position of piecewise constant signal and Diracs  from $m$ random sampled Fourier samples. The size of target signal ($n$) is $100$ and the annihilating filter size 
	$d$ was set to be $51$.   The left and right graphs correspond to the phase transition diagram of the $l_1$-TV compressed sensing approach and the proposed low-rank interpolation approach,  respectively. The success ratio is obtained by the success ratio from 300 Monte Carlo runs. Two transition lines from compressed sensing (blue) and low-rank interpolator (red) are overlaid.  }
	\label{fig:pw_pts}
\end{figure}

\begin{figure}[!htb]
    \centering
	\centerline{\includegraphics[width=8cm]{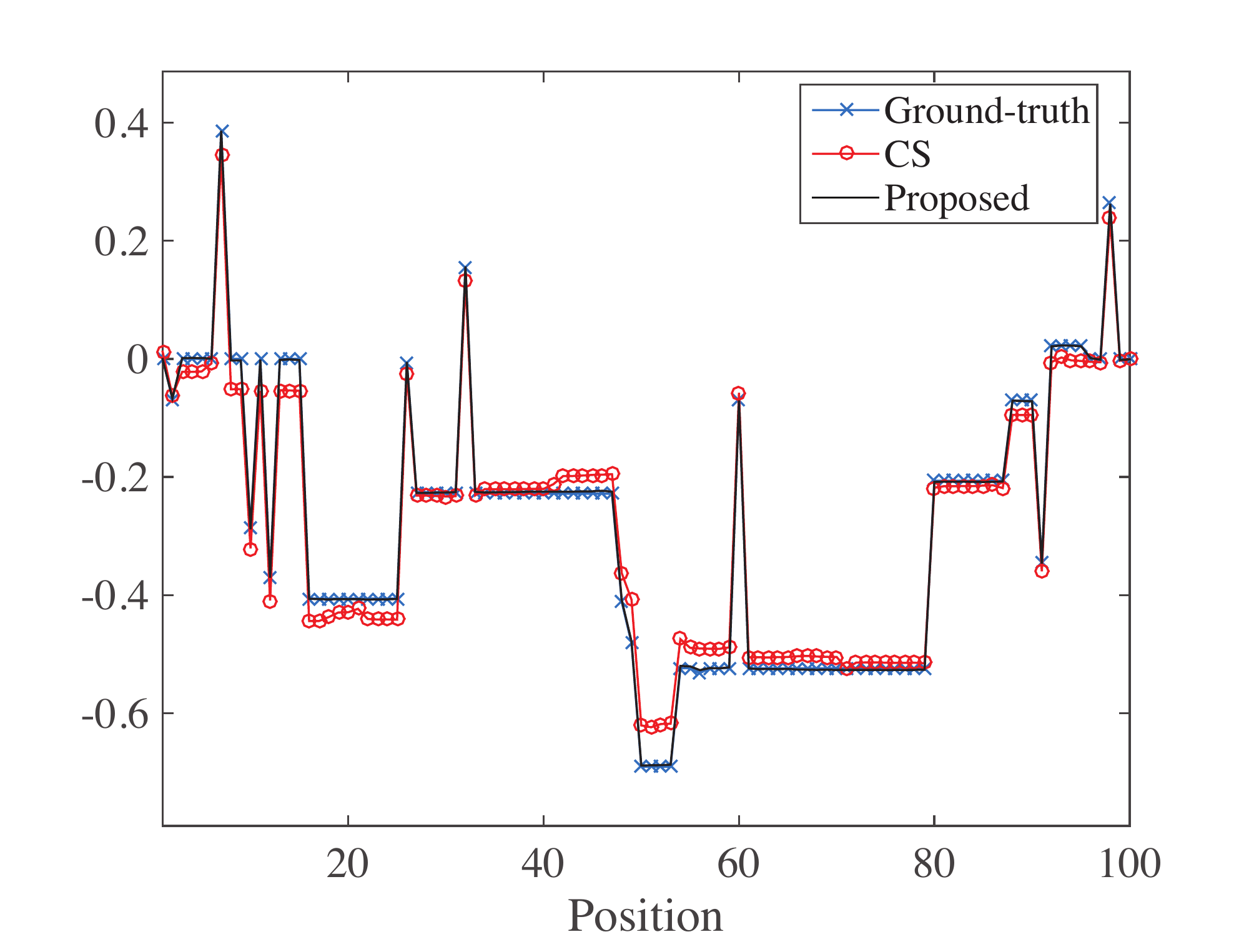}}
	\caption{
	Sample reconstruction results at the yellow star position in Fig.~\ref{fig:pw_pts}.   Ground-truth signal (original), 
	$l_1$-TV  (compressed sensing) and the proposed method (low-rank interpolator) were illustrated.  The parameters for the experiments are:  $n=100, d=51,m=50$ and the number of steps and Diracs were all 10.}
	\label{fig:single_pw_pts}
\end{figure}

We performed  additional experiments for reconstruction of  a super-position of piecewise constant signal and  Dirac spikes. 
Note that this corresponds to the first derivative of piecewise polynomial with maximum order of 1 (i.e. piecewise constant and linear signals).
The goal of this experiment was to verify the capability of recovering piecewise polynomials, but there was no widely used 
compressed sensing solution for this type of signals;  so for a fair comparison, we were interested in recovering their derivatives, because
 the conventional $l_1$-TV approach can be still used for recovering Diracs and piecewise constant signals. The optimal parameters for $l_1$-TV were $\mu=10^3,\lambda=1$, and outer and inner loop iterations  of 5 and 40, respectively. The hyper-parameters for the proposed method are as following: $\mu=10^{3}$, $200$ iterations, $tol=10^{-3}$ for LMaFit. 
In this case,  the sparsity level doubles at the Dirac locations when we use $l_1$-TV method for this type of signals. 
Similar sparsity doubling was observed in our approach.
More specifically,
  our method required the derivative operator as a whitening operator, which resulted in the first derivative of Diracs.
According to \eqref{eq:sparsitydouble},  this makes the the effective sparsity level doubled.
Accordingly, the comparison of $l_1$-TV and our low-rank interpolation approach was fair, and
the overall phase transition were expected to be inferior compared to those of piecewise constant signals. 
The simulation environment was set to be same as those of the previous piecewise constant setup except for  the signal generation. For signals, we generated equal number of  steps and Diracs. When the sparisty is an odd number, the numbers of Diracs was set to the number of steps minus 1.

As shown in Fig. \ref{fig:pw_pts}, there were  much more significant differences between the two approaches. Our algorithm still provided very clear and improved phase transition, whereas
the conventional $l_1$-TV approach resulted in a very fuzzy and inferior phase transition. 
In Fig. \ref{fig:single_pw_pts},  we also illustrated  sample recovery results from the same locations in the phase transition diagram, which are at the yellow star marked position in Fig.~\ref{fig:pw_pts}.
The proposed approach provided a near perfect reconstruction, whereas $l_1$-TV reconstruction exhibits blurrings.
This again confirms the effectiveness of our approach.

\subsubsection{Recovery from Noisy Measurements}

To verify the noise robustness, 
we  performed experiments using  piecewise constant signals by adding  the additive complex Gaussian noise to partial Fourier measurements. 
%
Fig.~\ref{fig:single_pw_pts2}(a) showed the recovery performance of the proposed low-rank interpolation method at   several signal to noise ratios (SNR).
%
All setting parameters for the proposed method are same with parameters of previous experiments except  the addition of $\lambda=10^5$. 
As expected from the theoretical results, the recovery performance was proportional to the noise level.
Fig.~\ref{fig:single_pw_pts2}(b) illustrated the example of reconstructions from 30dB noisy meausurements.  
Here, the optimal parameters for $l_1$-TV (CS) were $\mu=10^3,\lambda=1$, and outer-inner loop iterations  of 5 and 40, respectively.
The proposed approaches still provide accurate reconstruction result. 


\begin{figure}[!htb]
    \centering
	\centerline{\includegraphics[height=6cm]{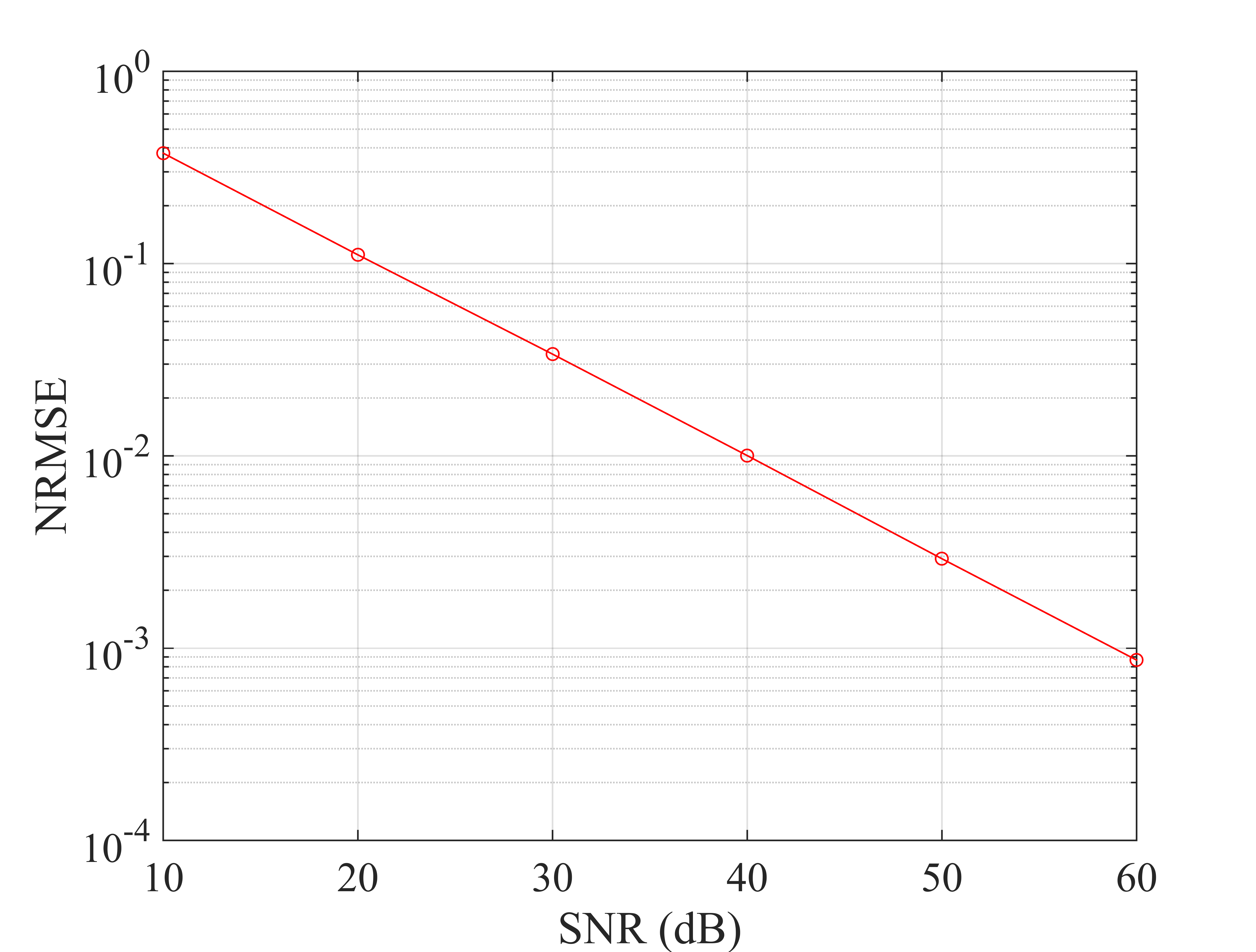}\includegraphics[height=6cm]{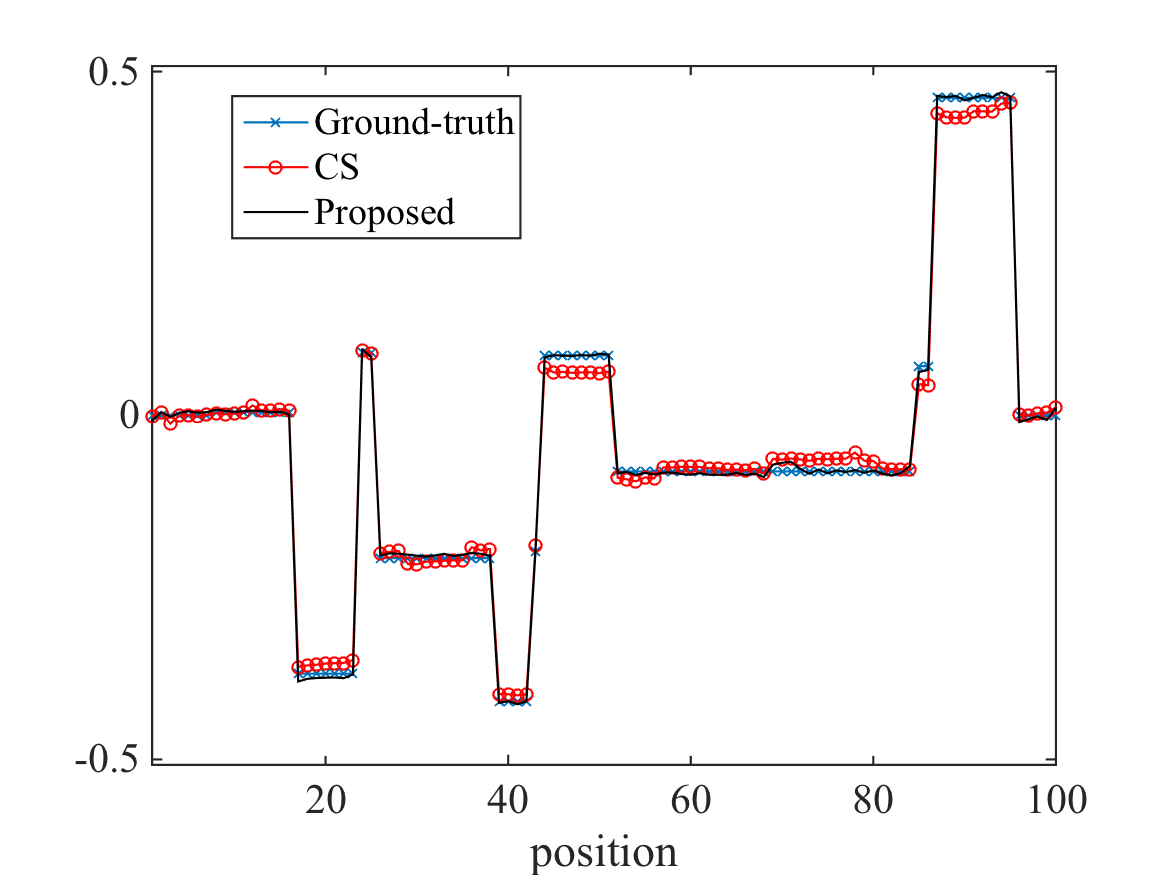}}
	\centerline{\mbox{(a)}\hspace{6cm}\mbox{(b)}}
	\caption{(a)
	The reconstruction NMSE plots  by the proposed low-rank interpolation scheme at various SNR values. For this simulation, 
	we set the annihilating filter size $d=51$, $n=100$,  the number of singularity due to steps was 10, and the number of measurements was 50.  
	(b) A reconstruction example from 30dB noisy measurement.}
	\label{fig:single_pw_pts2}
\end{figure}

\subsection{Recovery of  Off the Grid Signals}

To verify the performance of  off the grid signal recovery,
we  also performed additional experiments using piecewise constant signals  whose edges are located  on a continuous domain between $[0,1]$.
Specifically, we consider a signal composed of several rectangles whose edges are located on off-grid positions. 
Because there exists closed-form expression of  Fourier transform
of shifted rectangle functions, the measurement data could be generated accurately without using discrete Fourier transform.
%
%
Since the signal is composed of  rectangles, the  singularities are located at the edge position after the differentiation. 
Accordingly, 
the weighting
factor for low-rank interpolation was the spectrum of the continuous domain derivative, i.e. $\hat l(\omega)=i\omega$. 
Owing to the Nyquist criterion, the sampling grid in the Fourier domain corresponds to an integer grid, and 
the ambient signal dimension $n$  (which corresponds to the aperture size) was set to 100.
Then,  $m=36$ Fourier samples were randomly obtained from $[0,\cdots, 99]$. 
The parameters for the proposed low-rank interpolation
were as following: the tolerance for LMaFit = $10^{-1}$, number of iterations 300, $\mu = 10^3$.
Once the Fourier data were interpolated,
we used the matrix pencil method \cite{hua1990matrix,sarkar1995using} as described in Section~\ref{sec:matrix_pencil} for the second step  of signal recovery.

Fig. \ref{fig:offgrid_pw_pts2}(b) illustrated the interpolated measurement from 36 irregularly sampled Fourier data using the proposed low-rank interpolation. Because
the underlying signal is piecewise constant signal,  the optimal weighting $\hat l(\omega)=i\omega$ was applied for the Fourier data before the low-rank matrix completion 
was applied.  As shown in Fig. \ref{fig:offgrid_pw_pts2}(b),  near perfect interpolation was achieved.
On the other hand,  if the original Fourier data was used for the low-rank interpolation without the weighting,  then  the resulting interpolation was very different from the true
measurement (see Fig. \ref{fig:offgrid_pw_pts2}(a)). The results confirmed our theory.

 \begin{figure}[!htb]
    \centering
	\centerline{\includegraphics[width=16cm]{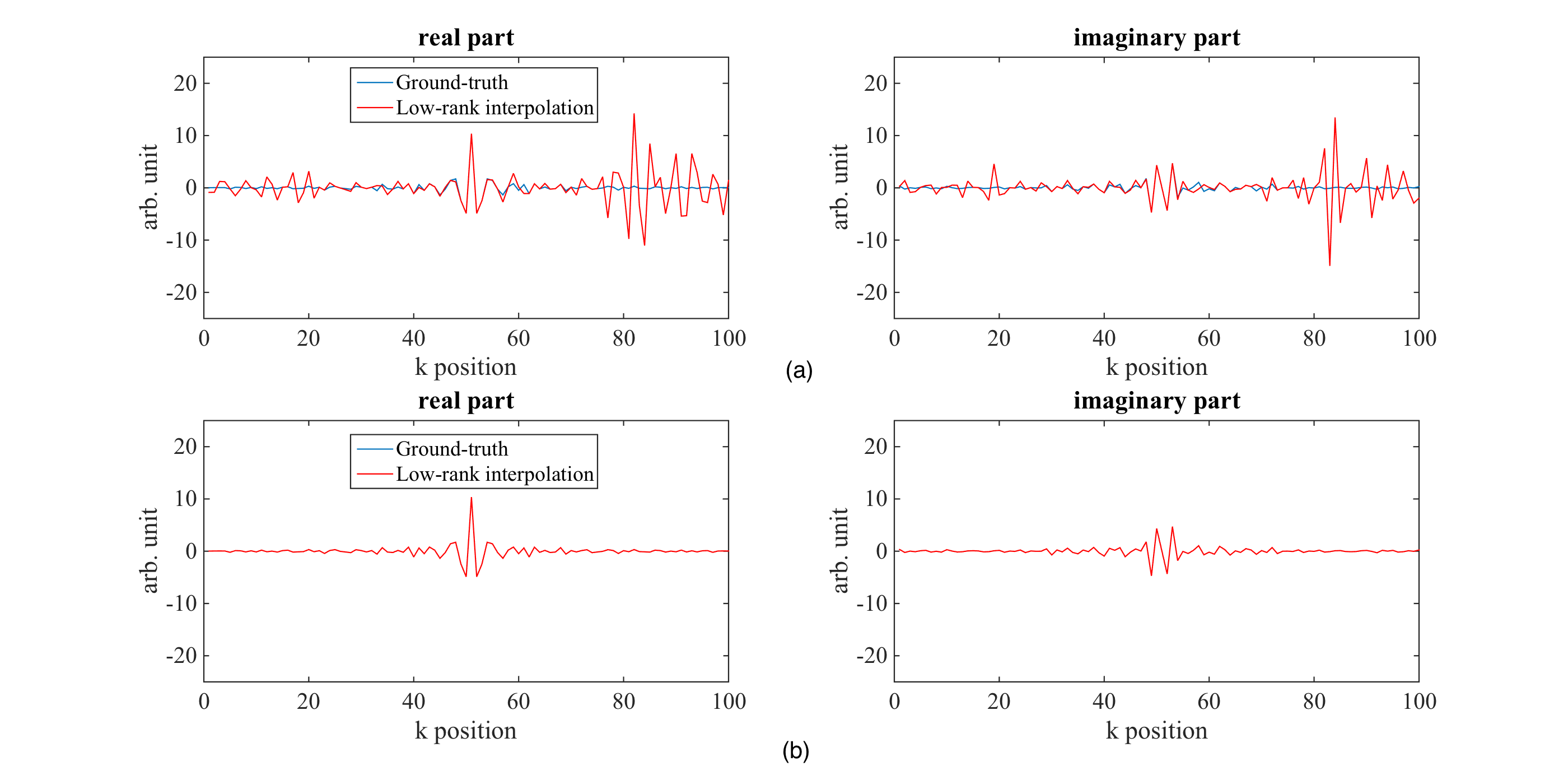}}
	\caption{
	Fully sampled Fourier measurement and the interpolated data from $m=36$ irregularly sampled data. (a) Low-rank interpolation results without spectrum weighting.
	(b) The proposed low-rank interpolation using the optimal weighting $\hat l(\omega)=i\omega$.
	 For this simulation, the following parameters were used: $d=51$, $n=100$,  and $m=36$. }
	\label{fig:offgrid_pw_pts2}
\end{figure}

Fig. \ref{fig:offgrid_recon} illustrates corresponding reconstruction results from noiseless measurements using  the proposed method. We also performed additional simulation under 40dB measurement noise.
The results clearly showed that the proposed approach accurately reconstruct the underlying piecewise constant signals.
Recall that there are no existing off the grid reconstruction algorithm
for piecewise constant signals. 
Therefore, the near perfect reconstructions by the proposed method clearly show that our theory is quite general so that it can be used for recovery of general FRI signals.


\begin{figure}[!htb]
    \centering
	\centerline{
	\includegraphics[width=8cm]{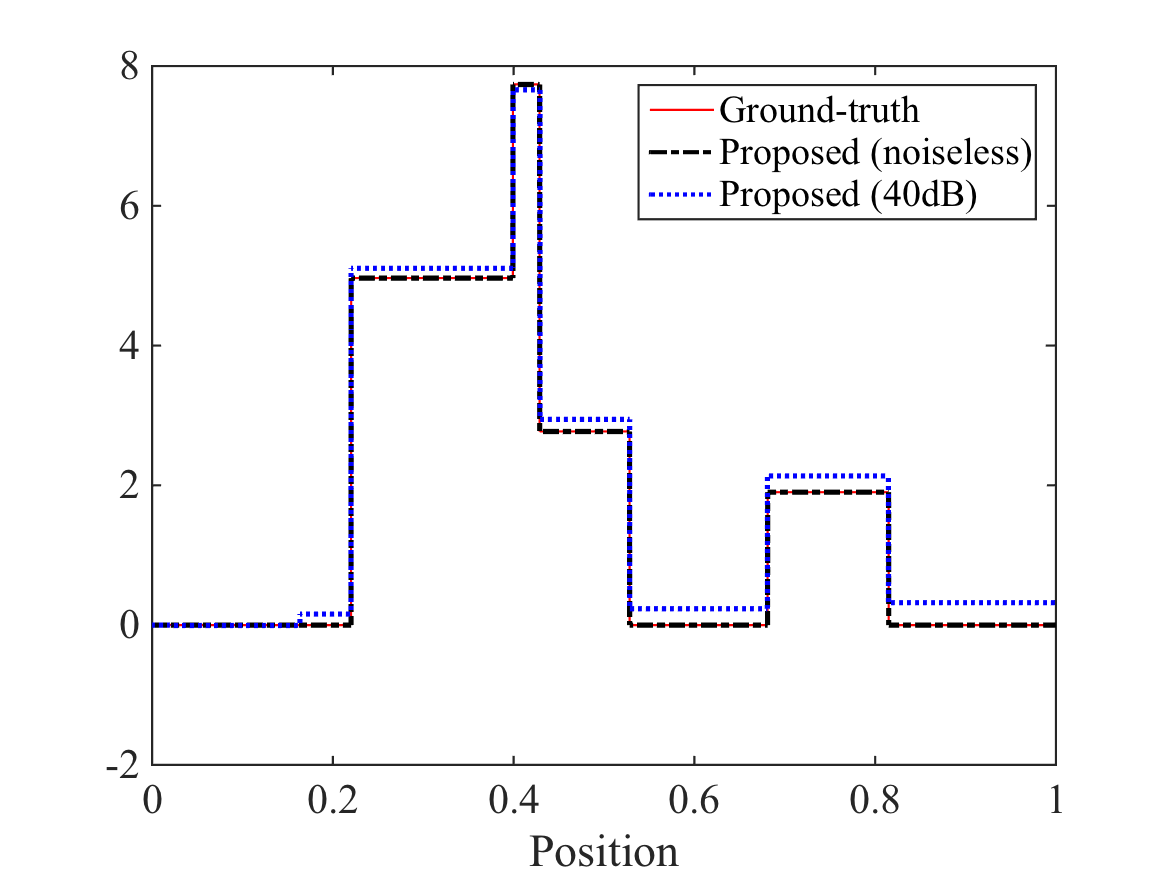}}
	\caption{
	Proposed reconstruction result of a piecewise constant signal from noiseless and 40dB noisy sparse Fourier samples. For this simulation, the following parameters were used: $d=51$, $n=100$,  and $m=36$. 
	The matrix pencil approach was used for signal recovery after the missing Fourier data was interpolated using the proposed low-rank interpolation with the optimal weighting
	$\hat l(\omega)=i\omega$.}
	\label{fig:offgrid_recon}
\end{figure}

\section{Conclusion}
\label{sec:conclusion}

While the recent theory of compressed sensing (CS) can overcome the Nyquist limit for recovering sparse signals, the standard recovery algorithms are usually implemented in discrete domain as   inverse problem approaches that 
are fully dependent on signal representations.
Moreover, the  existing spectral compressed sensing  algorithms of continuous signals such as off the grid spikes are very distinct from the discrete domain counter-parts.
To address these issues and unify the theory, this paper developed   a near optimal Fourier CS framework 
using a structured low-rank interpolator in the measurement domain before analytic reconstruction procedure is applied. This was founded by the fundamental duality between the sparsity in the primary space and the low-rankness of the structured matrix in the reciprocal spaces.
Compared to the existing spectral compressed sensing methods, our theory was generalized to encompass more general signals with finite rate of innovations, such as piecewise polynomials  and splines with provable performance guarantees.
Numerical results confirmed that the proposed methods exhibited significantly improved phase transition than the existing CS approaches.

\section*{Acknowledgement}
The authors would like to thank Prof. Michael Unser at EPFL for giving an insight into  cardinal spline representations.
This study was supported by Korea Science and Engineering Foundation under Grant NRF-2014R1A2A1A11052491.
Kiryung Lee was in part supported in part by the National Science Foundation (NSF) under Grant IIS 14-47879.

\appendix

  \setcounter{equation}{0}  
\setcounter{theorem}{0}
\counterwithin{theorem}{section}
  \renewcommand{\theequation}{A.\arabic{equation}}



\subsection{Properties of Linear Difference Equations with Constant Coefficients}

Before proving Theorem \ref{thm:hrank}, we review some important properties about linear difference equations \cite{elaydi2005introduction}. In general, the $r$-th order homogeneous difference equation has the form 
\begin{equation}\label{LinearDE}
x[k+r]+a_{r-1}[k]x[k+r-1]+\cdots+a_0[k]x[k]=0.
\end{equation}
where $\{a_i[k]\}_{0=1}^{r-1}$ are constant coefficients.
The functions $f_1[k],f_2[k],\cdots,f_r[k]$ are said to be {\it linearly independent} for $k\geq k_0$ if there are constants $c_1,c_2,\cdots,c_r$, not all zero, such that 
$$c_1f_1[k]+c_2f_2[k]+\cdots+c_rf_r[k]=0,~~~k\geq k_0.$$
And, a set of $r$ linearly independent solutions of \eqref{LinearDE} is called a {\it fundamental set} of solutions. 
As in the case of Wronskian in the theory of linear differential equations, we can easily check the linear independence of solution by using Casoratian $W[k]$ which is defined as : 
\begin{equation*}
W[k]=\begin{vmatrix}
x_1[k] & x_2[k] & \cdots & x_r[k] \\ 
x_1[k+1] & x_2[k+1] & \cdots & x_r[k+1] \\
\vdots & & & \vdots \\
x_1[k+r-1] & x_2[k+r-1] & \cdots & x_r[k+r-1] 
\end{vmatrix}
\end{equation*}

\begin{lemma}\cite[Lemma 2.13 Corollary 2.14]{elaydi2005introduction} \label{DEL1}
Let $x_1[k],\cdots,x_r[k]$ be solutions of \eqref{LinearDE} and $W[k]$ be their Casoratian. Then, 
\begin{enumerate}
\item For $k\geq k_0$, $$W[k]=(-1)^{r(k-k_0)}\left(\prod\limits_{i=k_0}^{k-1}a_0[i]\right)W[k_0];$$
\item Suppose that $a_0[k]\neq 0$ for all $k\geq k_0$. Then, $W[k]\neq 0$ for all $k\geq k_0$ iff $W[k_0]\neq 0.$
\end{enumerate}
\end{lemma}

Now, we state the criterion by which we can easily check the linear independence of solutions of \eqref{LinearDE}. 

\begin{lemma}\cite[Theorem 2.15, Theorem 2.18]{elaydi2005introduction}, \label{DEL2}
\begin{enumerate}
\item The set of solutions $x_1[k],x_2[k],\cdots,x_r[k]$ of \eqref{LinearDE} is a fundamental set iff $W[k_0]\neq 0$ for some $k_0\in\mathbb{N}.$
\item If $a_0[k]\neq 0$ for all $k\geq k_0$, then \eqref{LinearDE} has a fundamental solutions for $k\geq k_0.$
\end{enumerate}
\end{lemma}

Now, for the linear difference equation with constant coefficients, similarly to the case of differential equation with constant coefficients, we have the following result:

\begin{lemma}\cite[Lemma 2.22, Theorem 2.23]{elaydi2005introduction}\label{DEL3}
Assume that we have a linear difference equation 
\begin{equation}\label{LinearCDE}
z_{k+r}+a_{r-1}z_{k+r-1}+\cdots+a_1z_{k+1}+a_0z_k=0,
\end{equation}
where the coefficients $a_i$'s are constants and $a_0\neq 0$, and we are given 
$$P(\lambda):=\lambda^r+a_{r-1}\lambda^{r-1}+\cdots+a_1\lambda+a_0=\prod\limits_{i=0}^{p-1}(\lambda-\lambda_i)^{l_i},$$
as its characteristic equation, where $\lambda_0,\cdots,\lambda_{p-1}$ are all the distinct nonzero complex roots of $P(\lambda).$ Then the set 
$$G=\bigcup\limits_{i=0}^{p-1} G_i$$
is a fundamental set of solutions \eqref{LinearCDE}, where 
\begin{equation}\label{eq:Gi}
G_i =\left\{[\lambda_i^k]_k,\left[{k\choose 1}\lambda_i^{k-1}\right]_k,\left[{k\choose 2}\lambda_i^{k-2}\right]_k,\cdots,\left[{k\choose l_i-1}\lambda_i^{k-l_i+1}\right]_k\right\}.
\end{equation}
\end{lemma}

%
%


Here, we use the notation $[a[k]]_k$ to denote the sequence with $k$-th term $a[k]$. 
Now, we have the following results which will be used later. 
\begin{lemma}\label{DEL3-2}
The sequence set 
\begin{equation}\label{eq:gi}
\left\{\left[\lambda_i^k\right]_k,\left[{k\choose 1}\lambda_i^{k-1}\right]_k,\left[{k\choose 2}\lambda_i^{k-2}\right]_k,\cdots,\left[{k\choose l_i-1}\lambda_i^{k-l_i+1}\right]_k\right\},
\end{equation}
and 
\begin{equation}\label{eq:gi2}
\left\{\left[\lambda_i^k\right]_k,\left[k\lambda_i^k\right]_k,\cdots,\left[k^{l_i-1}\lambda_i^k\right]_k\right\}
\end{equation}
span the same sequence space. 
\end{lemma}
\begin{proof}
Note that \eqref{eq:gi} has the same span as 
$$\left\{\left[\lambda_i^k\right]_k,\left[{k\choose 1}\lambda_i^k\right]_k,\cdots,\left[{k\choose l_i-1}\lambda_i^k\right]_k\right\}$$
since  $a_0\neq 0$ and  $\lambda_i\neq 0$ for all $0\leq i\leq p-1$. Moreover, for any given $0\leq i\leq p-1$ and $0\leq j\leq l_i-1$,  
$${k\choose j}=\frac{k(k-1)\cdots(k-j+1)}{j!}$$
is a polynomial of $k$ with order $j$ so that each sequence $\left\{{k\choose j}\lambda_i^k\right\}_{k\in\mathbb{N}}$ is a linear combination of 
$$\left\{\left[\lambda_i^k\right]_k,\left[k\lambda_i^k\right]_k,\cdots,\left[k^j\lambda_i^k\right]_k\right\}.$$
On the other hand, every polynomial $P(k)$ on $\mathbb{Z}$ of degree at most $j$ uniquely expressible as a linear combination of 
$${k\choose 0},{k\choose 1},\cdots,{k\choose j}$$
and the explicit representation is given by 
$$k^j=\sum\limits_{l=0}^j a_l{k\choose l}~\text{where}~a_l=\sum\limits_{m=0}^l (-1)^{l-m}{l\choose m}m^j.$$
For the above formula, you may refer \cite{rota1975finite}. Thus, the sequence $\{k^j\lambda_i^j\}_{k\in\mathbb{N}}$ is given as a linear combination of 
$$\left\{\left[\lambda_i^k\right]_k,\left[{k\choose 1}\lambda_i^k\right]_k,\cdots,\left[{k\choose i-1}\lambda_i^k\right]_k\right\}.$$
Thus, \eqref{eq:gi} and the sequence
$$\left\{\left[\lambda_i^k\right]_k,\left[k\lambda_i^k\right]_k,\cdots,\left[k^{l_i-1}\lambda_i^k\right]_k\right\}$$
spans the same sequence space.  Q.E.D.
\end{proof}

\subsection{Proof of Theorem \ref{thm:hrank}}\label{ap:proof_hrank}
\begin{proof}
First, we will show that  $\hank(\hat \xb)$ has rank at most $r$.
Let $\hat \hb \in \Cd^{r+1}$ be the minimum size annihilating filter.
Then,   \eqref{eq:hba} can be represented as
\begin{eqnarray}
\hat\hb_a  &=&  \conv(\hat \hb) \hat \ab 
\end{eqnarray}
where $\hat\ab =\begin{bmatrix} \hat a[0] & \cdots & \hat a[k_1-1]\end{bmatrix}$ and 
$\conv(\hat \hb) \in \Cd^{d \times k_1}$ is a Toeplitz structured convolution matrix from $\hat \hb$:
\begin{eqnarray}
\conv(\hat \hb) =\begin{bmatrix} \hat h[0] & 0 & \cdots & 0 \\ \hat h[1]  & \hat h[0] & \cdots & 0 \\ \vdots &  \vdots & \ddots & \vdots  \\
\hat h[r] & \hat h[r-1] & \cdots & \hat h[r-k_1+1]  \\  \vdots & \vdots & \ddots & \vdots \\ 0 & 0 & \cdots &  \hat h[r] \end{bmatrix}  \in \Cd^{d \times k_1 }
 \end{eqnarray}
where $d=r+k_1$. Since $\conv(\hat \hb)$ is a convolution matrix, it is full ranked and 
we can show that  
$$\dim  \Ran\conv(\hat \hb)=   k_1,$$
where $\dim \Ran (\cdot)$ denotes the dimension of the range space.
Moreover, the range space of $\conv(\hat \hb)$ now  belongs to the null space of the Hankel matrix  $\hank(\hat \xb)$, so it is easy to show
$$k_1=\dim \Ran\conv(\hat \hb)\leq  \dim \Null \hank(\hat \xb),$$
where $\dim\Null(\cdot)$ represent the dimension of the null space.
Thus,
$$\rank \hank(\hat \xb) = \min\{d, n-d+1\} -  \dim \Null \hank(\hat \xb) \leq  d-k_1 = r.$$ 

Now,  we will show by contradiction that the rank of  $\hank(\hat \xb)$  cannot be smaller than $r$.
Since the rank of the  Hankel matrix is at most $r$, any set of $r+1$ consecutive rows (or columns) of $(n-d+1)\times d$ Hankel matrix with the entries $\hat x[k]$ must be linearly dependent. 
Therefore,  $\hat x[k]$ should be the solution of the following difference equation:
\begin{equation}\label{LDE}
z_{k+r}+a_{r-1}z_{k+r-1}+\cdots+a_1z_{k+1}+a_0z_k=0,~{\rm for}~0\leq k\leq n-r-1
\end{equation}
where  $\{a_i\}_{i=0}^{r-1}$ are coefficients of the linear difference equation,  and 
\begin{equation}\label{eq:P}
P(\lambda):=\lambda^r+a_{r-1}\lambda^{r-1}+\cdots+a_1\lambda+a_0=\prod\limits_{j=0}^{p-1}(\lambda-\lambda_j)^{l_j} ,\quad  
\end{equation}
is the characteristic polynomial of the linear difference equation, 
 where $\lambda_0,\cdots,\lambda_{p-1}$ are distinct nonzero complex numbers so that $a_0\neq 0$ and 
$$r=\sum\limits_{j=0}^{p-1}l_j.$$
From Lemma \ref{DEL3} and \ref{DEL3-2}, we know that
  the sequences $\{k^l\lambda_j^k\}_{k\in\mathbb{Z}}$ $(0\leq j\leq p-1~{\rm and}~0\leq l\leq l_j-1)$ are the fundamental solutions of the linear difference equation,
 and  $\hat x[k]$ can be represented as their linear combination  \cite{elaydi2005introduction}:
%
%
%
\begin{equation}\label{eq:signal}
\hat x[k]:=\sum\limits_{j=0}^{p-1}\sum\limits_{l=0}^{l_j-1}a_{j,l}k^l{\lambda_j}^k~{\rm for}~0\leq k\leq n-1,
\end{equation}
%
%
%
where all the leading coefficients $a_{j,l_j-1}$ $(0\leq j\leq p-1)$ are nonzero. 
 If we assume that the rank of the Hankel matrix with the sequence as in \eqref{eq:signal} is less than $r$, then the sequence $\hat x[k]$ must satisfy the recurrence relation of order $q<r$, since any collection of 
 $q+1$ consecutive rows (or columns) are linearly dependent. Thus, there exist a recurrence relation for $\hat x[k]$ of order $q<r$ such that 
\begin{equation}\label{LDE2}
z_{k+q}+b_{q-1}z_{k+q-1}+\cdots+b_1z_{k+1}+b_0z_k=0,~{\rm for}~0\leq k\leq n-q-1,
\end{equation}
whose solution is the sequence  given by 
\begin{equation}\label{signal2}
\hat x[k]=\sum\limits_{j=0}^{p'-1}\sum\limits_{l=0}^{l'_j-1}a'_{j,l}k^l{(\lambda'_j)}^k~{\rm for}~0\leq k\leq n-1,
\end{equation}
where $\sum_{j=0}^{p'-1}l'_j\leq r-1,$ and 
$$P_1(\lambda)=\lambda^q+b_{q-1}\lambda^{q-1}+\cdots+b_1\lambda+b_0=\prod\limits_{j=0}^{p'-1}(\lambda-\lambda'_j)^{l'_j}$$ 
is the characteristic polynomial of \eqref{LDE2}. Subtracting \eqref{signal2} from \eqref{eq:signal}, we have the equation 
\begin{equation}\label{dependency}
0=\sum_{j=0}^{P-1}\sum\limits_{l=0}^{L_j-1}c_{j,l}k^l(\Lambda_j)^k,~{\rm for}~0\leq k\leq n-1, 
\end{equation}
where 
$$\{\Lambda_j:1\leq j\leq P\}=\{\lambda_j:0\leq j\leq p-1\}\cup \{\lambda'_j:0\leq j\leq p'-1\},$$
and $L_j$ $(1\leq j\leq P)$ is the multiplicity of the root of $\Lambda_j$ for the least common multiple of $P(\lambda)$ and $P_1(\lambda).$ Moreover, we have $N:=\sum_{j=0}^{P-1}L_j\leq r+(r-1)=2r-1$.

Now, from Lemma \ref{DEL3}, we know that the sequences 
$\{\hat x_{l,j}[k]=k^l(\Lambda_j)^k:0\leq j\leq P-1,0\leq l\leq L_j-1\}$ are linearly independent sequences, and from the hypothesis $\min\{n-d+1,d\}>r$, we have $n>2r-1\geq N.$ 
Thus, if we write  \eqref{dependency}  as a matrix equation, 
$$\mathbf{0}=\Phi \mathbf{c}$$
where 
$$\mathbf{c}=[c_{0,0},\cdots,c_{0,L_0-1},c_{1,0},\cdots,\cdots,c_{P-1,0},\cdots,c_{P-1,L_{P-1}-1}]^T$$
is an $N\times 1$ matrix, and $\Phi$ is an $n\times N$ matrix which has 
$$[\Lambda_0^k,\cdots,k^{L_0-1}(\Lambda_0)^k,\Lambda_1^k,\cdots,\cdots,\Lambda_{P-1}^k,\cdots,k^{L_{P-1}-1}(\Lambda_{P-1})^k]$$
as an $k+1$-th row. By combining Lemma \ref{DEL1} and \ref{DEL2}, we can conclude that the $N\times N$ principal minor for the matrix $\Phi$ must have a nonzero determinant. Thus, the matrix $\Phi$ is of full column rank so that all the coefficients vector $\mathbf{c}$ must be zero. 


%
%
Thus, all the zeros and their multiplicities of zeros for the polynomials $P(\lambda)$, $P_1(\lambda)$ must be identical. That is a contradiction to the hypothesis that the degree of $P_1(\lambda)$ is less than that of $P(\lambda).$~~~~~~ 
\end{proof}

\subsection{A basis representation of structured matrices}
\label{subsec:basis}

\subsubsection{Hankel matrix and variations}

The linear space $\Hc(n,d)$ of $(n-d+1)$-by-$d$ Hankel matrices is spanned by a basis $\{A_k\}_{k=1}^n$ given by
\beq
\label{eq:basis}
A_k = \begin{dcases}
\frac{1}{\sqrt{k}}\sum_{i=1}^k \eb_i\eb_{k-i+1}^*,& 1 \leq k \leq d, \\
\frac{1}{\sqrt{d}}\sum_{i=1}^d \eb_i\eb_{k-i+1}^*,& d+1 \leq k \leq n-d+1, \\
\frac{1}{\sqrt{n-k+1}}\sum_{i={k-n+d}}^d \eb_i\eb_{k-i+1}^*,& n-d+2 \leq k \leq n.
\end{dcases}
\eeq

Note that $\{A_k\}_{k=1}^n$ satisfies the following properties.
First, $A_k$ is of unit Frobenius norm and all nonzero entries of $A_k$ are of the same value, i.e.,
\beq
\label{eq:const_nonzero}
[A_k]_{i,j} =
\begin{dcases}
\frac{1}{\sqrt{\norm{A_k}_0}}, & [A_k]_{i,j} \neq 0, \\
0, & \text{otherwise},
\end{dcases}
\eeq
for all $k = 1,\ldots,n$.
It follows from \eqref{eq:const_nonzero} that the spectral norm of $A_k$ is bounded by
\beq
\norm{A_k} \leq \norm{A_k}_0^{-1/2}.
\eeq
Second, each row and column of all $A_k$'s has at most one nonzero element, which implies
\beq
\label{eq:rowcolsp}
\sum_{j=1}^{n_2} \left(\sum_{i=1}^{n_1} |[A_k]_{i,j}|\right)^2 = 1,
\quad \text{and} \quad
\sum_{i=1}^{n_1} \left(\sum_{j=1}^{n_2} |[A_k]_{i,j}|\right)^2 = 1.
\eeq
Last, any two distinct elements of $\{A_k\}_{k=1}^n$ have disjoint supports,
which implies that $\{A_k\}_{k=1}^n$ constitutes an orthonormal basis for the subspace spanned by $\{A_k\}_{k=1}^n$.
 In fact, these properties are satisfied by bases for structured matrices of a similar nature
including Toeplitz, Hankel-block-Hankel, and multi-level Toeplitz matrices.

\subsubsection{Warp-around Hankel matrix and variations}

The linear space $\Hc_c(n,d)$ of $n$-by-$d$ wrap-around Hankel matrices for $n \geq d$ is spanned by a basis $\{A_k\}_{k=1}^n$ given by
\beq\label{eq:basis_cir}
A_k =
\begin{dcases}
\frac{1}{\sqrt{d}} \left( \sum_{i=1}^{k} \eb_i\eb_{(k+d-i-1)_n+1}^* + \sum_{j=n-d+k+1}^n \eb_j\eb_{(k+d-j-1)_n+1}^* \right), \quad 1 \leq k \leq d-1, \\
\frac{1}{\sqrt{d}} \sum_{i=k-d+1}^{k} \eb_i\eb_{(k+d-i-1)_n+1}^*, \quad d \leq k \leq n,
\end{dcases}
\eeq
where $(\cdot)_n$ denotes the modulo operation that finds the remainder after division by $n$.
The above basis $\{A_k\}_{k=1}^n$ for $\Hc_c(n,d)$ also satisfies the aforementioned properties of that for $\Hc(n,d)$.

Similarly, all elements of a structured matrix with the wrap-around property
(e.g., wrap-around Hankel matrix) are repeated by the same number of times.
Thus, the corresponding basis $\{A_k\}_{k=1}^n \subset \mathbb{C}^{n_1 \times n_2}$ has an extra property that
\beq
\label{eq:samesparseAk}
\norm{A_k}_0 = \min\{n_1,n_2\}, \quad \forall k = 1,\ldots,n.
\eeq

\subsection{Incoherence conditions}

The notion of the incoherence plays a crucial role in matrix completion and structured matrix completion.
We recall the definitions using our notations.
Suppose that $M \in \mathbb{C}^{n_1 \times n_2}$ is a rank-$r$ matrix whose SVD is $U \Sigma V^*$.
It was shown that the standard incoherence condition \eqref{eq:incoherence0} alone suffices to provide a near optimal sample complexity for matrix completion \cite{chen2015incoherence}.
For structured matrix completion, Chen and Chi \cite{chen2014robust} extended the notion of the standard incoherence as follows:
$M$ is said to satisfy the \emph{basis incoherence} condition with parameter $\mu$ if
\begin{equation}
\label{eq:incoherence2}
\begin{aligned}
\max_{1 \leq k \leq n} \norm{U^* A_k}_{\mathrm{F}} {} & \leq \sqrt{\frac{\mu r}{n_1}}, \\
\max_{1 \leq k \leq n} \norm{V^* A_k^*}_{\mathrm{F}} {} & \leq \sqrt{\frac{\mu r}{n_2}}.
\end{aligned}
\end{equation}
When $A_k = \eb_i \eb_j^*$, the basis incoherence reduces to the standard incoherence with the same parameter.
In general, two incoherence conditions are related as shown in the following lemma.

\begin{lemma}
\label{lem:two_inco}
Let $U \in \mathbb{C}^{n_1 \times r}$  and $V \in \mathbb{C}^{n_2 \times r}$.
Let $A_k \in \mathbb{C}^{n_1 \times n_2}$ for $k=1,\dots,n$. Then,
\begin{align*}
\max_{1 \leq k \leq n} \norm{U^* A_k}_{\mathrm{F}}
{} & \leq \left( \max_{1 \leq i' \leq n_1} \norm{ U^* \eb_{i'} }_2 \right) \cdot \max_{1 \leq k \leq n} \left[ \sum_{j=1}^{n_2} \left( \sum_{i=1}^{n_1} |[A_k]_{i,j}| \right)^2 \right]^{1/2}, \\
\max_{1 \leq k \leq n} \norm{V^* A_k^*}_{\mathrm{F}}
{} & \leq \left( \max_{1 \leq j' \leq n_2} \norm{ V^* \eb_{j'} }_2 \right) \cdot \max_{1 \leq k \leq n} \left[ \sum_{i=1}^{n_1} \left( \sum_{j=1}^{n_2} |[A_k]_{i,j}| \right)^2 \right]^{1/2}.
\end{align*}
\end{lemma}

\begin{proof}[Proof of Lemma~\ref{lem:two_inco}]
Let $k$ be an arbitrary in $\{1,\ldots,n\}$. Then,
\begin{align*}
\norm{U^* A_k}_{\mathrm{F}}^2
{} & = \norm{ U^* \sum_{i=1}^{n_1} \sum_{j=1}^{n_2} \eb_i \eb_j^* [A_k]_{i,j} }_{\mathrm{F}}^2 \\
{} & = \norm{ \sum_{j=1}^{n_2} \left( \sum_{i=1}^{n_1} [A_k]_{i,j} U^* \eb_i \right) \eb_j^* }_{\mathrm{F}}^2 \\
{} & = \sum_{j=1}^{n_2} \norm{ \sum_{i=1}^{n_1} [A_k]_{i,j} U^* \eb_i }_2^2 \\
{} & \leq \sum_{j=1}^{n_2} \left( \sum_{i=1}^{n_1} |[A_k]_{i,j}| \norm{ U^* \eb_i }_2 \right)^2 \\
{} & \leq \left( \max_{1 \leq i' \leq n_1} \norm{ U^* \eb_{i'} }_2 \right)^2 \sum_{j=1}^{n_2} \left( \sum_{i=1}^{n_1} |[A_k]_{i,j}| \right)^2.
\end{align*}
Therefore, the first claim follows by taking maximum over $k$.
The second claim is proved similarly by symmetry.
\end{proof}

%
%
%
%

By Lemma~\ref{lem:two_inco}, if \eqref{eq:rowcolsp} is satisfied,
then the standard incoherence condition implies the basis incoherence condition with the same parameter $\mu$.
However, the converse is not true in general.

\subsection{Proof of Theorem~\ref{thm:uniqueness}}
\label{subsec:pf:thm:uniqueness}
The previous work by Chen and Chi \cite{chen2014robust}
could have proved the claim (in the case without the wrap-around property) as in their Theorem~4.
However, a few steps in their proof depend on a Vandermonde decomposition of $\Lc (\xb)$,
where the generators of the involved Vandermonde matrices are of unit modulus.
Therefore, the original version \cite[Theorem~4]{chen2014robust} only applies to the spectral compressed sensing.

Essentially, their results apply to the setup in this paper with slight modifications.
In the below, a summary of the proof in the previous work \cite{chen2014robust} will be presented
with emphasis on necessary changes that enable the extension of the result by Chen and Chi \cite{chen2014robust} to the setup of this theorem.

We first adopt notations from the previous work \cite{chen2014robust}.
Define $\Ac_k: \mathbb{C}^{n_1 \times n_2} \to \mathbb{C}^{n_1 \times n_2}$ by
\[
\Ac_k(M) = A_k \langle A_k, M \rangle
\]
for $k = 1,\ldots,n$,
where 
$\langle A, B \rangle =  \mathrm{Tr}(A^*B)$
and  $\mathrm{Tr}(\cdot)$ is the trace of a matrix.
Then each $\Ac_k$ is an orthogonal projection onto the one-dimensional subspace spanned by $A_k$.
The orthogonal projection onto the subspace spanned by $\{A_k\}_{k=1}^n$ is given as $\Ac = \sum_{k=1}^n \Ac_k$.
The summation of the rank-1 projection operators in $\{\Ac_k\}_{k \in \Omega}$ is denoted by $\Ac_\Omega$,
i.e., $\Ac_\Omega := \sum_{k \in \Omega} \Ac_k$.
With repetitions in $\Omega$, $\Ac_\Omega$ is not a projection operator.
The summation of distinct elements in $\{\Ac_k\}_{k \in \Omega}$ is denoted by $\Ac'_\Omega$, which is a valid orthogonal projection.
Let $\Lc (\xb) = U \Lambda V^*$ denote the singular value decomposition of $\Lc (\xb)$.
Then the tangent space $T$ with respect to $\Lc (\xb)$ is defined as
\[
T := \{ U M^* + \widetilde{M} V^* :~ M \in \mathbb{C}^{n_2 \times r}, ~ \widetilde{M} \in \mathbb{C}^{n_1 \times r} \}.
\]
Then the projection onto $T$ and its orthogonal complement
will be denoted by $\Pc_T$ and $\Pc_{T^\perp}$, respectively.
Let $\mbox{sgn}(\widetilde{X})$ denote the sign matrix of $\widetilde{X}$, defined by $\widetilde{U} \widetilde{V}^*$,
where $\widetilde{X} = \widetilde{U} \widetilde{\Lambda} \widetilde{V}^*$ denotes the SVD of $\widetilde{X}$.
For example, $\mbox{sgn}[\Lc (\xb)] = U V^*$.
The identity operator for $\mathbb{C}^{n_1 \times n_2}$ will be denoted by $\id$.

The proof starts with Lemma~\ref{lemma:uniqueness}, which improves on the corresponding result by Chen and Chi \cite[Lemma~1]{chen2014robust}.

\begin{lemma}[{Refinement of \cite[Lemma~1]{chen2014robust}}]
\label{lemma:uniqueness}
Suppose that $\Ac_\Omega$ satisfies
\begin{equation}
\label{eq:local_isometry_rank_deficient}
\left\|
\frac{n}{m} \Pc_T \Ac_\Omega \Pc_T
- \Pc_T \Ac \Pc_T
\right\|
\leq \frac{1}{2}.
\end{equation}
If there exists a matrix $W \in \mathbb{C}^{n_1 \times n_2}$ satisfying
\begin{equation}
\label{eq:dualcert_vanish}
(\Ac - \Ac'_{\Omega}) (W) = 0,
\end{equation}
\begin{equation}
\label{eq:dualcert_sgn}
\norm{\Pc_T (W - \mbox{\upshape sgn}[\Lc (\xb)])}_{\mathrm{F}} \leq \frac{1}{7 n},
\end{equation}
and
\begin{equation}
\label{eq:dualcert_bnd}
\norm{\Pc_{T^\perp} (W)} \leq \frac{1}{2},
\end{equation}
then $\xb$ is the unique minimizer to (\ref{eq:nucmin}).
\end{lemma}

\begin{proof}[Proof of Lemma~\ref{lemma:uniqueness}]
See Appendix~\ref{subsec:pf:lemma:uniqueness}.
\end{proof}

Lemma~\ref{lemma:uniqueness}, similarly to \cite[Lemma~1]{chen2014robust}, claims that
if there exists a dual certificate matrix $W$, which satisfies \cref{eq:dualcert_vanish,eq:dualcert_sgn,eq:dualcert_bnd},
then $\xb$ is the unique minimizer to \eqref{eq:nucmin}.
Compared to the previous result \cite[Lemma~1]{chen2014robust},
Lemma~\ref{lemma:uniqueness} allows a larger deviation of the dual certificate $W$ from the sign matrix of $\Lc (\xb)$.
(Previously, the upper bound was in the order of $n^{-2}$.)
\begin{remark}
The relaxed condition on $W$ in \eqref{eq:dualcert_sgn} provides a performance guarantee
at sample complexity of the same order compared to the previous work \cite{chen2014robust}.
However, in the noisy case, this relaxed condition provides an improved performance guarantee
with significantly smaller noise amplification factor given in Theorem~\ref{thm:stability}.
\end{remark}

The next step is to construct a dual certificate $W$ that satisfies \cref{eq:dualcert_vanish,eq:dualcert_sgn,eq:dualcert_bnd}.
The version of the golfing scheme by Chen and Chi \cite{chen2014robust} still works in the setup of this theorem.
They construct a dual certificate $W$ as follows:
recall that the elements of $\Omega$ are i.i.d. following the uniform distribution on $[n]:=\{0,\cdots, n-1\}$. 
The multi-set $\Omega$ is partitioned into $j_0$ multi-sets, $\Omega_1,\ldots,\Omega_{j_0}$
so that each $\Omega_j$ contains $m/j_0$ i.i.d. samples.
A sequence of matrices $(F_0,\ldots,F_{j_0})$ are generated recursively by
\[
F_j = \Pc_T \left(\Ac - \frac{n j_0}{m} A_{\Omega_j}^* \right) F_{j-1}, \quad j = 1,\ldots,j_0,
\]
starting from $F_0 = \mbox{sgn}[\Lc (\xb)] = U V^*$.
Then, $W$ is obtained by
\[
W = \sum_{j=1}^{j_0} \left(\frac{n j_0}{m} \Ac_{\Omega_j}^* + \id - \Ac \right) F_{j-1}.
\]

{Chen and Chi showed that if $j_0 = 3 \log_{1/\epsilon} n$ for a small constant $\epsilon < e^{-1}$, 
then $W$ satisfies \cref{eq:dualcert_vanish,eq:dualcert_sgn} with high probability \cite[Section~VI.C]{chen2014robust}. 
In fact, they showed that a sufficient condition for \eqref{eq:dualcert_sgn} given by
\[
\norm{\Pc_T (W - \mbox{\upshape sgn}[\Lc (\xb)])}_{\mathrm{F}} \leq \frac{1}{2 n^2}
\]
is satisfied. Thus, without any modification, their arguments so far apply to completion of structured matrices in the setup of this theorem.}
Chen and Chi verified that $W$ satisfies \cref{eq:dualcert_vanish,eq:dualcert_sgn} \cite[Section~VI.C]{chen2014robust}.
Without any modification, their arguments so far apply to completion of structured matrices in the setup of this theorem.

They verified that $W$ also satisfies the last property in \eqref{eq:dualcert_bnd} with some technical conditions \cite[Section~VI.D]{chen2014robust}. Specifically, Chen and Chi verified that $W$ satisfies \eqref{eq:dualcert_bnd} through a sequence of lemmas \cite[Lemmas~4,5,6,7]{chen2014robust}
using intermediate quantities given in terms of the following two norms:
\beq
\label{eq:defAcinf}
\norm{M}_{\Ac,\infty} := \max_{1 \leq k \leq n} \left| \langle A_k, M \rangle \right| \norm{A_k},
\eeq
and
\beq
\label{eq:defAc2}
\norm{M}_{\Ac,2} := \left( \sum_{k=1}^n \left| \langle A_k, M \rangle \right|^2 \norm{A_k}^2 \right)^{1/2}.
\eeq
Since most of their arguments in \cite[Sections~VI.D, VI.E]{chen2014robust} generalize to the setup of our theorem,
we do not repeat technical details here.
However, there was one place where the generalization fails.
The results in \cite[Lemma~7]{chen2014robust} provide upper bounds on the initialization of the dual certificate algorithm in the above two norms.
We found that Chen and Chi used both the standard incoherence \eqref{eq:incoherence0} and the basis incoherence \eqref{eq:incoherence2} in this step.
In fact, the proof of \cite[Lemma~7]{chen2014robust} depends crucially on a Vandermonde decomposition with generators of unit modulus.
In spectral compressed sensing, by controlling the condition number of Vandermonde matrices,
both incoherence properties are satisfied with the same parameter.
In fact, this is the place where their proof fails to generalize to other structured matrix completion.


In our setup, we assume that \eqref{eq:rowcolsp} is satisfied.
By Lemma~\ref{lem:two_inco}, the standard incoherence property implies the basis incoherence, and then the dependence on the structure due to a Vandermonde decomposition disappears.
Thus, only the standard incoherence of $\Lc (\xb)$ is included among the hypotheses.

{
\begin{lemma}{\cite[Lemma~7]{chen2014robust}}
\label{lem:ref7cc}
Let $\norm{\cdot}_{\Ac,\infty}$ and $\norm{\cdot}_{\Ac,2}$ be defined respectively in \eqref{eq:defAcinf} and \eqref{eq:defAc2}.
The standard incoherence property with parameter $\mu$ implies that there exists an absolute constant $c_6$ such that
\begin{align}
\norm{U V^*}_{\Ac,\infty} {} & \leq \frac{\mu r}{\min(n_1,n_2)}, \label{eq:lem:ref7cc:a}\\
\norm{U V^*}_{\Ac,2}^2 {} & \leq \frac{c_6 \mu r \log^2 n}{\min(n_1,n_2)}, \label{eq:lem:ref7cc:b}
\intertext{and}
\norm{\Pc_T \left(\norm{A_k}_0^{1/2} A_k\right)}_{\Ac,2}^2 {} & \leq \frac{c_6 \mu r \log^2 n}{\min(n_1,n_2)}, \quad \forall k = 1,\ldots,n. \label{eq:lem:ref7cc:c}
\end{align}
\end{lemma}}

\begin{proof}[Proof of Lemma~\ref{lem:ref7cc}]
{
Although \cite[Lemma~7]{chen2014robust} did not assume that $U \in \mathbb{C}^{n_1 \times r}$ (resp. $V \in \mathbb{C}^{n_2 \times r}$) consists of the left (resp. right) singular vectors of a rank-$r$ Hankel matrix with a Vandermode decomposition with generators of unit modulus,
this condition was used in the proof by Chen and Chi \cite[Appendix~H]{chen2014robust}.
More precisely, they used the Vandermonde decomposition to get the following inequalities:
\[
\max_{1 \leq i \leq n_1} \norm{U^* \eb_i}_2^2 \leq \frac{\mu r}{n_1}
\quad \text{and} \quad
\max_{1 \leq j \leq n_2} \norm{V^* \eb_j}_2^2 \leq \frac{\mu r}{n_2}.
\]
These inequalities are exactly the standard incoherence property with parameter $\mu$.
Except these inequalities, their proof generalizes without requiring the Vandermonde decomposition.
Thus, we slightly modify \cite[Lemma~7]{chen2014robust} by including the standard incoherence as an assumption to the lemma.}
\end{proof}

The proof by Chen and Chi \cite{chen2014robust} focused on the Hankel-block-Hankel matrix where
the elements in the basis $\{A_k\}_{k=1}^n$ have varying sparsity levels.
In the case of structured matrices with the wrap-around property,
the sparsity levels of $\{A_k\}_{k=1}^n$ are the same.
Thus, this additional property can be used to tighten the sample complexity by reducing the order of $\log n$ term.
More specifically, we improve \cite[Lemma~7]{chen2014robust} with the wrap-around property in the next lemma.
(The upper bounds on the terms in $\norm{\cdot}_{\Ac,2}$ were larger by factor of $\log^2 n$ in \cite[Lemma~7]{chen2014robust}.)
\begin{lemma}[{Analog of \cite[Lemma~7]{chen2014robust} with the wrap-around property}]
\label{lem:ref7}
Let $\norm{\cdot}_{\Ac,\infty}$ and $\norm{\cdot}_{\Ac,2}$ be defined respectively in \eqref{eq:defAcinf} and \eqref{eq:defAc2}.
The standard incoherence property with parameter $\mu$ implies
\begin{align}
\norm{U V^*}_{\Ac,\infty} {} & \leq \frac{\mu r}{\min(n_1,n_2)}, \label{eq:lem:ref7:a}\\
\norm{U V^*}_{\Ac,2}^2 {} & \leq \frac{\mu r}{\min(n_1,n_2)}, \label{eq:lem:ref7:b}
\intertext{and}
\norm{\Pc_T \left(\norm{A_k}_0^{1/2} A_k\right)}_{\Ac,2}^2 {} & \leq \frac{9 \mu r}{\min(n_1,n_2)}, \quad \forall k = 1,\ldots,n. \label{eq:lem:ref7:c}
\end{align}
\end{lemma}

\begin{proof}[Proof of Lemma~\ref{lem:ref7}]
See Appendix~\ref{subsec:pf:lem:ref7}.
\end{proof}

In the wrap-around case, it only remains to verify that we can drop the order of $\log n$ from 4 to 2.
In the previous work \cite[Section~VI.E]{chen2014robust},
the $\log^4 n$ term appears only through the parameter $\mu_5$, which is in the order of $\log^2 n$.
Due to Lemma~\ref{lem:ref7}, parameter $\mu_5$ reduces by factor of $\log^2 n$. Thus, the sample complexity reduces by the same factor.
This completes the proof.

\subsection{Proof of Lemma~\ref{lemma:uniqueness}}
\label{subsec:pf:lemma:uniqueness}
Our proof essentially adapts the arguments of Chen and Chi \cite[Appendix~B]{chen2014robust}.
The upper bound on the deviation of $W$ from $\mbox{\upshape sgn}[\Lc (\xb)]$ in \eqref{eq:dualcert_sgn} is sharpened in order by optimizing parameters.

Let $\hat{\xb} = \xb + \hb$ be the minimizer to (\ref{eq:nucmin}). We show that $\Lc (\hb) = 0$ in two complementary cases.
Then by the injectivity of $\Lc$, $\hb = 0$, or equivalently, $\hat{\xb} = \xb$.

\noindent\textbf{Case 1:} We first consider the case when $\Lc (\hb)$ satisfies
\begin{equation}
\label{eq:case1}
\norm{\Pc_T \Lc (\hb)}_{\mathrm{F}} \leq 3 n \norm{\Pc_{T^\perp} \Lc (\hb)}_{\mathrm{F}}.
\end{equation}

Since $T$ is the tangent space of $\Lc (\xb)$, $\Pc_{\Tc^\perp} \Lc (\xb) = 0$.
Thus $\Pc_T (\mbox{\upshape sgn}[\Lc (\xb)] + \mbox{\upshape sgn}[\Pc_{T^\perp} \Lc (\hb)]) = \Pc_T (\mbox{\upshape sgn}[\Lc (\xb)])$.
Furthermore, $\norm{\mbox{\upshape sgn}[\Lc (\xb)] + \mbox{\upshape sgn}[\Pc_{T^\perp} \Lc (\hb)]} \leq 1$.
Therefore, $\mbox{\upshape sgn}[\Lc (\xb)] + \mbox{\upshape sgn}[\Pc_{T^\perp} \Lc (\hb)]$ is a valid sub-gradient of the nuclear norm at $\Lc (\xb)$.
Then it follows that
\begin{equation}
\label{eq:pf_lemma_uniqueness:ineq1}
\begin{aligned}
\norm{\Lc (\xb) + \Lc (\hb)}_*
{} & \geq \norm{\Lc (\xb)}_* + \langle \mbox{\upshape sgn}[\Lc (\xb)] + \mbox{\upshape sgn}[\Pc_{T^\perp} \Lc (\hb)], ~ \Lc (\hb) \rangle \\
{} & = \norm{\Lc (\xb)}_*
+ \langle W, \Lc (\hb) \rangle
+ \langle \mbox{\upshape sgn}[\Pc_{T^\perp} \Lc (\hb)], ~ \Lc (\hb) \rangle
- \langle W - \mbox{\upshape sgn}[\Lc (\xb)], \Lc (\hb) \rangle.
\end{aligned}
\end{equation}

In fact, $\langle W, \Lc (\hb) \rangle  = 0$ as shown below.
The inner product of $\Lc (\hb)$ and $W$ is decomposed as
\begin{equation}
\label{eq:pf_lemma_uniqueness:ip}
\langle W, \Lc (\hb) \rangle
= \langle W, (\id - \Ac) \Lc (\hb) \rangle
+ \langle W, (\Ac - \Ac'_{\Omega}) \Lc (\hb) \rangle
+ \langle W, \Ac'_{\Omega} \Lc (\hb) \rangle.
\end{equation}
Indeed, all three terms in the right-hand-side of \eqref{eq:pf_lemma_uniqueness:ip} are 0.
This can be shown as follows. Since $\Ac$ is the orthogonal projection onto the range space of $\Lc$, the first term is 0.
The second term is 0 by the assumption on $W$ in \eqref{eq:dualcert_vanish}.
Since $\hat{\xb}$ is feasible for \eqref{eq:nucmin}, $P_\Omega (\hat{\xb}) = P_\Omega (\xb)$.
Thus $P_\Omega (\hb) = P_\Omega (\hat{\xb} - \xb) = 0$.
Since $\{A_k\}_{k=1}^n$ is an orthonormal basis, we have
\begin{equation}
\label{eq:pf_lemma_uniqueness:vanish}
\Ac_\omega \Lc (\hb) =
\sum_{k \in [n]\setminus\Omega} \langle \eb_k, \hb \rangle \langle \Ac_\omega, A_k \rangle = 0,
\quad \forall \omega \in \Omega.
\end{equation}
It follows that $\Ac'_{\Omega} \Lc (\hb) = 0$. Thus, the third term of the right-hand-side of \eqref{eq:pf_lemma_uniqueness:ip} is 0.

Since the $\mbox{\upshape sgn}(\cdot)$ operator commutes with $\Pc_{T^\perp}$, and $\Pc_{T^\perp}$ is idempotent, we get
\begin{align*}
\langle \mbox{\upshape sgn}[\Pc_{T^\perp} \Lc (\hb)], ~ \Lc (\hb) \rangle
{} & = \langle \Pc_{T^\perp} \mbox{\upshape sgn}[\Pc_{T^\perp} \Lc (\hb)], ~ \Lc (\hb) \rangle \\
{} & = \langle \mbox{\upshape sgn}[\Pc_{T^\perp} \Lc (\hb)], ~ \Pc_{T^\perp} \Lc (\hb) \rangle \\
{} & = \norm{\Pc_{T^\perp} \Lc (\hb)}_*.
\end{align*}

Then \eqref{eq:pf_lemma_uniqueness:ineq1} implies
\begin{equation}
\label{eq:pf_lemma_uniqueness:ineq2}
\norm{\Lc (\xb) + \Lc (\hb)}_* \geq \norm{\Lc (\xb)}_* + \norm{\Pc_{T^\perp} \Lc (\hb)}_* - \langle W - \mbox{\upshape sgn}[\Lc (\xb)], \Lc (\hb) \rangle.
\end{equation}

We derive an upper bound on the magnitude of the third term in the right-hand-side of \eqref{eq:pf_lemma_uniqueness:ineq2} given by
\begin{subequations}
\label{eq:pf_lemma_uniqueness:ineq3}
\begin{align}
| \langle W - \mbox{\upshape sgn}[\Lc (\xb)], \Lc (\hb) \rangle |
{} & = | \langle \Pc_T (W - \mbox{\upshape sgn}[\Lc (\xb)]), \Lc (\hb) \rangle + \langle \Pc_{T^\perp} (W - \mbox{\upshape sgn}[\Lc (\xb)]), \Lc (\hb) \rangle | \nonumber \\
{} & \leq | \langle \Pc_T (W - \mbox{\upshape sgn}[\Lc (\xb)]), \Lc (\hb) \rangle | + | \langle \Pc_{T^\perp} (W), \Lc (\hb) \rangle | \label{eq:pf_lemma_uniqueness:ineq3a} \\
{} & \leq \norm{\Pc_T (W - \mbox{\upshape sgn}[\Lc (\xb)])}_{\mathrm{F}} \norm{\Pc_T \Lc (\hb)}_{\mathrm{F}}
+ \norm{\Pc_{T^\perp} (W)} \norm{\Pc_{T^\perp} \Lc (\hb)}_* \label{eq:pf_lemma_uniqueness:ineq3b} \\
{} & \leq \frac{1}{7 n} \norm{\Pc_T \Lc (\hb)}_{\mathrm{F}} + \frac{1}{2} \norm{\Pc_{T^\perp} \Lc (\hb)}_*, \label{eq:pf_lemma_uniqueness:ineq3c}
\end{align}
\end{subequations}
where \eqref{eq:pf_lemma_uniqueness:ineq3a} holds by the triangle inequality and the fact that $\Pc_{T^\perp} \Lc (\xb) = 0$;
\eqref{eq:pf_lemma_uniqueness:ineq3b} by H\"older's inequality;
\eqref{eq:pf_lemma_uniqueness:ineq3c} by the assumptions on $W$ in \eqref{eq:dualcert_sgn} and \eqref{eq:dualcert_bnd}.

We continue by applying \eqref{eq:pf_lemma_uniqueness:ineq3} to \eqref{eq:pf_lemma_uniqueness:ineq2} and get
\begin{align*}
\norm{\Lc (\xb) + \Lc (\hb)}_*
{} & \geq \norm{\Lc (\xb)}_* - \frac{1}{7 n} \norm{\Pc_T \Lc (\hb)}_{\mathrm{F}} + \frac{1}{2} \norm{\Pc_{T^\perp} \Lc (\hb)}_* \\
{} & \geq \norm{\Lc (\xb)}_* - \frac{3}{7} \norm{\Pc_{T^\perp} \Lc (\hb)}_{\mathrm{F}} + \frac{1}{2} \norm{\Pc_{T^\perp} \Lc (\hb)}_{\mathrm{F}} \\
{} & = \norm{\Lc (\xb)}_* + \frac{1}{14} \norm{\Pc_{T^\perp} \Lc (\hb)}_{\mathrm{F}},
\end{align*}
where the second step follows from \eqref{eq:case1}.

Then, $\norm{\Lc (\hat{\xb})}_* \geq \norm{\Lc (\xb)}_* \geq \norm{\Lc (\hat{\xb})}_*$, which implies $\Pc_{T^\perp} \Lc (\hb) = 0$.
By (\ref{eq:case1}), we also have $\Pc_T \Lc (\hb) = 0$.
Therefore, it follows that $\Lc (\hb) = 0$.

\noindent\textbf{Case 2:} Next, we consider the complementary case when $\Lc (\hb)$ satisfies
\begin{equation}
\label{eq:case2}
\norm{\Pc_T \Lc (\hb)}_{\mathrm{F}} {\geq} 3 n \norm{\Pc_{T^\perp} \Lc (\hb)}_{\mathrm{F}}.
\end{equation}

Note that \eqref{eq:pf_lemma_uniqueness:vanish} implies $\Ac_\Omega \Lc (\hb) = 0$.
Then together with $(\id - \Ac) \Lc = 0$, we get
\begin{align*}
\left(\frac{n}{m} \Ac_\Omega + \id - \Ac\right) \Lc (\hb) = 0,
\end{align*}
which implies
\begin{equation}
\label{eq:pf_lemma_uniqueness:ineq4}
\begin{aligned}
0 {} & = \left\langle \Pc_T \Lc (\hb), \left(\frac{n}{m} \Ac_\Omega + \id - \Ac\right) \Lc (\hb) \right\rangle \\
{} & = \left\langle \Pc_T \Lc (\hb), \left(\frac{n}{m} \Ac_\Omega + \id - \Ac\right) \Pc_T \Lc (\hb) \right\rangle \\
{} & \quad + \left\langle \Pc_T \Lc (\hb), \left(\frac{n}{m} \Ac_\Omega + \id - \Ac\right) \Pc_{T^\perp} \Lc (\hb) \right\rangle.
\end{aligned}
\end{equation}

The magnitude of the first term in the right-hand-side of \eqref{eq:pf_lemma_uniqueness:ineq4} is lower-bounded by
\begin{equation}
\label{eq:pf_lemma_uniqueness:lb}
\begin{aligned}
{} & \left| \left\langle \Pc_T \Lc (\hb), \left(\frac{n}{m} \Ac_\Omega + \id - \Ac\right) \Pc_T \Lc (\hb) \right\rangle \right| \\
{} & = \left| \langle \Pc_T \Lc (\hb), \Pc_T \Lc (\hb) \rangle \right|
- \left| \left\langle \Pc_T \Lc (\hb), \left(\Ac - \frac{n}{m} \Ac_\Omega\right) \Pc_T \Lc (\hb) \right\rangle \right| \\
{} & \geq \norm{\Pc_T \Lc (\hb)}_{\mathrm{F}}^2 - \left\|\Pc_T \Ac \Pc_T - \frac{n}{m} \Pc_T \Ac_\Omega \Pc_T\right\| \norm{\Pc_T \Lc (\hb)}_{\mathrm{F}}^2 \\
{} & \geq \frac{1}{2} \norm{\Pc_T \Lc (\hb)}_{\mathrm{F}}^2,
\end{aligned}
\end{equation}
where the last step follows from the assumption in \eqref{eq:local_isometry_rank_deficient}.

Next, we derive an upper bound on the second term in the right-hand-side of \eqref{eq:pf_lemma_uniqueness:ineq4}.
Since $\Ac_{\omega_j}$ is an orthogonal projection for $j \in [m]$,
the operator norm of $\frac{n}{m} \Ac_\Omega + \id - \Ac$ is upper-bounded by
\begin{equation}
\label{eq:pf_lemma_uniqueness:ineq5}
\begin{aligned}
\left\| \frac{n}{m} \Ac_\Omega + \id - \Ac \right\|
{} & \leq \frac{n}{m} \left( \norm{\Ac_{\omega_1} + \id - \Ac} + \sum_{j=2}^m \norm{\Ac_{\omega_j}}_{\mathrm{F}} \right) \\
{} & \leq \frac{n}{m} \left( \max(\norm{\Ac_{\omega_1}}, \norm{\id - \Ac}) + \sum_{j=2}^m \norm{\Ac_{\omega_j}}_{\mathrm{F}} \right) \\
{} & \leq n,
\end{aligned}
\end{equation}
where the second step follows since $\Ac_{\omega_1} (\id - \Ac) = 0$.

The second term in the right-hand-side of \eqref{eq:pf_lemma_uniqueness:ineq4} is then upper-bounded by
\begin{equation}
\label{eq:pf_lemma_uniqueness:ub}
\begin{aligned}
{} & \left| \left\langle \Pc_T \Lc (\hb), \left(\frac{n}{m} \Ac_\Omega + \id - \Ac\right) \Pc_{T^\perp} \Lc (\hb) \right\rangle \right| \\
{} & \leq \left\| \frac{n}{m} \Ac_\Omega + \id - \Ac \right\| \norm{\Pc_T \Lc (\hb)}_{\mathrm{F}} \norm{\Pc_{T^\perp} \Lc (\hb)}_{\mathrm{F}} \\
{} & \leq n \norm{\Pc_T \Lc (\hb)}_{\mathrm{F}} \norm{\Pc_{T^\perp} \Lc (\hb)}_{\mathrm{F}},
\end{aligned}
\end{equation}
where the last step follows from \eqref{eq:pf_lemma_uniqueness:ineq5}.

Applying \eqref{eq:pf_lemma_uniqueness:lb} and \eqref{eq:pf_lemma_uniqueness:ub} to \eqref{eq:pf_lemma_uniqueness:ineq4} provides
\begin{align*}
0 {} & = \left| \left\langle \Pc_T \Lc (\hb), \left(\frac{n}{m} \Ac_\Omega + \id - \Ac\right) \Pc_T \Lc (\hb) \right\rangle \right| \\
{} & \quad - \left| \left\langle \Pc_T \Lc (\hb), \left(\frac{n}{m} \Ac_\Omega + \id - \Ac\right) \Pc_{T^\perp} \Lc (\hb) \right\rangle \right| \\
{} & \geq \frac{1}{2} \norm{\Pc_T \Lc (\hb)}_{\mathrm{F}}^2 - n \norm{\Pc_T \Lc (\hb)}_{\mathrm{F}} \norm{\Pc_{T^\perp} \Lc (\hb)}_{\mathrm{F}} \\
{} & \geq \frac{1}{2} \norm{\Pc_T \Lc (\hb)}_{\mathrm{F}}^2 - \frac{1}{3} \norm{\Pc_T \Lc (\hb)}_{\mathrm{F}}^2 \\
{} & = \frac{1}{6} \norm{\Pc_T \Lc (\hb)}_{\mathrm{F}}^2 \geq 0,
\end{align*}
where the second inequality holds by \eqref{eq:case2}.

Then, it follows that $\Pc_T \Lc (\hb) = 0$. By (\ref{eq:case2}), we also have $\Pc_{T^\perp} \Lc (\hb) = 0$.
Therefore, $\Lc (\hb) = 0$, which completes the proof.

\subsection{Proof of Lemma~\ref{lem:ref7}}
\label{subsec:pf:lem:ref7}
The proof is obtained by slightly modifying that of \cite[Lemma~7]{chen2014robust}.

The first upper bound in \eqref{eq:lem:ref7:a} is derived as follows:
\begin{align*}
\norm{U V^*}_{\Ac,\infty}
{} & = \max_{1 \leq k \leq n} \left| \langle A_k, U V^* \rangle \right| \norm{A_k} \\
{} & = \max_{1 \leq k \leq n} \frac{\left| \sum_{(i,j) \in \text{supp}(A_k)} [U V^*]_{i,j} \right|}{\norm{A_k}_0} \\
{} & \leq \max_{1 \leq k \leq n} \max_{(i,j) \in \text{supp}(A_k)} |[U V^*]_{i,j}| \\
{} & = \max_{1 \leq i \leq n_1} \max_{1 \leq j \leq n_2} |[U V^*]_{i,j}| \\
{} & = \max_{1 \leq i \leq n_1} \max_{1 \leq j \leq n_2} |\eb_i^* U V^* \eb_j| \\
{} & = \max_{1 \leq i \leq n_1} \norm{U^* \eb_i}_2 \max_{1 \leq j \leq n_2} \norm{V^* \eb_j}_2 \\
{} & \leq \frac{\mu r}{\sqrt{n_1 n_2}} \leq \frac{\mu r}{\min(n_1,n_2)}.
\end{align*}
This proves \eqref{eq:lem:ref7:a}.

Next, to prove \cref{eq:lem:ref7:b,eq:lem:ref7:c}, we use the following lemma.

\begin{lemma}
\label{lem:Ac2bnd}
Let $M \in \mathbb{C}^{n_1 \times n_2}$. Then,
\beq
\label{eq:pf:lem:ref7:ineq}
\norm{M}_{\Ac,2}^2
\leq
\max \left( \max_{1 \leq i \leq n_1} \norm{\eb_i^* M}_2^2, ~ \max_{1 \leq j \leq n_2} \norm{M \eb_j}_2^2 \right)
\eeq
\end{lemma}

\begin{proof}[Proof of Lemma~\ref{lem:Ac2bnd}]
See Appendix~\ref{subsec:pf:lem:Ac2bnd}.
\end{proof}

Then, \eqref{eq:lem:ref7:b} is proved as follows: Since $U$ and $V$ are unitary matrices, we have
\[
\norm{\eb_i^* U V^*}_{\mathrm{F}} = \norm{U^* \eb_i}_2
\quad \text{and} \quad
\norm{U V^* \eb_j}_{\mathrm{F}} = \norm{V^* \eb_j}_2
\]
for all $1 \leq i \leq n_1$ and for all $1 \leq j \leq n_2$. Thus,
\beq
\label{eq:pf:lem:ref7:bnd1}
\max\left(
\max_{1 \leq i \leq n_1} \norm{\eb_i^* U V^*}_{\mathrm{F}}^2, ~
\max_{1 \leq j \leq n_2} \norm{U V^* \eb_j}_{\mathrm{F}}^2
\right)
\leq \frac{\mu r}{\min(n_1,n_2)}.
\eeq
Then \eqref{eq:lem:ref7:b} follows by applying \eqref{eq:pf:lem:ref7:bnd1} to Lemma~\ref{lem:Ac2bnd} with $M = U V^*$.

Lastly, we prove \eqref{eq:lem:ref7:c}. By definition of $\Pc_T$,
\beq
\label{eq:pf:lem:ref7:decomp}
\begin{aligned}
\norm{\eb_i^* \left[ \Pc_T \left(\norm{A_k}_0^{1/2} A_k\right) \right]}_{\mathrm{F}}^2
{} & \leq 3 \norm{\eb_i^* U U^* \norm{A_k}_0^{1/2} A_k}_{\mathrm{F}}^2 \\
{} & \quad + 3 \norm{\eb_i^* \norm{A_k}_0^{1/2} A_k V V^*}_{\mathrm{F}}^2 \\
{} & \quad + 3 \norm{\eb_i^* U U^* \norm{A_k}_0^{1/2} A_k V V^*}_{\mathrm{F}}^2,
\end{aligned}
\eeq
for all $i \in \{1,\ldots,n_1\}$.
The first term in the right-hand-side of \eqref{eq:pf:lem:ref7:decomp} is upper-bounded by
\beq
\label{eq:pf:lem:ref7:ub1}
\norm{\eb_i^* U U^* \norm{A_k}_0^{1/2} A_k}_{\mathrm{F}}^2
\leq \norm{\eb_i^* U}_2^2 \norm{\norm{A_k}_0^{1/2} A_k}^2 \leq \frac{\mu r}{n_1},
\eeq
where the last step follows from $\norm{A_k} \leq \norm{A_k}_0^{-1/2}$.
Since $\norm{V V^*} \leq 1$, the first term dominates the third term in the right-hand-side of \eqref{eq:pf:lem:ref7:decomp}.
Note that $\norm{A_k}_0^{1/2} A_k$ is a submatrix of a permutation matrix.
Therefore, $\eb_i^* \norm{A_k}_0^{1/2} A_k = \eb_j^*$ for some $j \in \{1,\ldots,n_1\}$.
Then, the second term in the right-hand-side of \eqref{eq:pf:lem:ref7:decomp} is upper-bounded by
\beq
\label{eq:pf:lem:ref7:ub2}
\norm{\eb_i^* \norm{A_k}_0^{1/2} A_k V V^*}_{\mathrm{F}}^2
= \norm{\eb_j^* V V^*}_{\mathrm{F}}^2 \leq \frac{\mu r}{n_2}.
\eeq

Plugging \cref{eq:pf:lem:ref7:ub1,eq:pf:lem:ref7:ub2} to \eqref{eq:pf:lem:ref7:decomp} provides
\beq
\label{eq:pf:lem:ref7:bnd2a}
\max_{1 \leq i \leq n_1} \norm{\eb_i^* \left[ \Pc_T \left(\norm{A_k}_0^{1/2} A_k\right) \right]}_{\mathrm{F}}^2
\leq \frac{9 \mu r}{\min(n_1,n_2)}.
\eeq

By symmetry, we also get
\beq
\label{eq:pf:lem:ref7:bnd2b}
\max_{1 \leq j \leq n_2} \norm{\left[ \Pc_T \left(\norm{A_k}_0^{1/2} A_k\right) \right] \eb_j}_{\mathrm{F}}^2
\leq \frac{9 \mu r}{\min(n_1,n_2)}.
\eeq

Applying \cref{eq:pf:lem:ref7:bnd2a,eq:pf:lem:ref7:bnd2b} to Lemma~\ref{lem:Ac2bnd} with $M = \Pc_T \left(\norm{A_k}_0^{1/2} A_k\right)$
completes the proof.

\subsection{Proof of Lemma~\ref{lem:Ac2bnd}}
\label{subsec:pf:lem:Ac2bnd}
The inequality in \eqref{eq:pf:lem:ref7:ineq} is proved as follows:
\begin{align*}
\norm{M}_{\Ac,2}^2
{} & = \sum_{k=1}^n \left| \langle A_k, M \rangle \right|^2 \norm{A_k}^2 \\
{} & = \sum_{k=1}^n \frac{\left| \sum_{(i,j) \in \text{supp}(A_k)} [M]_{i,j} \right|^2}{\norm{A_k}_0} \norm{A_k}^2 \\
{} & \leq \sum_{k=1}^n \frac{\left( \sum_{(i,j) \in \text{supp}(A_k)} |[M]_{i,j}| \right)^2}{\norm{A_k}_0} \norm{A_k}^2 \\
{} & \leq \sum_{k=1}^n \sum_{(i,j) \in \text{supp}(A_k)} |[M]_{i,j}|^2 \norm{A_k}^2 \\
{} & \leq \frac{1}{\min\{n_1,n_2\}} \sum_{k=1}^n \sum_{(i,j) \in \text{supp}(A_k)} |[M]_{i,j}|^2 \\
{} & = \frac{1}{\min\{n_1,n_2\}} \sum_{i=1}^{n_1} \sum_{j=1}^{n_2} |[M]_{i,j}|^2 \\
{} & \leq \max \left( \frac{1}{n_1} \sum_{i=1}^{n_1} \norm{\eb_i^* M}_2^2 ,~ \frac{1}{n_2} \sum_{j=1}^{n_2} \norm{M \eb_j}_2^2 \right) \\
{} & \leq \max \left( \max_{1 \leq i \leq n_1} \norm{\eb_i^* M}_2^2, ~ \max_{1 \leq j \leq n_2} \norm{M \eb_j}_2^2 \right),
\end{align*}
where the third inequality follow from \eqref{eq:samesparseAk}.


\end{document}